\documentclass[aps,prx,twocolumn,footinbib,superscriptaddress]{revtex4-2}
\usepackage{graphicx}
\usepackage{indentfirst}
\usepackage{physics}
\usepackage{braket}
\usepackage{float}
\usepackage{amsmath}
\usepackage{mathtools}
\usepackage{epstopdf}
\usepackage{footnote}
\usepackage{CJK}
\usepackage{esint}
\usepackage{comment}
\usepackage{color}
\usepackage[T1]{fontenc}
\usepackage{subfigure}
\usepackage{amsfonts}
\usepackage{footmisc}
\usepackage{scrextend}
\usepackage{multirow}
\usepackage[hyperfootnotes=false]{hyperref}
\usepackage[acronym]{glossaries}
\usepackage[english]{babel}
\usepackage{url}
\usepackage{bm}
\usepackage{hyperref}
\usepackage{algpseudocode}
\usepackage{svg}
\usepackage{diagbox}
\definecolor{darkblue}{rgb}{0,0,0.5}
\hypersetup{
colorlinks=true,
linkcolor=black,
filecolor=blue,
citecolor=darkblue,  
urlcolor=black,
}

\newtheorem{theorem}{Theorem}

\newtheorem{corollary}[theorem]{Corollary}

\newtheorem{definition}[theorem]{Definition}

\newtheorem{lemma}[theorem]{Lemma}

\newtheorem{proposition}[theorem]{Proposition}

\newenvironment{proof}[1][Proof]{\noindent\textbf{#1.} }{\ \rule{0.5em}{0.5em}}

\newcommand{\calC}{{\cal C}}

\newcommand{\calN}{{\cal N}} 

\newcommand{\calG}{{\cal G}}

\newcommand{\calU}{{\cal U}}

\newcommand{\1}{^{(1)}}

\usepackage{mathtools} 

\newcommand{\bS}{\boldsymbol S}
\newcommand{\bV}{\boldsymbol V}
\newcommand{\bI}{\boldsymbol I}
\newcommand{\bK}{\boldsymbol K}

\def\be{\begin{equation}}
\def\ee{\end{equation}}
\def\ba{\begin{eqnarray}}
\def\ea{\end{eqnarray}}

\begin{document}

\title{Energy-dependent barren plateau in bosonic variational quantum circuits}

\author{Bingzhi Zhang}
\affiliation{
Ming Hsieh Department of Electrical and Computer Engineering, University of Southern California, Los
Angeles, California 90089, USA
}

\author{Quntao Zhuang}
\email{qzhuang@usc.edu}
\affiliation{
Ming Hsieh Department of Electrical and Computer Engineering, University of Southern California, Los
Angeles, California 90089, USA
}
\affiliation{ Department of Physics and Astronomy, University of Southern California, Los
Angeles, California 90089, USA
}

\begin{abstract}
Bosonic continuous-variable Variational quantum circuits (VQCs) are crucial for information processing in cavity quantum electrodynamics and optical systems, widely applicable in quantum communication, sensing and error correction. The trainability of such VQCs is less understood, hindered by the lack of theoretical tools such as $t$-design due to the infinite dimension of the physical systems involved. We overcome this difficulty to reveal an energy-dependent barren plateau in such VQCs. The variance of the gradient decays as $1/E^{M\nu}$, exponential in the number of modes $M$ but polynomial in the (per-mode) circuit energy $E$. The exponent $\nu=1$ for shallow circuits and $\nu=2$ for deep circuits. We prove these results for state preparation of general Gaussian states and number states. We also provide numerical evidence that the results extend to general state preparation tasks. As circuit energy is a controllable parameter, we provide a strategy to mitigate the barren plateau in continuous-variable VQCs.

\end{abstract}
\maketitle

\section{Introduction}

Variational quantum circuits (VQCs)~\cite{cerezo2021variational} are candidates for achieving practical quantum advantages in the noisy intermediate-scale quantum (NISQ) era~\cite{Preskill2018quantumcomputingin}, when scalable error-corrected quantum computers are not yet available. VQCs utilize classical control to optimize a quantum circuit to solve computation problems, including optimization~\cite{farhi2014quantum}, eigen-system problem~\cite{peruzzo2014variational,kandala2017hardware,mcclean2016theory,o2016scalable,colless2018computation,bravo2020scaling, wiersema2020exploring}, partial-differential equations~\cite{lubasch2020variational}, quantum simulation~\cite{li2017efficient,dumitrescu2018cloud,mcardle2019variational} and machine learning~\cite{schuld2015introduction,biamonte2017quantum,dunjko2018machine,rebentrost2018quantum,killoran2019continuous,havlivcek2019supervised,schuld2019quantum,du2020expressive,yang2021provable}. As a general approach of designing quantum circuits, it has also found applications in the approximation~\cite{benedetti2019adversarial}, preparation~\cite{wecker2015progress,chen2020demonstration}, classification~\cite{patterson2021quantum,chen2021universal,cong2019quantum,maccormack2020branching,zhang2022fast} and tomography~\cite{liu2020variational} of quantum states.


Initial works on VQCs concern discrete-variable (DV) finite-dimensional systems of qubits, which are natural for computation; while continous-variable (CV) systems of bosonic qumodes are less explored. Yet, many important quantum systems are modelled by qumodes. For example, quantum communication and networking~\cite{gisin2007quantum,kimble2008quantum,biamonte2019complex,wehner2018quantum,kozlowski2019towards} rely on photons---the only flying quantum information carrier. In this regard, quantum transduction and entanglement distillation are shown to be enhanced by CV VQCs~\cite{zhang2022hybrid}; Photonic quantum computers~\cite{baragiola2019GKPuniversal,larsen2021} are also relying on bosonic encoding such as the cat code and Gottesman-Kitaev-Preskill (GKP) code~\cite{gottesman2001encoding}, which has shown great promise~\cite{ofek2016extending,sivak2022breakeven}. The engineering of such code states are greatly boosted by CV VQCs~\cite{heeres2015,krastanov2015universal,campagne2020quantum,eickbusch2022fast}; Finally, distributed entangled sensor networks ubiquitously rely on CV VQCs to achieve quantum advantages in sensing~\cite{zhuang2018distributed,zhang2021distributed,brady2022entangled,xia2023entanglement} and data classification~\cite{zhuang2019physical,xia2021quantum}.

Different from traditional algorithms, the runtime of VQCs is characterized by the time necessary to train the variational parameters to optimize a cost function. Therefore, the landscape of the cost function determines the VQC's trainability---large gradients across the landscape will guarantee a smooth training process. For DV systems, thanks to the well-established toolbox of random unitaries and $t$-design~\cite{gross2007evenly,ambainis2007quantum,roberts2017chaos,brandao2016local}, rigorous results unveil the barren plateau phenomena~\cite{mcclean2018barren, cerezo2021cost,wang2021noise,carlos2021entanglement}, which states that cost function gradient typically vanishes exponentially with the number of qubits when the circuit depth is not too shallow.
However, the trainability of CV VQCs is unexplored, perhaps due to the fundamental problem of the non-existence of $t$-design with $t\ge 2$ in infinite dimension~\cite{blume2014curious,zhuang2019scrambling,iosue2022continuous}. 

In this work, we overcome the problem by developing the notion of energy-constrained random quantum circuits and prove a unique energy-dependent barren plateau phenomena for such circuits. 
For an $M$-mode state preparation task, the variance of the gradient asymptotically decays as $\sim 1/E^{M\nu}$, exponential in the number of modes $M$ while power-law in the circuit energy $E$. The barren-plateau exponent $\nu=1$ for shallow circuits and $\nu=2$ for deep circuits. We prove the results for state preparation tasks of general Gaussian states and number states. As the VQCs are randomly initialized, we expect the energy-dependent barren plateau to be universal for all state preparation tasks and provide supporting numerical evidence. Moreover, as the energy of the circuit is a controllable parameter upon random initialization, we provide a strategy to mitigate barren plateau by heuristically choosing the initial circuit energy in each training problem.

\begin{figure}[t]
    \centering
    \includegraphics[width=0.45\textwidth]{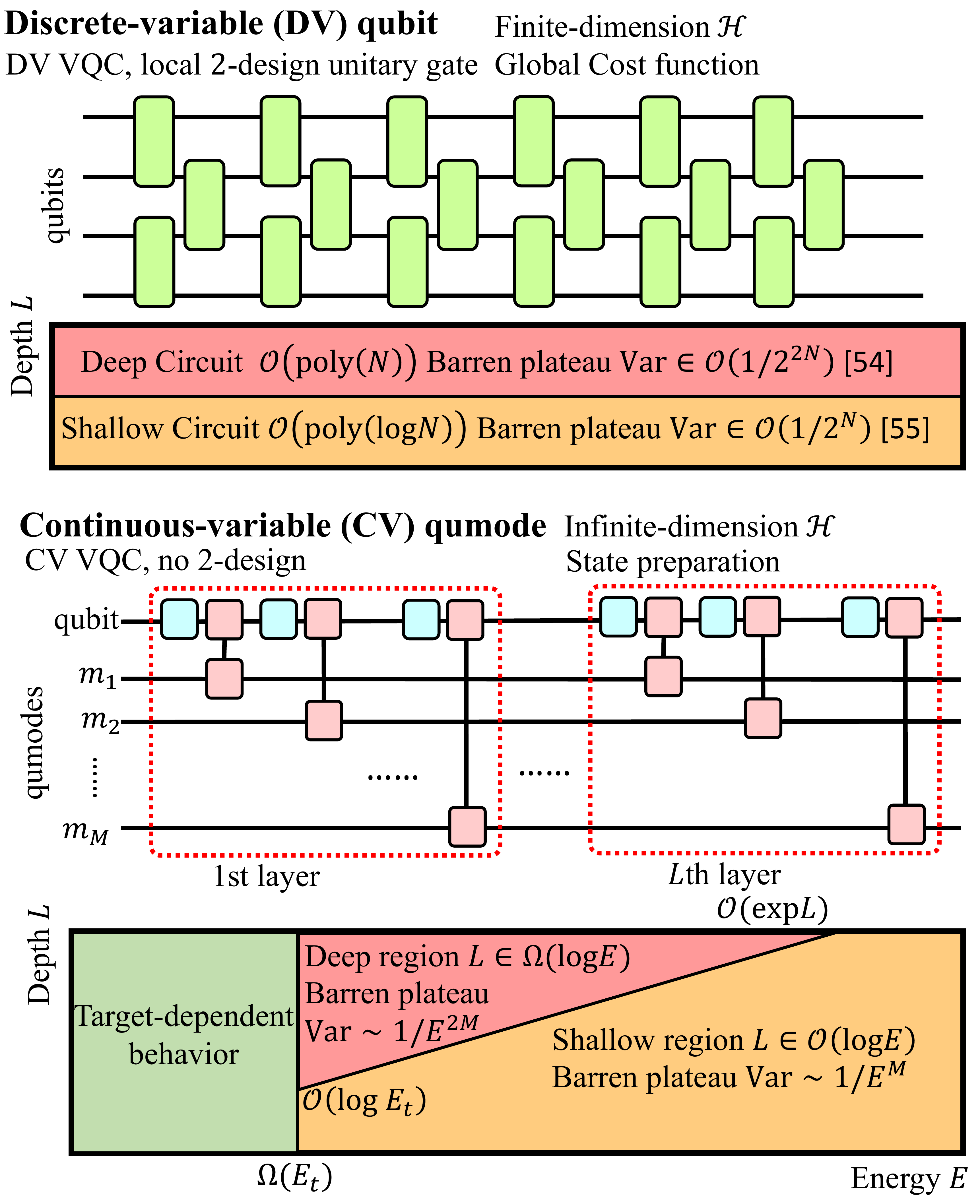}
    \caption{Summary of VQC trainability in DV and CV systems. The Hilbert space is finite-dimension in DV systems while for CV ones it is infinite-dimension. The universal DV VQC is built from local $2$-design unitary gates (lime green), and universal CV VQCs consist of single qubit rotations (cyan) and echoed conditional displacement (ECD) gates (pink)~\cite{eickbusch2022fast}. In DV VQCs, the variance of the gradient decays exponentially with the number of qubits $N$ in shallow~\cite{cerezo2021cost} and deep~\cite{mcclean2018barren} DV VQCs optimizing a global cost function. In this paper, we show the variance decays exponentially in number of modes $M$ but polynomially with circuit energy $E$ in shallow and deep region of CV VQCs.
    }
    \label{fig:scheme}
\end{figure}

\section{Results}

In general, the goal of a VQC $U(\bm x)$ is to minimize a cost function
\begin{align}
    \calC(\bm x) = \Tr\left[OU(\bm x)\rho_0 U^\dagger(\bm x)\right]
    \label{Cost_func}
\end{align}
over the tunable parameters $\bm x$,
where $\rho_0$ is the initial state and $O$ is a Hermitian observable. 
The performance of VQCs relies on a balance between expressivity and trainability. Expressivity concerns about the size of the solution space a VQC ansatz can cover and can be quantified by the cardinality of the unitary set needed to well approximate the VQC ansatz~\cite{du2022efficient}; while trainability concerns about how fast a VQC can converge to an optimal configuration within the solution space. 

The study of DV VQCs has revealed a trade-off between expressivity and trainability~\cite{larocca2022diagnosing,zhang2022quantum}. To ensure expressivity, DV VQCs consist of single-qubit rotations and two-qubit or multi-qubit gates (see Fig.~\ref{fig:scheme}), for example in the hardware-efficient ansatz~\cite{kandala2017hardware}. When the number of layers is linear in the number of qubits $N$, a VQC can approximate complex unitary ensembles classified as $t$-design~\cite{gross2007evenly,ambainis2007quantum,roberts2017chaos,brandao2016local}. However, increasing the expressibility makes the training of VQCs more challenging. Upon random initialization, the gradient of cost function $\calC$ with respect to any parameter is exponentially small in the number of qubits. More precisely, Refs.~\cite{mcclean2018barren, cerezo2021cost} show that while the mean is zero, the variance of the gradient decays as $\mathcal{O}(1/2^N)$ for shallow circuits and as $\mathcal{O}(1/2^{2N})$ for deep circuits (summarized in Fig.~\ref{fig:scheme} top panel). 


For the CV case, to ensure expressivity over $M$ oscillators, we consider the approach of echoed conditional displacement (ECD) gates via weak dispersive coupling to a qubit~\cite{eickbusch2022fast} (see Fig.~\ref{fig:scheme}). 
Each layer of the VQC consists of $M$ ECD gates in between single-qubit rotations: the qubit controls the displacement of each qumode via an ECD gate $U_{\rm ECD}(\beta) = D(\beta)\otimes \ket{1}\bra{0} + D(-\beta)\otimes\ket{0}\bra{1}$ after a single qubit rotation $U_{\rm R}(\theta, \phi) = \exp[-i\theta/2(\cos\phi\sigma^x + \sin\phi\sigma^y)]$, where $\sigma^x$ and $\sigma^y$ are Pauli-X and Y operator and $D(\beta) \equiv \exp(\beta m^\dagger-\beta^* m)$ is the displacement operator on mode $m$. Here, despite the lack of Haar random unitaries and $t$-design ($t\ge2$)~\cite{blume2014curious,zhuang2019scrambling,iosue2022continuous} due to the infinite dimension, we establish the energy-regularized circuit ensemble to represent `typical' qubit-qumode circuits to enable the analyses.

Our main result shows that when training such a CV VQC for state preparation, the cost function exhibits an energy-dependent barren plateau. Specifically, when preparing an $M$-mode state with energy per mode $E_t$, the variance of the gradient decays polynomially with circuit energy when the circuit energy per mode $E\in\Omega(E_t)$. When the circuit has a shallow log-depth ($L \in \mathcal{O}(\log E)$), the variance of the gradient decays as $1/E^{M}$; in the deep circuit region ($L \in \Omega (\log E)$), it decays as $1/E^{2M}$. Alternatively, for a fixed circuit depth $L\gtrsim \log(E_t)$, as the circuit energy increases, the variance of the gradient first displays target-dependent behaviors before a quick $1/E^{2M}$ decay, followed by a transition to the $1/E^M$ scaling at the critical energy $E\sim \exp(L)$. We prove these results for the state preparation of a fairly general class of states known as Gaussian states~\cite{weedbrook2012gaussian}, including non-classical multipartite entangled states useful in quantum computing~\cite{gu2009quantum} and distributed quantum sensing~\cite{zhuang2018distributed,zhang2021distributed,zhuang2019physical,xia2021quantum}. To extend the results beyond Gaussian states, we also prove them for Fock number states and verify them numerically for general random states. 

Furthermore, as circuit energy $E$ is a continuously tunable real parameter, the energy-dependent barren plateau provides an opportunity of mitigating the challenges in training.

\subsection{Energy-regularized circuit ensemble}

To ensure the expressivity of the CV VQCs, we adopt a universal gate set over a qubit and $M$ qumodes. While it is known that ECD gates with single-qubit unitaries can achieve universal control on one qubit and one mode~\cite{eickbusch2022fast}, we extend the proof to the system of multiple qubits and qumodes (see Appendix~\ref{app:universal_control}). We state the following lemma:
\begin{lemma}
\label{lemma_universality}
Universal control over $M\ge1$ qumodes and $N\ge 1$ qubits can be realized by all ECD gates between any qumode and any qubit and all single qubit rotations. 
\end{lemma}
Alternatively, universal control can also be achieved by a variant of the ECD gate---the Conditional Not Displacement gate~\cite{diringer2023conditional}. The results of our paper still hold in that case.

Although the ECD gates can have arbitrarily large complex amplitudes, as realistic systems are always subject to energy constraints, we introduce the following definition to model a `typical' circuit layout for the CV VQCs:
\begin{definition}
\label{definition:Dlayer_circuit}
The energy-regularized $L$-layer qubit-qumode circuits $\calU_{E,L,M}$ is the ensemble of unitaries 
\be 
U=\prod_{\ell=1}^L\prod_{j=1}^M U_{\rm ECD}^{(j)}\left(\beta_\ell^{(j)}\right)U_{\rm R}\left(\theta_{M(\ell-1)+j},\phi_{M(\ell-1)+j}\right), 
\label{eq:U_M}
\ee 
where each ECD gate's complex displacement $\beta_\ell^{(j)} \sim \calN_{E/L}^{\rm C}$ is Gaussian distributed with zero mean and variance $E/L$, and all qubit rotation angles $\{\theta_{k}, \phi_{k} \}_{k=1}^{ML} \sim {\bf U}[0,2\pi)$ are uniformly distributed.
\label{unitary_ensemble}
\end{definition}
For simplicity, we denote a zero-mean complex Gaussian distribution with variance $\sigma^2/2$ on both real and imaginary parts as $\calN_{\sigma^2}^{\rm C}$. 
It is worth noting that the choice of a Gaussian distribution for $\beta_\ell^{(j)}$ is convenient but not essential due to the central limit theorem. Although the ensemble of circuits comes from physical regularizations, we will see later that it also enables analytical solutions to various properties of the VQC and provides an excellent playground of CV quantum information processing. 

The expressivity of the VQCs in $\calU_{E,L,M}$ increases with the depth $L$. Due to the universality in Lemma~\ref{lemma_universality}, the VQCs in $\calU_{E,L,M}$ with $L\gg1$ contain all unitaries relevant to the energy scale $E$, and a larger circuit energy $E$ will enable expressibility over larger Hilbert space volumes. At the same time, the trainability of VQCs in $\calU_{E,L,M}$ is unclear, which we aim to resolve in this work. As we focus on state preparation tasks, below we consider the states generated by the random VQCs applied on trivial initial states. This is analogous to the relationship between unitary design and state design, both of which unfortunately do not exist in the CV case.


We define the energy-regularized state ensemble $\Psi_{E,L,M}$ as $\ket{\psi(\bm x)}_{q,\bm m} = U(\bm x)\ket{0}_q\otimes_{j=1}^M\ket{0}_{m_j}$, where we randomly apply $U(\bm x)\in \calU_{E,L,M}$ on spin-up qubit and vacuum qumodes. Here we denote the overall parameters as $\bm x$, including all $\beta_\ell^{(j)}$, $\theta_k$ and $\phi_k$. In our notation, $\ket{a}_q$ denotes the $a=0,1$ state of the qubit `$q$' and $\ket{\alpha}_{m_j}$ denotes the coherent state with complex amplitude $\alpha$ of the mode $m_j$ (see Appendix~\ref{app:Gaussian_intro}) and $\bm m=(m_1,\cdots, m_M)$ denote all modes.

With random qubit rotations $U_{\rm R}$'s in Eq.~\eqref{eq:U_M}, the consequent ECD gate $U_{\rm ECD}^{(j)}(\beta_\ell^{(j)})$ applied on each mode will lead to a random superposition of performing complex displacement $+\beta_\ell^{(j)}$ and $-\beta_\ell^{(j)}$ for all modes $j=1,\cdots,M$. Therefore, each layer of ECD gates with random qubit rotations corresponds to a superposition of all possible choices of a random-walk step in the $2M$-dimensional phase space. The accumulation of the displacements leads to the final output state in a superposition of random-walk trajectories. For example, in the one-mode case ($M=1$), we have the superposition
\be 
\ket{\psi(\bm x)}_{q,m}= \sum_{a=0}^1\sum_{\bm s}v_{\bm s,a}(\bm x)\ket{a}_q\ket{(-1)^a \bm s \cdot \bm \beta}_m,
\label{eq:state_ensemble}
\ee 
where sign vector $\bm s$ sums over $\{-1,+1\}^{L}$ under the constraint that $\bm s_L=-1$ and $\bm \beta=\{\beta_1,\cdots, \beta_L\}$ are the amplitudes of displacemnts.
The coefficients $v_{\bm s,a}(\bm x)$ are lengthy, we specify them and present the detailed proof of the state representation in Appendix~\ref{app:state_repre}.

One can directly check that in the energy regularized ensemble, the displacement $\bm s \cdot \bm \beta$ with arbitrary choice of sign vectors $\bm s$ obeys a Gaussian distribution and the ensemble average energy of the states in $\Psi_{E,L,1}$ is $\mathbb{E}\expval{ m^\dag m}=E$ due to energy regularization. Moreover, because of the accumulation of displacements in the summation $\bm s \cdot \bm \beta=\sum_\ell (\pm)\beta_\ell$, central limit theorem dictates that a Gaussian distribution is universal in the amplitudes of the final displacement, regardless of the distribution of each $\beta_\ell$.  

The $M$-mode case is a direct generalization of Eq.~\eqref{eq:state_ensemble}, via extending the coherent state to product of coherent states in the superposition (see Appendix~\ref{app:multi-mode}). At the same time, the corresponding energy per mode is still $E$.

\subsection{Barren plateau in state preparation}

The trainability of the circuit $U(\bm x)$ initialized in the energy-regularized ensemble $\calU_{E,L,M}$ is characterized by the typical gradient of cost function $\calC$ in Eq.~\eqref{Cost_func}. Here, we focus on the state engineering for the qumodes, while the qubit acts as an ancilla. To prepare a general $M$-mode state $\ket{\psi}_{\bm m}$, the observable can be chosen as $O = \ketbra{0}{0}_q \otimes\ketbra{\psi}{\psi}_{\bm m}$ and we can set the initial state as vacuum without loss of generality.

Consider the gradient with respect to the $k$th qubit rotation angle $\theta_k$, one can check that the parameter shift rule~\cite{mitarai2018quantum} holds for qubit rotation angles in the CV circuit, therefore
$
    \partial_{\theta_k}\calC = \left(\braket{O}_{k^{(+)}} -\braket{O}_{k^{(-)}}\right)/2,
$
where $\braket{O}_{k^{(\pm)}}$ corresponds to expectation of $O$ for the state $\rho_{k^{(\pm)}}$ prepared under circuit $U(\bm x)$ with the rotation angle $\theta_k$ shifted by $\pm \pi/2$. In the same spirit of Refs.~\cite{mcclean2018barren,cerezo2021cost}, we consider random initial conditions---when $U(\bm x)$ is randomly sampled from the energy-regularized ensemble $\calU_{E,L,M}$, the mean of the gradient
\begin{equation}
    \mathbb{E}\left[\partial_{\theta_k}\calC\right] = \frac{1}{2}\left(\mathbb{E}\left[\braket{O}_{k^{(+)}}\right] - \mathbb{E}\left[\braket{O}_{k^{(-)}}\right]\right) = 0,
\end{equation}
due to the fact that the ensemble average over parameters with $\theta_k\pm \pi/2$ are equal. In this case, the typical values of gradients are characterized by the variance
\begin{equation}
\begin{split}
    {\rm Var}\left[\partial_{\theta_k}\calC\right]
    = \frac{1}{2}\mathbb{E}\left[\braket{O}_{k^{(+)}}^2\right] -\frac{1}{2}\mathbb{E}\Big[\braket{O}_{k^{(+)}}\braket{O}_{k^{(-)}}\Big].
    \label{eq:var}
\end{split}
\end{equation}

\begin{figure*}
    \centering
    \includegraphics[width=1\textwidth]{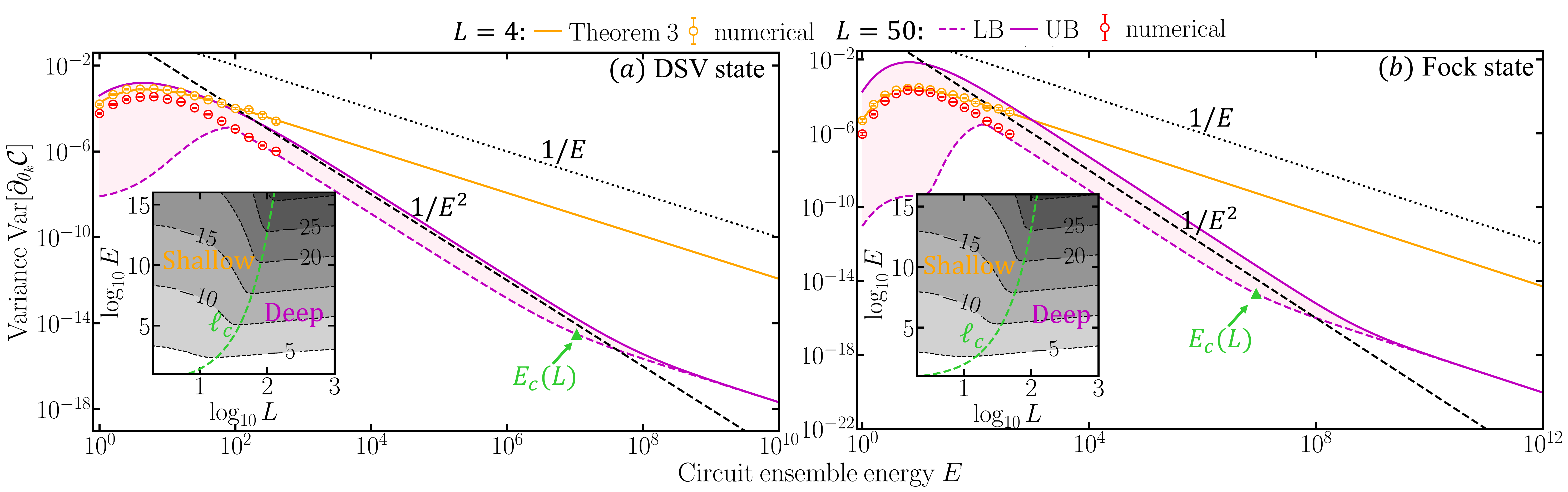}
    \caption{Variance of gradient ${\rm Var}[\partial_{\theta_{k}}\calC]$ at $k=L/2$ in preparation of (a) displaced squeezed vacuum (DSV) state with $\gamma = 2, \zeta=\sinh^{-1}(2)$ and (b) Fock state with $E_t=8$. Orange and red dots with errorbars show numerical results of variance in shallow and deep circuits. Orange solid curve represents the analytical variance in Theorem~\ref{res:shallow}; the dashed and solid magenta curves show the lower and upper bounds in Ineqs.~\eqref{eq:LB_main},~\eqref{eq:UB_main}. Black dotted and dashed lines indicate the scaling of $1/E$ and $1/E^2$. Insets in (a)(b) we plot the logarithm in base ten of the upper bound in Ineq.~\eqref{eq:UB_main} versus circuit depth and energy. Green triangle (main) and line (inset) show the corresponding boundary of variance at $E_c(L)$ and $\ell_c(E)$.
    }
    \label{fig:var_states}
\end{figure*}

The exact evaluation of ${\rm Var}\left[\partial_{\theta_k}\calC\right]$ is in general challenging, due to the lack of theoretical tools such as $t$-design in infinite dimension. Instead, our proposed energy-regularized ensemble allows us to solve a pair of asymptotic lower and upper bounds when $E\gg 1$ as (ignoring higher-order terms, see Appendix~\ref{app:multi-mode})
\begin{align}
    &{\rm Var}\left[\partial_{\theta_k}\calC\right]\ge \frac{1}{2}\left[\frac{3^{ML-1}}{4^{ML}} C_1+\left(\frac{1}{4}-\frac{3^{ML}}{4^{ML}}\right)\min_{\bm \ell}C_2\left(\frac{\bm \ell}{L}\right)\right], \label{eq:LB_main}\\
    &{\rm Var}\left[\partial_{\theta_k}\calC\right] \le \frac{1}{2}\left[\frac{3^{ML-1}}{4^{ML}} C_1+\left(\frac{1}{4}+\frac{2^{ML-1}}{4^{ML}}\right)\max_{\bm \ell} C_2\left(\frac{\bm \ell}{L}\right)\right] \label{eq:UB_main},
\end{align}
where the minimization and maximization are over the vector $\bm \ell=(\ell^{(1)},\dots, \ell^{(M)})^T$ with each integer element $\ell^{(j)}\in [1, L-1]\cap\mathbb{N}$. The state-dependent correlators are
\begin{align}
C_1&=\mathbb{E}_{\bm \alpha }\left[\left\lvert{} {}_{\bm m}\langle \psi| \bm \alpha\rangle_{\bm m}\right\rvert^4\right],
\label{eq:C1_main}\\ 
C_2(\bm z)&=\mathbb{E}_{\bm \alpha_{\bm z}, \bm \alpha_{\bm 1-\bm z} }\left[\prod_{h=0}^1\left\lvert{}_{\bm m}\langle \psi|\bm \alpha_{\bm z}+(-1)^h \bm \alpha_{\bm 1-\bm z}\rangle_{\bm m} \right\rvert^2\right].\label{eq:C2_main}
\end{align}
Here we have defined the vector notation $\ket{\bm \alpha}_{\bm m}=\bigotimes_{j=1}^M \ket{\alpha^{(j)}}_{m_j}$, with each $\ket{\alpha^{(j)}}_{m_j}$ being a coherent state with displacement $\alpha^{(j)}$ for mode $m_j$. The ensemble average in Eq.~\eqref{eq:C1_main} is over each component $\alpha^{(j)}$ sampling from $\calN_{E}^{\rm C}$, due to Definition~\ref{definition:Dlayer_circuit}. Similarly, $\ket{\bm \alpha_{\bm z}+(-1)^h \bm \alpha_{\bm 1-\bm z}}_{\bm m}=\bigotimes_{j=1}^M\ket{\alpha_{z_j}+(-1)^h\alpha_{1-z_j}}_{m_j}$, with each component $\alpha_{z_j}\sim \calN_{z_j E}^{\rm C}$ and $\alpha_{1-z_j}\sim \calN_{(1-z_j) E}^{\rm C}$. As the above two correlators are only functions of state fidelity, numerical evaluation is often efficient and analytical evaluation is sometimes possible. Note that for $M\ge 2$, Eq.~\eqref{eq:C2_main} has already ignored terms exponentially small in $L$, which will not affect our results (see Appendix~\ref{app:multi-mode}).
In the following, we present analytical results for general Gaussian states and number states, and numerical results for randomly generated states in the lower and upper bounds. Although Gaussian states can be efficiently prepared with a Gaussian circuit of linear optics and squeezers~\cite{weedbrook2012gaussian}, as the random initialized VQCs are not leveraging any heuristics, we expect trainability found there to be universal, which is supported by our other results. We begin with one-mode state preparation and then generalize to the multi-mode case.

\subsubsection{One-mode state preparation}
\label{sec:one_mode_prepare}
For one-mode state preparation, we consider Gaussian states, number states and random states.
An one-mode (pure) Gaussian state can be represented by a displaced squeezed vacuum state $D(\gamma)S(\zeta)\ket{0}_m$, where $S(\zeta) = \exp[\zeta(m^2-m^{\dagger 2})]$ is the squeezing operator (see Appendix~\ref{app:Gaussian_intro}). This is an important class of states, as coherent states provide a quantum model for lasers and squeezed vacuum is a key resource of quantum sensing, e.g. Laser Interferometer Gravitational-Wave Observatory~\cite{tse2019quantum,aasi2013enhanced,abadie2011gravitational} and dark matter search~\cite{backes2021}. Here we omit any possible phase rotation, which is treated in Appendix~\ref{app:var_gaussian}. The energy for such a state is $E_t=|\gamma|^2 + \sinh^2(\zeta)$. We can analytically solve the two correlators for one-mode Gaussian state as
\begin{align}
    C_1^{\rm Gauss} &= \frac{\sech^2(\zeta)e^{-R(E)/G_1(E)}}{\sqrt{G_1(E)}}, \label{eq:C1_gauss_main}\\
    C_2^{\rm Gauss}(z) &= \frac{\sech^2(\zeta)e^{-R(zE)/G_1(zE)}}{\sqrt{G_1(E-zE)G_1(zE)}},
    \label{eq:C2_gauss_main}
\end{align}
where we define 
$G_1(x) = 1+4x+4\sech^2(\zeta)x$ and $R(x) = 2|\gamma|^2+2\tanh(\zeta)(\Re{\gamma}^2-\Im{\gamma}^2)+4\sech^2(\zeta)|\gamma|^2x$. 

To go beyond Gaussian states, we consider the Fock number state with $E_t\in \mathbb{N}$ photons, whose correlators can be solved in closed-form as
\begin{align}
    C_1^{\rm Fock} &= \frac{(2E_t)!}{(2^{E_t}E_t!)^2}\frac{(1+1/2E)^{-2E_t}}{1+2E}, \label{eq:C1_fock_main}\\
    C_2^{\rm Fock}(z) &= \eta \frac{(1+2E_t)(2E_t)!}{(2^{E_t}E_t!)^2}\frac{(1+1/2Ez)^{-2E_t}}{(1+2Ez)[1+2E(1-z)]}. \label{eq:C2_fock_main}
\end{align}
In the second line, the right-hand-side represents the lower and upper bounds of $C_2^{\rm Fock}(z)$ as exact evaluation becomes hard: for the upper bound $\eta=1$ and for the lower bound $\eta={}_2F_1(1/2,-E_t,1,1)$ where ${}_2F_1$ is the original hypergeometric function. 
By having those correlators in Ineqs.~\eqref{eq:LB_main},~\eqref{eq:UB_main}, we can have the corresponding bounds for variance of gradient in preparation of coherent states and Fock states. 

With the above correlators in hand, we found that when the circuit depth $L$ is shallow, the lower and upper bounds in Ineqs.~\eqref{eq:LB_main} and~\eqref{eq:UB_main} coincide to the leading order of $\sim1/E$ for large $E$. When depth $L$ is large, the $\sim 1/E^2$ terms will dominate in both lower and upper bounds. Quantitatively, the cross-over of the scaling happens at depth
$
\ell_c (E) \in \mathcal{O}(\log E)
$.
For the full formula of $\ell_c (E)$, please refer to Appendix~\ref{app:variance}.
Equivalently, for circuits with a fixed depth $L$, the transition of scaling from $1/E^2$ to $1/E$ takes place at $E_c(L)\in \Omega(\exp(L))$. Overall, we have the following theorems:
\begin{theorem}
(Barren plateau for shallow depth.) For a single-mode ($M=1$) CV VQC randomly initialized from the energy-regularized ensemble $\calU_{E,L,1}$ with a shallow depth 
$L \le \ell_c(E) \in \mathcal{O}(\log E)$, the gradient with respect to qubit rotation angles when preparing a Gaussian or a Fock state has zero-mean and variance 
\begin{align}
       {\rm Var}\left[\partial_{\theta_k}\calC\right] &= \frac{1}{6}\left(\frac{3}{4}\right)^L C_1 + \mathcal{O}\left(\frac{1}{E^2}\right).
       \label{eq:var_shallow}
\end{align}
where $C_1$ is correlator in Eq.~\eqref{eq:C1_gauss_main} or~\eqref{eq:C1_fock_main}. In particular, ${\rm Var}\left[\partial_{\theta_k}\calC\right]\sim 1/E$ in the large $E$ limit.
\label{res:shallow}
\end{theorem}
\begin{theorem}
(Barren plateau for deep depth.)
For a single mode ($M=1$) CV VQC randomly initialized from the energy-regularized enmseble $\calU_{E,L,1}$, with layers $L \ge  \ell_c(E) \in \Omega (\log E)$, the gradient with respect to qubit rotation angles when preparing a Gaussian state or a Fock state has zero-mean and asymptotic variance
$
{\rm Var}\sim 1/{E^2}.
$
\label{res:deep}
\end{theorem}
In terms of the asymptotic region, in practice we find the scaling $1/E^2$ to hold as long as $E\in \Omega(E_t)$.

\begin{figure}[t]
    \centering
    \includegraphics[width=0.45\textwidth]{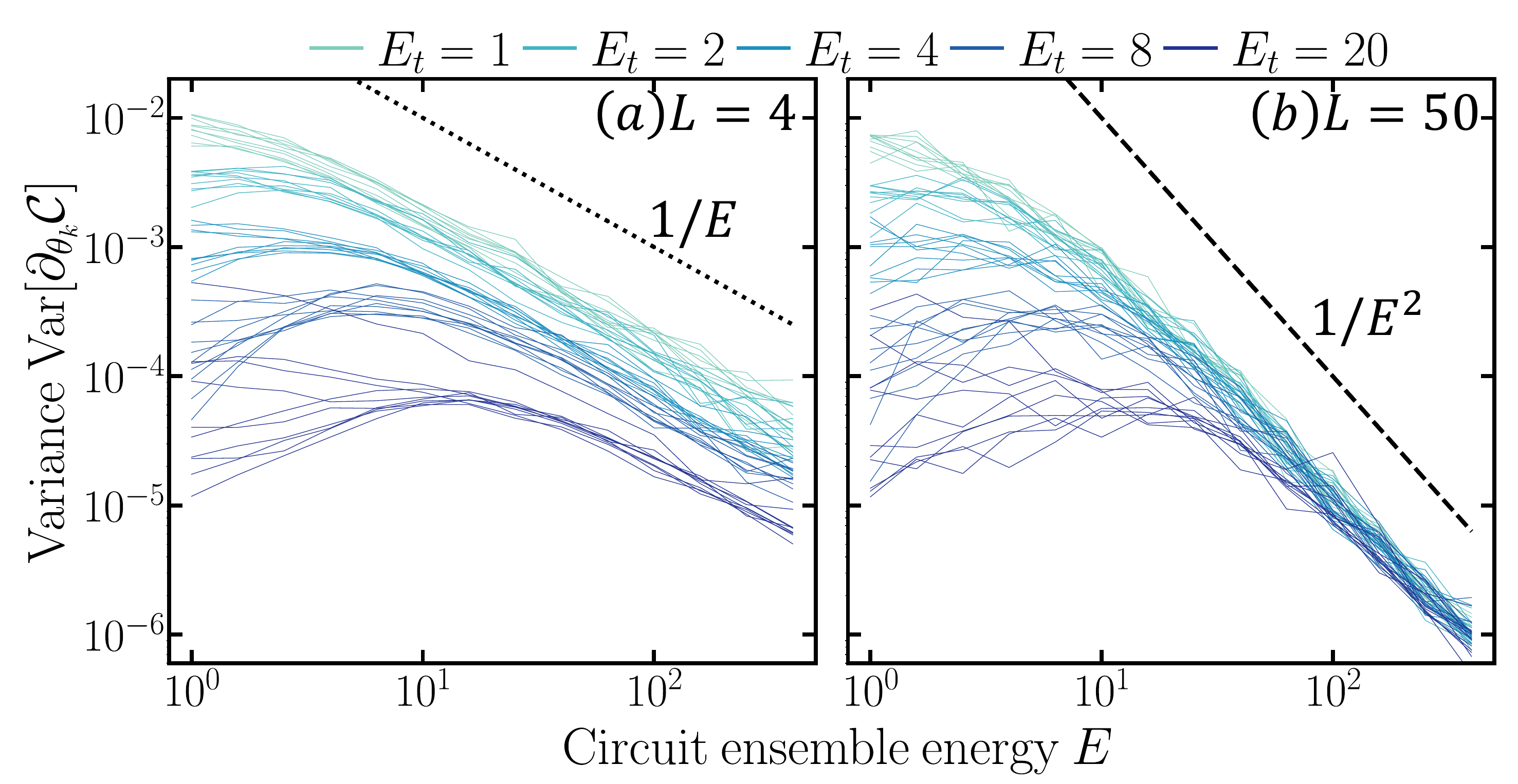}
    \caption{Variance of gradient ${\rm Var}[\partial_{\theta_{k}}\calC]$ at $k=L/2$ in preparation of random CV states $\ket{\psi}_m=\sum_{n} b_n \ket{n}_m$ with $L=4$ (a) and $L=50$ (b) circuits. Curves with same color show the variance of different sample target states. Black dotted and dashed lines in (a) and (b) represent the scaling of $1/E$ and $1/E^2$. In our calculation, we have chosen cutoff $n_c\sim 2E_t$ and $\epsilon=0.1$. 
    }
    \label{fig:var_rand}
\end{figure}

In Fig.~\ref{fig:var_states}(a),(b), we show the variance of the gradient ${\rm Var}\left[\partial_{\theta_k}\calC\right]$ versus ensemble energy $E$ in shallow and deep circuits and compare numerical results to the bounds and theorems. We consider the preparation of a displaced squeezed vacuum (DSV) state with $\gamma=2,\zeta=\sinh^{-1}(2)$ and Fock states with $E_t=8$ separately. We see that the numerical variance in shallow circuits $L=4$ (orange dots) agrees well with Eq.~\eqref{eq:var_shallow} stated in Theorem~\ref{res:shallow} (orange line), which suggests a scaling of $1/E$ in the asymptotic region of $E$. For deep circuits $L=50$, the numerical results (red) lie between the lower bound and upper bound, and indeed obey the $1/E^2$ scaling, which supports Theorem~\ref{res:deep}. To our surprise, the lower bound in Ineq.~\eqref{eq:LB_main} becomes extremely tight in asymptotic region of $E$ in both cases. At the same time, despite the large circuit depth $L=50$, given extremely high ensemble energy above $E_c(L)$, the VQCs are again shallow compared with $\ell_c(E)$ and the variance of the gradient obeys the $1/E$ scaling for shallow circuits. To understand the transition between shallow and deep circuits, in the inset of Fig.~\ref{fig:var_states} we present the contours of upper bound in Ineq.~\eqref{eq:UB_main} versus circuit depth $L$ and energy $E$. Here we identify a clear contrast between shallow and deep circuits in terms of circuit depth and ensemble energy, where the boundary representing $\ell_c\in \mathcal{O}(\log E)$ is indicated by the green curve. 
Besides DSV states, we also consider special cases of coherent and single mode squeezed vacuum (SMSV) states in Appendix~\ref{app:one_mode_gaussian}, where our bounds and theorems are again verified.

For general state preparation, the evaluation of the correlators $C_1, C_2$ in Ineqs.~\eqref{eq:LB_main} and~\eqref{eq:UB_main} can be hard. However, informed from the above results, we conjecture that Theorems~\ref{res:shallow},~\ref{res:deep} hold for arbitrary single-mode state preparation. To support this conjecture, we present numerical evidence for the preparation of randomly generated CV states. These states are random superposition of the number bases in the form of $\ket{\psi}_m\propto \sum_{n=0}^{n_c} b_n \ket{n}_m$, where each $b_n\sim \calN_{2}^{\rm C}$ is randomly chosen. We post-select states within the energy window $[E_t- \epsilon,E_t+\epsilon]$. With the target states generated, we evaluate the variances of gradient in preparing each fixed state versus circuit energy for different choices of $E_t$ (indicated by the color) in Fig.~\ref{fig:var_rand}. Despite state-dependent behaviors in the low energy part, a universal decay of the gradients with the energy can be identified in the $E\gtrsim E_t$ region. The decay shows a scaling of $\sim 1/E$ in shallow circuits (subplot (a)) and a scaling of $\sim 1/E^2$ in deep circuits (subplot (b)).

\subsubsection{Multi-mode state preparation}

Now we generalize our results of energy-dependent barren plateau to the multi-mode case, including analytical results on general Gaussian states, number states and numerical results for random states. We begin our discussion with the simple case of a product state $\ket{\psi}_{\bm m} = \otimes_{j=1}^M \ket{\psi_j}_{m_j}$, whose correlators in Eqs.~\eqref{eq:C1_main},~\eqref{eq:C2_main} have the form
\begin{align}
    C_1^{\rm Prod} &= \prod_{j=1}^M \left(\mathbb{E}_{\alpha^{(j)}\sim \calN_E^{\rm C}}\left[\left\lvert\braket{\psi_j|\alpha^{(j)}}_{m_j}\right\rvert ^4\right]\right), \label{eq:C1_prod_main}\\
    C_2^{\rm Prod}(\bm z) &= \prod_{j=1}^M \mathbb{E}_{\alpha_y\sim \calN_{yE}^{\rm C}}\left[\prod_{h=0}^1 \left\lvert\braket{\psi_j|\alpha_{z_j}+(-1)^h\alpha_{1-z_j}}_{m_j}\right\rvert^2\right], \label{eq:C2_prod_main}
\end{align}
which reduce to a product of single-mode correlators. Consequently, Theorems~\ref{res:shallow} and~\ref{res:deep} directly generalize to state preparation of products of single-mode Gaussian states and products of number states: the scaling of the variance of the gradient is $1/E^M$ for shallow circuits $L\in\mathcal{O}(\log E)$ and $1/E^{2M}$ for deep circuits $L\in \Omega(\log E)$ (see Appendix~\ref{app:var_prod} for a detailed proof). 

\begin{figure}[t]
    \centering
    \includegraphics[width=0.5\textwidth]{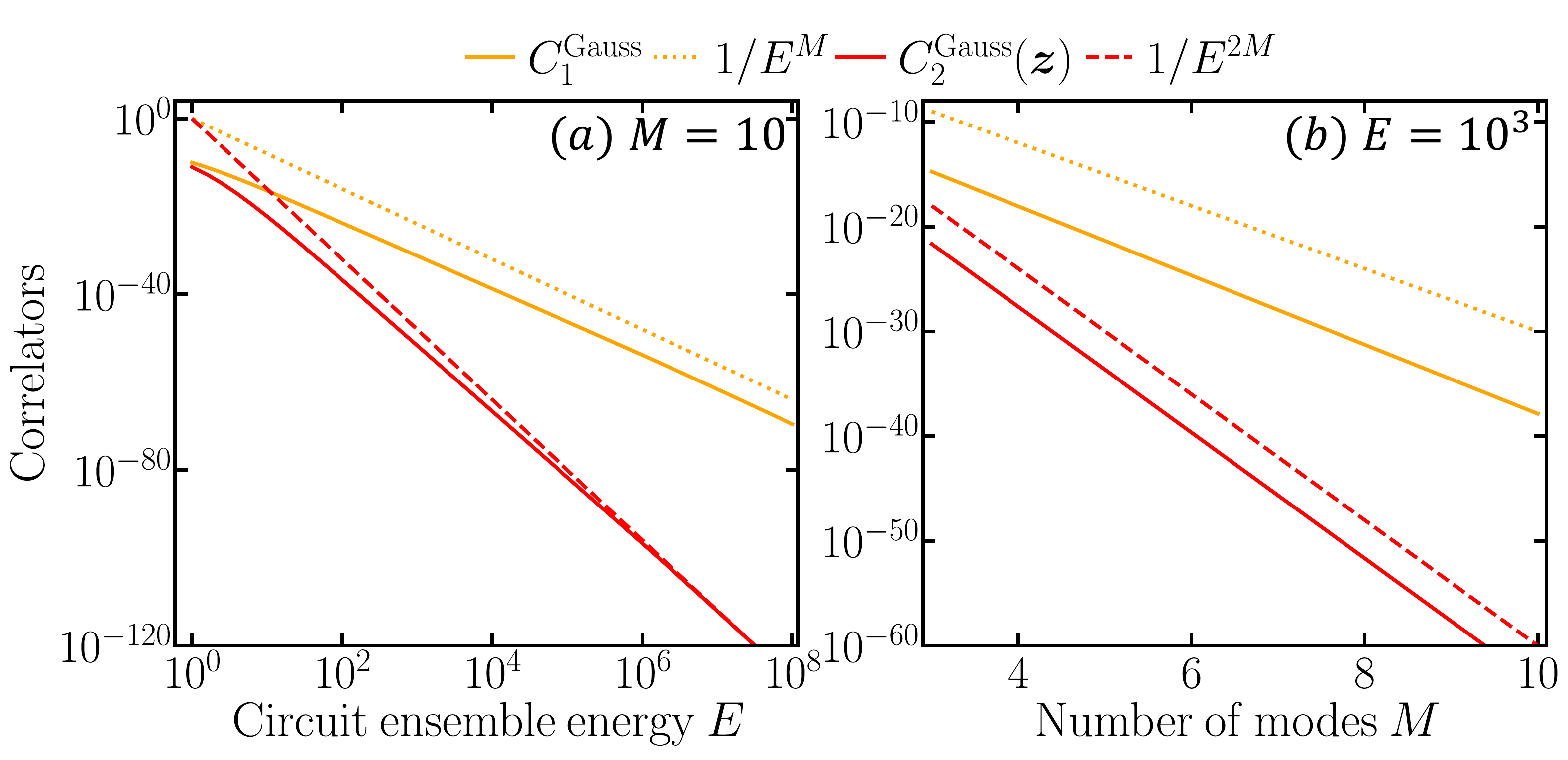}
    \caption{Correlators $C_1^{\rm Gauss}$ and $C_2^{\rm Gauss}(\bm z)$ with $\bm z=\{1/2,\cdots,1/2\}$ in Eqs.~\eqref{eq:C1_gauss_modes_main},~\eqref{eq:C2_gauss_modes_main} versus (a) ensemble energy $E$ and (b) modes $M$. The target state $\ket{\psi}_{\bm m}$ is generated by global random passive Gaussian unitary following a single-mode squeezer with strength $r=8$.}
    \label{fig:ruc}
\end{figure}

To go beyond product states, we consider an arbitrary $M$-mode Gaussian state, which is typically highly entangled~\cite{iosue2022page}. A general Gaussian state $\ket{\psi}_{\bm m}$ can be described by its mean and covariance matrix $\bV_{\bm m}$ of its Wigner function (see Appendix~\ref{app:Gaussian_intro}). For simplicity, we show the zero-mean results in the main text, and the non-zero mean case is presented in Appendix~\ref{app:var_Gaussian_modes}. The correlators $C_1$ and $C_2$ can be analytically solved as (see Appendix~\ref{app:var_Gaussian_modes})
\begin{align}
    C_1^{\rm Gauss} &= \frac{4^{M}\det(\bK)}{\sqrt{\det(4\bK+\bI/E)} E^M}\label{eq:C1_gauss_modes_main}\\
    C_2^{\rm Gauss}(\bm z) &= \frac{4^M\det(\bK)}{\sqrt{\det (4 \bK+\bS_{\bm z})\det(4 \bK+\bS_{\bm 1-\bm z}) }}\nonumber\\
    &\quad \times\frac{1}{\left[\prod_{j=1}^M z_j(1-z_j)\right]E^{2M}},
    \label{eq:C2_gauss_modes_main}
\end{align}
where $\bI$ is the $2M\times 2M$ identity matrix. Here we have defined $\bK=(\bV_{\bm m}+\bI)^{-1}$ and $\bS_{\bm z}= \oplus_{j=1}^M \bI_2/(z_jE)$ with $\bI_2$ being the $2\times 2$ identity matrix. In the asymptotic limit of $E\gg 1$, one can directly see that $C_1^{\rm Gauss}\sim 1/E^M$ while $C_2^{\rm Gauss}(\bm z)\sim 1/E^{2M}$ (see Appendix~\ref{app:multi-mode} for a proof in the non-zero mean case), which leads to the following theorem.

\begin{theorem}
\label{theorem:multi}
(Barren plateau for multi-mode Gaussian states) For an $M$-mode CV VQC randomly initialized from the energy-regularized ensemble $\calU_{E,L,M}$, the gradient with respect to qubit rotation angles when preparing an $M$-mode general Gaussian state has zero-mean and asymptotic variance ${\rm Var}\sim 1/E^M$ with a shallow depth $L \in \mathcal{O}(\log E)$, while ${\rm Var}\sim 1/E^{2M}$ with a deep depth $L\in \Omega(\log E)$. 
\end{theorem}

\begin{figure}[t]
    \centering
    \includegraphics[width=0.45\textwidth]{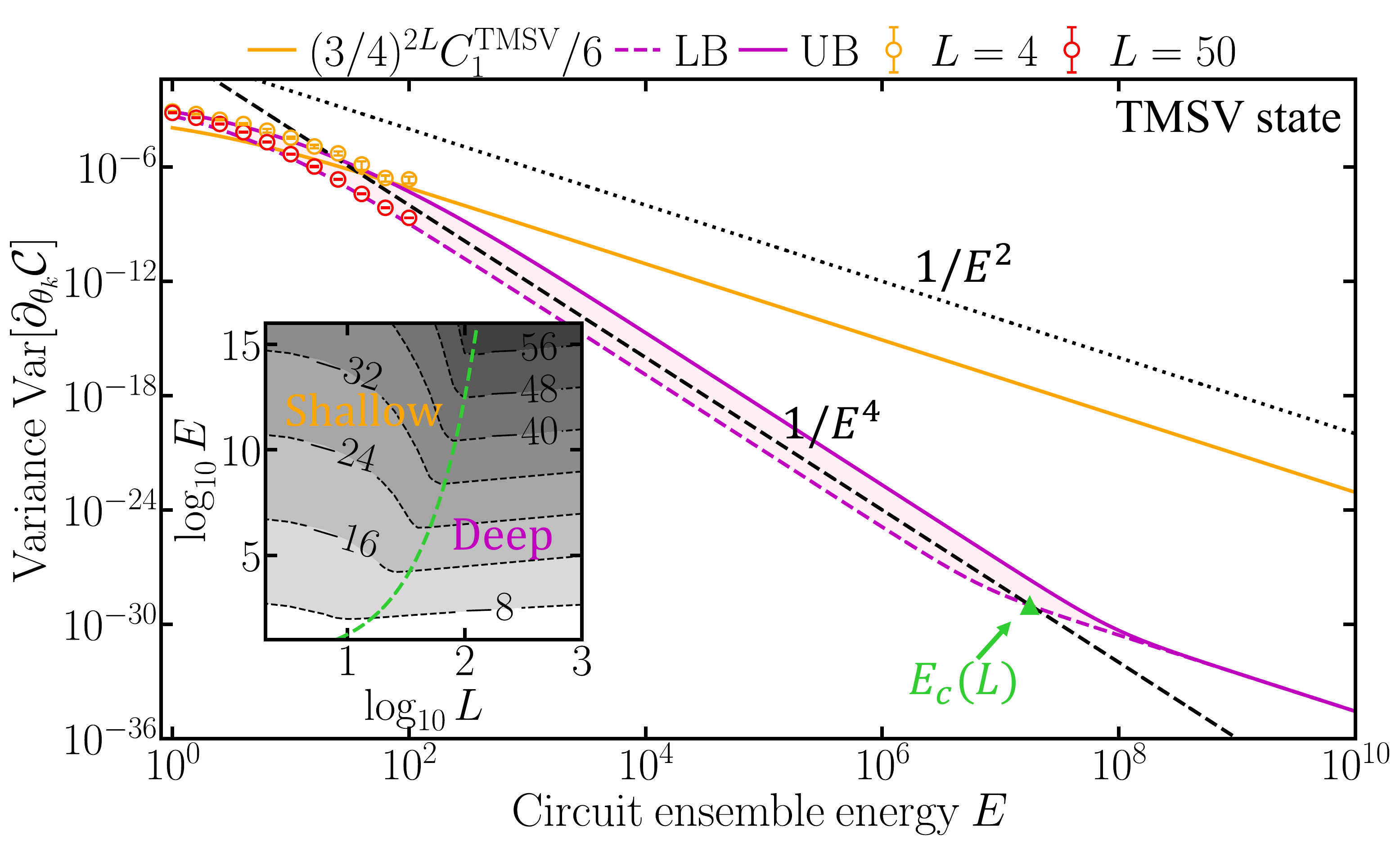}
    \caption{Variance of gradient ${\rm Var}[\partial_{\theta_{k}}\calC]$ at $k=ML/2$ in preparation of a TMSV state $\ket{\zeta}_{\rm TMSV}$ with $\zeta = \sinh^{-1}(2)$. Orange and red dots with errorbars show numerical results of variance in shallow and deep circuits. Orange solid curve represents the $(3/4)^{2L}C_1^{\rm TMSV}/6$ for reference; the dashed and solid magenta curves show the lower and upper bounds in Ineqs.~\eqref{eq:LB_main},~\eqref{eq:UB_main}. Black dotted and dashed lines indicate the scaling of $1/E^2$ and $1/E^4$. Inset shows the logarithm in base ten upper bound versus circuit depth and energy. Green triangle (main) and line (inset) show the corresponding boundary of variance at $E_c(L)$ and $\ell_c(E)$.}
    \label{fig:var_tmsv}
\end{figure}

As an example, we consider a multipartite entangled distributed squeezed state generated by passing a single-mode squeezed vacuum over a random beamsplitter array (a global random passive Gaussian unitary), which is known to be typically highly entangled in the study of continuous-variable quantum information scrambling~\cite{zhang2021entanglement}. They are also the crucial form of entanglement that empowers distributed quantum sensing applications~\cite{zhuang2018distributed,zhang2021distributed,zhuang2019physical,xia2021quantum,brady2022entangled,xia2023entanglement}. For a fixed number of modes $M=10$, in Fig.~\ref{fig:ruc}(a), we see that the correlators $C_1^{\rm Gauss}$ (orange) and $C_2^{\rm Gauss}(\bm z)$ with $\bm z=\{1/2,\cdots,1/2\}$ (red) approach the scaling of $1/E^{M}$ (orange dotted) and $1/E^{2M}$ (red dashed) separately. On the other hand, with an asymptotic energy of $E=10^3$, we see the two correlators decays exponentially with the mode number $M$ in Fig.~\ref{fig:ruc}(b). In this case, however, the direct evalaution of gradients is challenging due to the large number of $M=10$ modes.

\begin{figure}[t]
    \centering
    \includegraphics[width=0.45\textwidth]{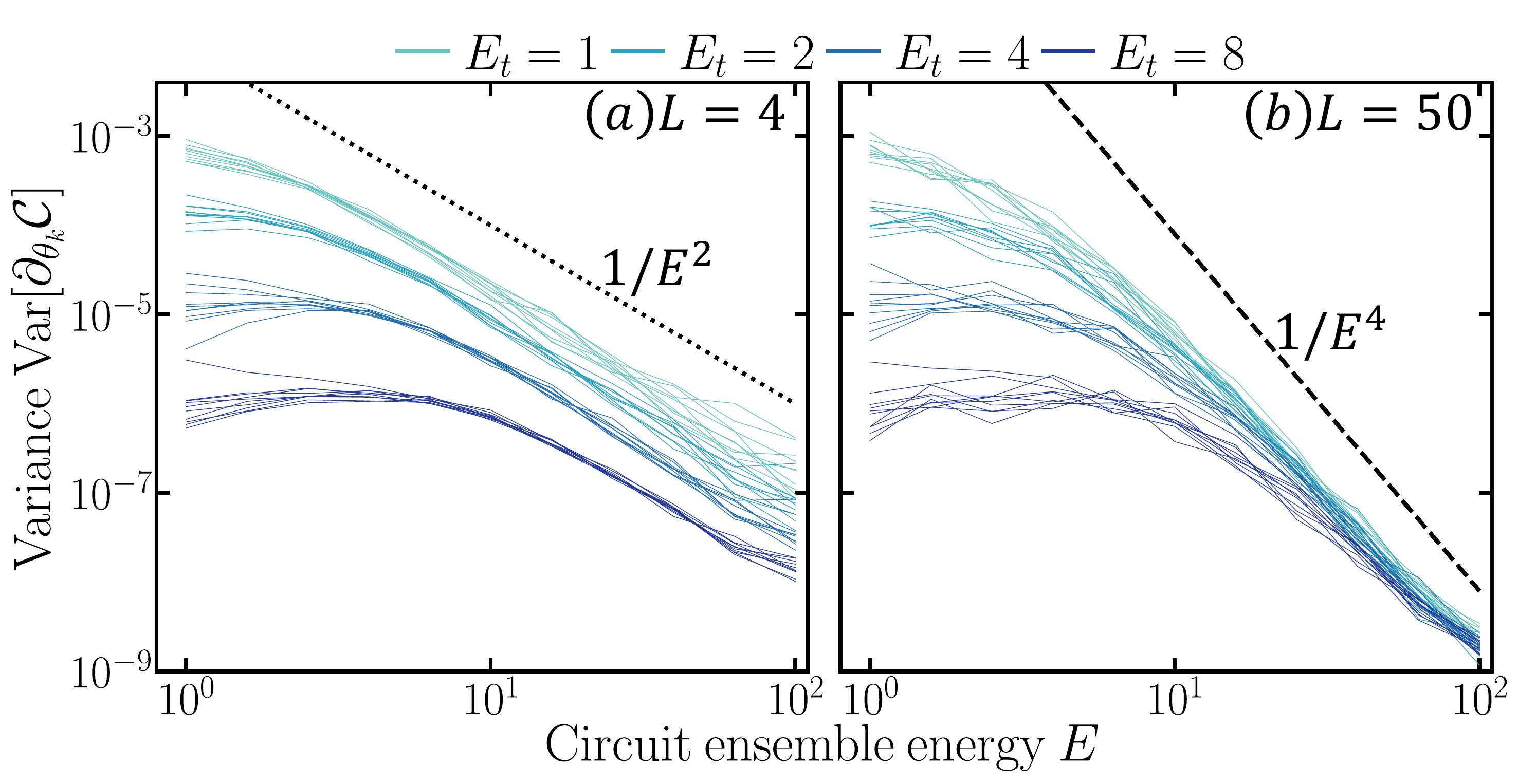}
    \caption{Variance of gradient ${\rm Var}[\partial_{\theta_{k}}\calC]$ at $k=ML/2$ in preparation of random two-mode CV states $\ket{\psi}_{\bm m}=\sum_{n_1,n_2} b_{n_1,n_2} \ket{n_1}_{m_1}\ket{n_2}_{m_2}$ with $L=4$ (a) and $L=50$ (b) circuits. Curves with same color show the variance of different sample target states. Black dotted and dashed lines in (a) and (b) represent the scaling of $1/E^2$ and $1/E^4$. In the calculation, we have chosen $\epsilon=0.1$ and $n_c\sim 2 E_t$. 
    }
    \label{fig:var_rand2}
\end{figure}

To evaluate the gradient of CV VQCs in preparation of Gaussian states for a direct comparison with Theorem~\ref{theorem:multi}, we consider the two-mode squeezed vacuum (TMSV) states, the CV analog of Bell states. A TMSV state $\ket{\zeta}_{\rm TMSV}$ is generated by applying a two-mode squeezing operator $S_2(\zeta) = \exp[\zeta(m_1 m_2- m_1^\dagger m_2^\dagger)/2]$ on vacuums $\ket{0}_{m_1}\ket{0}_{m_2}$, and has energy per mode $E_t = \sinh^2(\zeta)$. The correlators can be found utilizing Eqs.~\eqref{eq:C1_gauss_modes_main},~\eqref{eq:C2_gauss_modes_main} as
\begin{align}
    C_1^{\rm TMSV} &= \frac{\sech^4(\zeta)}{G_1(E)}, \label{eq:C1_tmsv_main}\\
    C_2^{\rm TMSV}(z_1, z_2) &= \frac{\sech^4(\zeta)}{G_2(z_1, z_2)G_2(1-z_1,1-z_2)}, \label{eq:C2_tmsv_main}
\end{align}
where $G_2(z_1,z_2) = 1+2(z_1+z_2)E + 4\sech^2(\zeta)z_1 z_2 E^2$. In the asymptotic region, the correlators show the scaling of $1/E^2$ and $1/E^4$ separately, and thus correspond to the scaling of gradient variance.
We compare the results above to numerical simulation in Fig.~\ref{fig:var_tmsv}, and see good agreement in asymptotic region of $E$ for both shallow and deep circuits, while the variance in shallow circuits with finite energy $E$ could deviate from the asymptotic predictions.


Similar to the single-mode case, the evaluation of correlators for non-Gaussian states is in general challenging. We conjecture that for general non-Gaussian states the barren plateau also holds, and support it with numerical results of preparation of randomly generated target states---a natural generalization to the one-mode case studied in Fig.~\ref{fig:var_rand}, $\ket{\psi}_{\bm m} \propto \sum_{n_1,n_2=0}^{n_c} b_{n_1,n_2}\ket{n_1}_{m_1}\ket{n_2}_{m_2}$, where each $b_{n_1,n_2}\sim \calN_{2}^{\rm C}$ is randomly chosen. We post-select states within the energy window $[E_t- \epsilon,E_t+\epsilon]$. With the target states generated, we evaluate the gradient variance versus circuit energy for different choices of target state energy $E_t$ (indicated by the color) in Fig.~\ref{fig:var_rand2}. Again, the scaling of $\sim 1/E^2$ in shallow circuits (left) and $\sim 1/E^4$ in deep circuits (right) are verified. However, numerical simulation for non-Gaussian states with more modes is still challenging due to the enormous demand of computing resource.

\begin{figure}[t]
    \centering
    \includegraphics[width=0.45\textwidth]{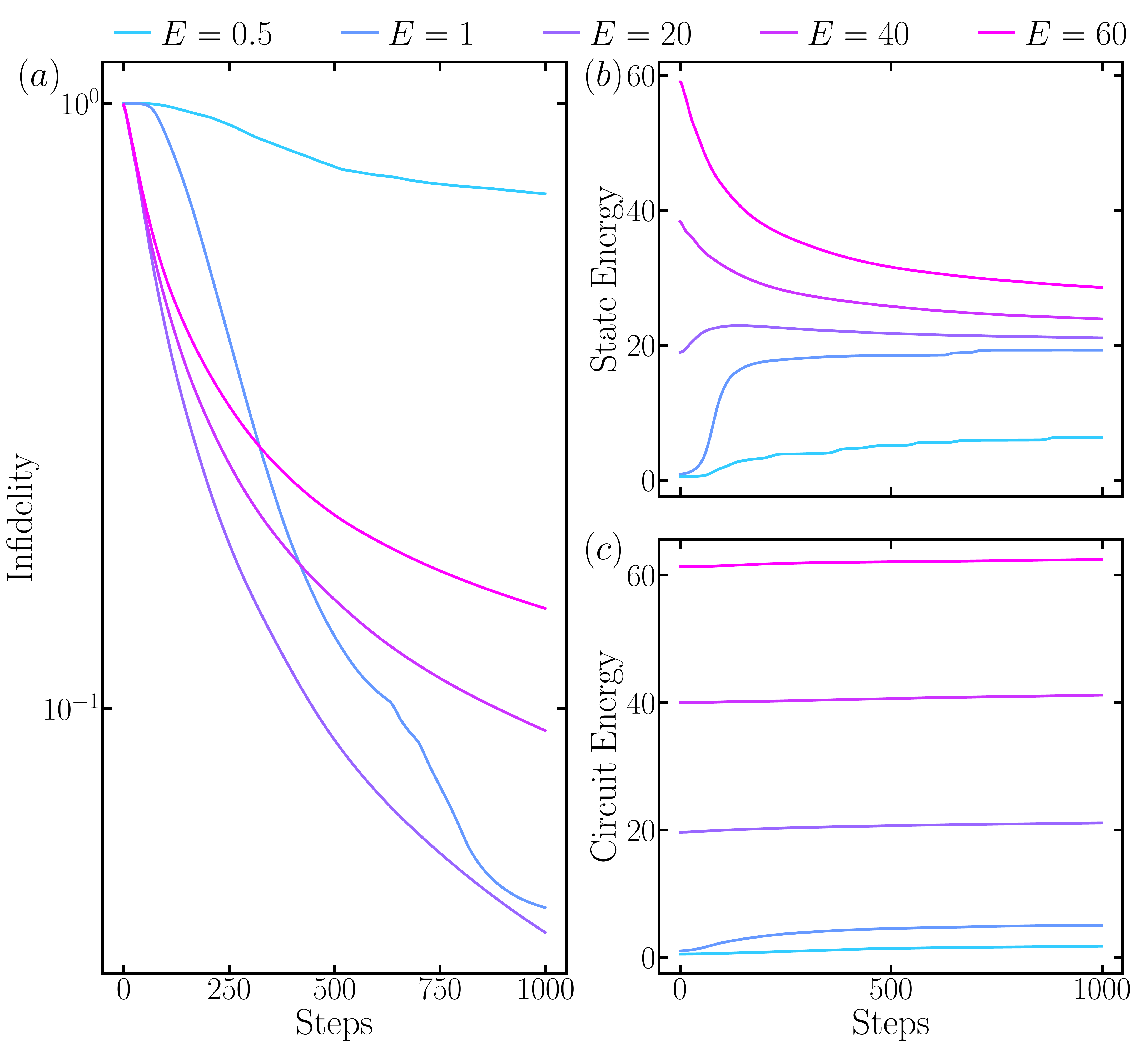}
    \caption{Training for Fock state $\ket{\psi}_m = \ket{20}_m$ with a $L=50$ CV VQC initilized with different ensemble energy $E$. We show (a) average infidelity of output state, (b) average output state energy and (c) average circuit energy $\sum_{j=1}^L |\beta_j|^2$ versus training steps.}
    \label{fig:training}
\end{figure}

\subsection{Strategies to circumvent training issues}
Barren plateau in general creates a challenge in training VQCs to solve problems. Although general approaches of entirely solving the training issues seem out of reach due to complexity arguments, problem-specific heuristics is promising in mitigating barren plateau. 
For DV VQCs, various methods~\cite{grant2019initialization,volkoff2021large,kiani2022learning,stefan2022,cichy2022perturbative,liu2022mitigating} have been proposed to mitigate barren plateau in training. 

In the case of the CV VQCs, the random ensemble $\calU_{E,L,M}$ has a unique tunable parameter---the circuit energy $E$---besides the circuit depth, thanks to the infinite dimensional Hilbert space. Therefore, different random initialization strategies can be adopted by varying the circuit energy $E$. As shown in Fig.~\ref{fig:training}(a) for the preparation of a Fock number state, when we adopt different initial circuit energy $E$, the training history of the cost function (state infidelity) is drastically different. When the circuit energy $E$ roughly matches the target state energy $E_t=20$, we see the best convergence. This is due to the peak of variance of gradient shown in Fig.~\ref{fig:var_states}(b) in Fock state preparation.
As expected, the decrease of infidelity is also reflected from the ensemble state energy which decays to the target state energy, shown in Fig.~\ref{fig:training}(b). In contrast, we also plot the circuit energy, defined as $\sum_{\ell=1}^L |\beta_\ell|^2$, in Fig.~\ref{fig:training}(c), and all of them changed only slightly during the training. Similarly, we also verify the strategy in the preparation of an SMSV state (see Appendix~\ref{app:one_mode_gaussian}), initialization with the proper energy also improves the convergence. As the gradients are maximal with zero energy in that case, low energy initializations lead to the best training convergence.

Overall, by initializing the circuit at low energy, one can mitigate the barren plateau to speed up the convergence. The best initial energy is, however, target state dependent. For instance, to prepare a Gaussian state, the quadrature mean of target state affects monotonicity of variance of gradient in the non-asymptotic region of $E$ (see Eqs.~\eqref{eq:C1_gauss_main} and~\eqref{eq:C2_gauss_main}): for zero-mean cases, the gradients keep decreasing with increasing circuit energy, as we see in Fig.~\ref{fig:var_tmsv} for TMSV and similarly for SMSV in Appendix~\ref{app:one_mode_gaussian}; for non-zero mean cases, the gradients are maximal when the circuit energy $E$ matches the target state energy $E_t$, as we see in Fig.~\ref{fig:var_states}(a). For general states, both monotonic decreasing cases and peaked cases can be found in Fig.~\ref{fig:var_rand} for random CV states, while the number state is found to be peaked in gradient variance in Fig.~\ref{fig:var_states}(b). In practice, when training the circuit to prepare a specific state, one can spend some computation resource in evaluating the gradients in the low energy region ($E\lesssim E_t$) to pick the best initial point before the actual training, which can drastically speed up the convergence.

\section{Discussion}

In this paper, we explore the trainability of CV VQCs implemented with universal control based on ECD gates. Through examples of preparing $M$-mode general Gaussian states and Fock number states, we analytically identify the barren plateau phenomena that the variance of the gradient decays with the circuit energy ${\rm Var}\sim 1/E^M$ in shallow circuit, while ${\rm Var}\sim 1/E^{2M}$ in deep ones. The barren plateau is also numerically extended to arbitrary non-Gaussian single-mode and two-mode CV states. To mitigate the barren plateau, we propose a strategy by tuning the ensemble energy in initialization to match the behavior of gradient variance, which is shown to be able to boost performance. 

Finally, we point out a few open problems. We have focused on the gradients with respect to qubit rotations to study the trainability of CV VQCs (as low trainability in part of the parameters suffices to demonstrate a barren plateau), it is unknown how the gradients with respect to the displacement parameters in the ECD gates decay as the energy and number of modes. It is also an important task to generalize our results to general tasks other than state preparation.
In this regard, we believe that the energy-dependent barren plateau can be generalized to the training including any bounded multi-mode CV observables. As any bounded operator acting on qumodes has a Glaubenr–Sudarshan P representation~\cite{mehta1967diagonal,vourdas2006analytic} in coherent state basis, the expectation value of the bounded operator can then be obtained by integration over coherent states' expectation values. Therefore, arbitrary bounded cost function can always be interpreted as a weighted average (with possible complex weight) of the cost function for coherent states, and the energy-dependent barren plateau in state preparation should typically represent the trainability involving any bounded operator. Beyond barren plateau, it is also of interest to explore other challenges such as traps~\cite{anschuetz2022quantum} in training in the CV VQCs.




\begin{acknowledgements}
This project is supported by the NSF CAREER Award CCF-2142882. QZ also acknowledges support from Defense Advanced Research Projects Agency (DARPA) under Young Faculty Award (YFA) Grant No. N660012014029, Office of Naval Research (ONR) Grant No. N000142312296 and Cisco Systems, Inc.. 
This research was supported in part by the National Science Foundation under Grant No. NSF PHY-1748958, during QZ's participation in KITP DYNISQ22 workshop.

QZ proposed the study. BZ performed the analyses and generated the figures, under the supervision of QZ. Both authors wrote the manuscript.
\end{acknowledgements}

%

\appendix 

\tableofcontents

\section{Gaussian states}
\label{app:Gaussian_intro}
Here we provide a succinct introduction to Gaussian states. More details can be found in Ref.~\cite{weedbrook2012gaussian}, the convention of which is utilized here. 

A system consisting of $M$ modes is described by $M$ annihilation and creation operators, $\{m_j,m^\dagger_j\}_{j=1}^M$, and they satisfy the commutation relation $[m_j,m^\dagger_{j^\prime}] = \delta_{j,j^\prime}$. One can also describe it via the position and momentum operators in the phase space $q_j = m_j+m_j^\dagger$ and $p_j = i(m_j^\dagger - m_j)$. Those quadratures can be grouped together to form a quadrature vector as $\mathcal{X}=(q_1,p_1,\dots,q_M,p_M)^T$. The first and second moments, which are also known as the mean quadrature and covariance matrix (CM) as
\begin{align}
    \overline{\mathcal{X}} &= \braket{\mathcal{X}}\\
    \bV_{ij} &= \frac{1}{2}\braket{\{q_i-\braket{q_i},q_j-\braket{q_j}\}},
    \label{eq:CM_def}
\end{align}
where $\braket{\cdot}$ represents the expectation and $\{A,B\}$ is the anticommutator of operators $A,B$. Any one-mode pure Gaussian state can be represented as a rotated and displaced squeezed state 
\begin{align}
    \ket{\psi}_m = D(\gamma)R(\tau)S(\zeta)\ket{0}_m ,\label{eq:one_mode_gaussian}
\end{align}
where $R(\tau) = \exp(-i\tau m^\dagger m)$ is the phase rotation and $S(\zeta) = \exp[\zeta(m^2-m^{\dagger 2})/2]$ is the squeezing operator. The displacement operator $D(\beta) = \exp(\beta m^\dagger-\beta^* m)$ satisfies the braiding relation $D(\alpha)D(\beta)=e^{(\alpha \beta^\star-\alpha^\star \beta)/2}D(\alpha+\beta)$. Its mean quadrature and CM are 
\small 
\begin{align}
    \overline{\mathcal{X}} &= (2\Re{\gamma},2\Im{\gamma})^T \label{eq:mean_one_mode}\\
    V &= \begin{pmatrix}
        e^{2\zeta} \sin^2(\tau)+e^{-2\zeta}\cos^2(\tau) & \sin(2\tau) \sinh(2\zeta) \\
        \sin(2\tau) \sinh(2\zeta) &  e^{2\zeta} \cos^2(\tau)+e^{-2\zeta}\sin^2(\tau)
    \end{pmatrix},
    \label{eq:CM_one_mode}
\end{align}
\normalsize
Below we specify the CM and mean for some examples; A coherent state has $\tau=\zeta=0$. Therefore, the CM is reduced to $\bI_2$, a $2\times 2$ identity matrix while the mean can be abitrary; A single-mode squeezed vacuum (SMSV) state has $\gamma=\tau=0$. Therefore, it has zero mean $\overline{\mathcal{X}}=0$ and a diagonal CM $V={\rm diag}(e^{-2\zeta},e^{2\zeta})$.

A two-mode squeezed vacuum (TMSV) state is the maximally entangled Gaussian state for two modes, which can be generated by a two-mode squeezing operator $S_2(\zeta) = \exp[\zeta(m_1m_2-m_1^\dagger m_2^\dagger)/2]$ on a product of vacuum states $\ket{0}_{m_1}\ket{0}_{m_2}$. It also has a zero mean and its CM becomes
\begin{align}
    \bV_{\rm TMSV} = \begin{pmatrix}
        \cosh (2\zeta) \bI_2 &\sinh (2\zeta) \sigma^z\\
        \sinh (2\zeta) \sigma^z &\cosh (2\zeta) \bI_2
    \end{pmatrix},
    \label{eq:CM_tmsv}
\end{align}
where $\bI_2$ is $2\times 2$ identity operator and $\sigma^z$ is Pauli-Z operator.

The fidelity between two Gaussian quantum state $\rho_A, \rho_B$ is fully determined by the mean quadratures $\overline{\mathcal{X}}_A, \overline{\mathcal{X}}_B$ and CMs $V_A,V_B$. Moreover, it can be analytically solved~\cite{marian2012uhlmann,spedalieri2012limit,banchi2015quantum}. When one of the state is pure, it has the simple form as~\cite{spedalieri2012limit}
\begin{align}
    F(\rho_A, \rho_B) &= \left(\Tr{\sqrt{\sqrt{\rho_A}\rho_B\sqrt{\rho_A}}}\right)^2\\
    &= \frac{2^M}{\sqrt{\det(V_A+V_B)}}\exp\left[-\frac{1}{2}{\bm d}^T (V_A + V_B)^{-1} {\bm d}\right],
    \label{eq:Fidelity_gaussian}
\end{align}
where $\bm d \equiv \overline{\mathcal{X}}_A - \overline{\mathcal{X}}_B$.

\section{Universality of ECD gate sets: proof of Lemma~\ref{lemma_universality}}
\label{app:universal_control}

For the CV VQCs in the main text involving interactions between a qubit and a qumode, Ref.~\cite{eickbusch2022fast} has shown that the gate set of ECD gates and single qubit rotations is universal, in the sense that linear combinations of repeated nested commutators of the gate generators cover the full Lie algebra of the qubit-qumode system~\cite{braunstein2005quantum, d2021introduction}. Below, we review the proof. The generators of ECD gates and single qubit rotations are
\begin{equation}
    \calG = \{q\sigma^z, p\sigma^z, \sigma^x, \sigma^y\},
    \label{eq:generators}
\end{equation}
where $q, p$ are position and momentum of the qumode, and $\sigma^c$ with $c\in \{x,y,z\}$ are Pauli operators on the qubit.
First, the commutators between $q\sigma^z, p\sigma^z$ and $\sigma^x, \sigma^y$ produce operators in the form of $q\sigma^c, p\sigma^c$ with $c \in \{x, y, z\}$. To obtain operators of higher order polynomials of $q, p$, one can consider the commutator $[q\sigma^a, q\sigma^b]\propto \epsilon_{abc}q^2\sigma^c$ where $\epsilon_{abc}$ is the three-dimensional Levi-Civita symbol. The commutation can be repeated to generate operators $q^j\sigma^a$ with $j \ge 2$ and similarly for $ p^j\sigma^a$. To obtain the terms coupling quadratures along with Pauli operators, we consider the commutator
\begin{equation}
    [q^{j+1}\sigma^a,p\sigma^b] 
    = 2i\epsilon_{abc}q^{j+1}p\sigma^c + (j+1)\epsilon_{abc}q^{j}\sigma^c,
\end{equation}
assuming $a \neq b$.
Combined with $q^j\sigma^c$, it leads to $p^{j+1}\sigma^c$. By repeating the process on commutator with $p\sigma^b$, one has all polynomial quadrature terms $q^jp^k\sigma^c$ with $c\in\{x,y,z\}$. The last step is to eliminate the Pauli operators by commutator $[q^{j+1}p^k\sigma^a,p\sigma^a] \propto q^jp^k$. Therefore, all unitaries with generators in the form $q^jp^k\sigma^c$ with $\sigma^c\in\{\bI_2, \sigma^x,\sigma^y,\sigma^z\}$ are achieved by the gate sets involving ECD gates and single-qubit rotations.

Now we generalize the universality of ECD gates and single-qubit unitaries to arbitrary $M\ge 1$ qumodes and $N\ge 1$ qubits. The generators of ECD gates (between any qubit-qumode pair) combined with single qubit rotations are
\begin{equation}
    \calG = \left[\bigcup_{\ell=1}^M\bigcup_{r=1}^N \{q_\ell\sigma^z_r,p_\ell\sigma^z_r\}\right]\bigcup\left[\bigcup_{r=1}^N\{\sigma_r^x,\sigma_r^y\}\right].
\end{equation}
For two modes $\ell,\ell^\prime$ that are coupled to a qubit $r$ by ECD gates, the commutator between $q_\ell^{j_\ell}p_\ell^{k_\ell}\sigma_r^a$ and $q_{\ell^\prime}^{j_{\ell^\prime}}p_{\ell^\prime}^{k_{\ell^\prime}}\sigma_r^b$ (assuming $\ell\neq \ell^\prime$) is 
\begin{equation}
\begin{split}
    \left[q_\ell^{j_\ell}p_\ell^{k_\ell}\sigma_r^a, q_{\ell^\prime}^{j_{\ell^\prime}}p_{\ell^\prime}^{k_{\ell^\prime}}\sigma_r^b\right] &= q_\ell^{j_\ell}p_\ell^{k_\ell}q_{\ell^\prime}^{j_{\ell^\prime}}p_{\ell^\prime}^{k_{\ell^\prime}}\left[\sigma_r^a,\sigma_r^b\right]\\
    &\propto q_\ell^{j_\ell}p_\ell^{k_\ell}q_{\ell^\prime}^{j_{\ell^\prime}}p_{\ell^\prime}^{k_{\ell^\prime}}\epsilon_{abc}\sigma_r^c.
\end{split}
\end{equation}
Lastly, the commutator $
\left[q_\ell^{j_\ell}p_\ell^{k_\ell}q_{\ell^\prime}^{j_{\ell^\prime}+1}p_{\ell^\prime}^{k_{\ell^\prime}}\sigma_r^x, p_{\ell^\prime}\sigma_r^x\right] \propto q_\ell^{j_\ell}p_\ell^{k_\ell}q_{\ell^\prime}^{j_{\ell^\prime}}p_{\ell^\prime}^{k_\ell^\prime}$, where any operators in the form $q_\ell^{j_\ell}p_\ell^{k_\ell}q_{\ell^\prime}^{j_{\ell^\prime}}p_{\ell^\prime}^{k_{\ell^\prime}}\sigma_r^c$ with $\sigma_r^c\in\{\bI_2,\sigma_r^x,\sigma_r^y, \sigma_r^z\}$ can be generated. Through repeated process above for the other $M-2$ modes and one can generate arbitrary unitary with generator $\left(\prod_{\ell=1}^M q_\ell^{j_\ell}p_\ell^{k_\ell}\right)\sigma_r^c$ with universal control on $M$ modes and one qubit.
Next, we consider two qubits $r, r^\prime$ connected to one mode $\ell$, the commutator between $q_\ell^{j_\ell}p_\ell^{k_\ell+1}\sigma_r^z$ and $q_\ell\sigma_{r^\prime}^z$ is
\begin{equation}
\begin{split}
    \left[q_\ell^{j_\ell}p_\ell^{k_\ell+1}\sigma_r^z, q_\ell\sigma_{r^\prime}^z\right] &= q_\ell^{j_\ell}\left[p_\ell^{k_\ell+1},q_\ell \right]\sigma_r^z\sigma_{r^\prime}^z\\
    &\propto q_\ell^{j_\ell}p_\ell^{k_\ell}\sigma_r^z\sigma_{r^\prime}^z.
\end{split}
\end{equation}
With the commutators between $q_\ell^{j_\ell}p_\ell^{k_\ell}\sigma_r^z\sigma_{r^\prime}^z$ and $\{\sigma_r^x, \sigma_{r^\prime}^x, \sigma_r^y, \sigma_{r^\prime}^y\}$, we have the form $q_\ell^{j_\ell}p_\ell^{k_\ell}\sigma_r^c\sigma_{r^\prime}^{c^\prime}$ with $\sigma_r^c,\sigma_{r^\prime}^{c^\prime} \in \{\bI_2, \sigma_r^x, \sigma_r^y, \sigma_r^z\}$. By repeating the process discussed above, we can generate all unitaries whose generator is in the form of $q_\ell^{j_\ell}p_\ell^{k_\ell}\left(\prod_r \sigma_r^c\right)$. Combined with the results with $M$ modes and one qubit, we finally obtain generators in the form
$
\left(\prod_{\ell =1}^M q_\ell^{j_\ell}p_\ell^{k_\ell}\right)\left(\prod_{r=1}^N \sigma_r^c\right)
$
where the universal control is showed to be performed on the system with $M$ modes and $n$ qubits.

\begin{widetext}

\section{Representation of states: single-mode case}
\label{app:state_repre}

In this section, we present the proof for the representation of the single-mode energy-regularized ensemble of states in Eq.~\eqref{eq:state_ensemble}. The generalization to the multi-mode case is presented in Appendix~\ref{app:multi-mode}. We first prove that the representation holds for all $L$-layer circuits, and then provide analysis on the energy-regularization. 

To simplify the notation, we define ${\bm \beta} = (\beta_1, \dots, \beta_L)^T$, ${\bm \theta}= (\theta_1, \dots, \theta_{L})^T$, ${\bm \phi}= (\phi_1, \dots, \phi_{L})^T$ and the overall parameters $\bm x=({\bm \beta},{\bm \theta}, {\bm \phi})$.

Consider a single-mode qubit-qumode variational circuit consisting of ECD blocks, in the qubit $\ket{0}$ and $\ket{1}$ basis, each ECD block can be written in the matrix form as
\begin{equation}
    U_{\rm ECD}(\beta)U_{\rm R}(\theta, \phi)=\begin{pmatrix}
    e^{i\phi}\sin\frac{\theta}{2}D(-\beta) & \cos\frac{\theta}{2}D(-\beta)\\
    \cos\frac{\theta}{2}D(\beta) & e^{i(\pi-\phi)}\sin\frac{\theta}{2}D(\beta)
    \end{pmatrix},
    \label{eq:block_matrix}
\end{equation}
where for convenience, we relabel the variable as $\phi-\pi/2\rightarrow\phi$, and we use this definition through the following. The output state from unitary $U = \prod_{\ell=1}^L U_{\rm ECD}(\beta_\ell)U_{\rm R}(\theta_\ell, \phi_\ell)$ on the input state $\ket{0}_q\ket{0}_m$ is
\begin{align}
    \ket{\psi(\bm \beta, \bm \theta, \bm \phi)}_{q,m} &= U(\bm \beta,\bm \theta, \bm \phi)\ket{0}_q\ket{0}_m\nonumber\\
    &= \sum_{a=0}^1 \sum_{\bm s} w_{\bm s,a}(\bm \theta, \bm \phi)\ket{a}_q e^{i\chi_{\bm s}(\bm \beta)}\ket{(-1)^a \bm s\cdot \bm \beta}_m
    \equiv\sum_{a=0}^1 \sum_{\bm s} w_{\bm s,a}(\bm \theta, \bm \phi)\ket{a}_q \ket{B_{\bm s,a}}_m \label{eq:psi}
\end{align}
where the length-$L$ sign vector $\bm s$ is defined as $\bm s=(\bm s_{1:L-1},-1)$ with $\bm s_{1:L-1}\in \{-1,1\}^{L-1}$. We absorb the extra phase in qumode ket vector as $\ket{B_{\bm s,a}}\equiv e^{i\chi_{\bm s}(\bm \beta)}\ket{(-1)^a \bm s\cdot \bm \beta}_m$ for convenience, where the phase is defined to be
\begin{align}
    \chi_{\bm s}(\bm \beta) \equiv \sum_{\ell=1}^{L-1}\sum_{\ell^\prime=\ell+1}^{L}\bm s_\ell \bm s_{\ell^\prime} (\Re{\beta_\ell}\Im{\beta_{\ell^\prime}}-\Im{\beta_\ell}\Re{\beta_{\ell^\prime}}), \label{eq:chi_def}
\end{align}
Via defining $v_{\bm s,a}(\bm \theta,\bm \phi, \bm \beta)\equiv e^{i\chi_{\bm s}(\bm \beta)}w_{\bm s,a}(\bm \theta, \bm \phi)$, Eq.~\eqref{eq:psi} is equivalent to Eq.~\eqref{eq:state_ensemble} in the main paper.
In Eq.~\eqref{eq:psi}, $\ket{a}_q$ is the qubit state in computational basis. The weight for each of the superpositions is 
\begin{align}
    w_{\bm s,a}(\bm \theta, \bm \phi) = e^{i\Phi_{\bm s,a}(\bm \phi)}T_{\bm s,a}(\bm \theta),
    \label{eq:psi_w}
\end{align}
where
\begin{align}
    \Phi_{\bm s,a}(\bm \phi) &\equiv a(n_{\bm s}\pi+\phi_1)+\mathbb{P}(a)\sum_{\ell=1}^L \left({\bf d}\bm s_\ell-1\right)\left(\bm s_\ell\phi_\ell-\delta_{\bm s_\ell,1}\pi\right), \label{eq:Phi_def}\\
    T_{\bm s, a}(\bm \theta) &\equiv \prod_{\ell=2}^L \sin\left(\frac{\theta_\ell+{\bf d}\bm s_\ell\pi}{2}\right) \sin\left(\frac{\theta_1}{2}+\frac{(\mathbb{P}(a)\bm  s_1+1)\pi}{4}\right).  \label{eq:T_def}
\end{align}
There are some notations to be explained. 
We define $\mathbb{P}(x) \equiv (-1)^x$ as the parity of a variable $x$, which equals $\pm 1$ given $x$ is even/odd. We define the difference sign vector ${\bf d}\bm s$ to represent the change of signs in the vector $\bm s$---the $\ell$th element of the difference sign vector is 
\begin{align}
    {\bf d}\bm s_\ell \equiv |\bm s_{\ell}-\bm s_{\ell-1}|/2,
    \label{eq:sign_diff}
\end{align}
which is zero when $\bm s_{\ell+1}=\bm s_\ell$ but $1$ otherwise. Note that $(-1)^{{\bf d}\bm s_\ell}=\bm s_{\ell}\times \bm s_{\ell-1}$. As $\bm s_L= -1$ always, we assign ${\bf d}\bm s_1 = |\bm s_1 - \bm s_L|/2 = (\bm s_1 + 1)/2$ such that it reflects the value of $\bm s_1$. 

We can find that the size of $\{{\bf d}\bm s_\ell|{\bf d}\bm s_\ell=0\}_{\ell\ge 2}$ is 
\begin{equation}
    n_{\bm s} = L-1-\sum_{\ell=2}^L {\bf d}\bm s_\ell
\end{equation}
since ${\bf d}\bm s_\ell = 0,1$. We would like to comment that $n_{\bm s}$ also equals the number of $\{\phi_\ell\}_{\ell \ge 2}$ appearing in $\Phi_{\bm s,a}$ with nonzero coefficient and the number of $\{\sin(\theta_\ell/2)\}_{\ell\ge 2}$ appearing in $T_{\bm s,a}$. From Eq.~\eqref{eq:Phi_def}, as the coefficient of each $\phi_\ell$ in $\Phi_{\bm s,a}$ with $\ell\ge 2$ is $\mathbb{P}(a)({\bf d}\bm s_\ell-1)\bm s_\ell$, since $\mathbb{P}(a)=(-1)^a$ and $\bm s_\ell=\pm 1$, only ${\bf d}\bm s_\ell = 1$ can make the coefficient be zero, and thus $n_{\bm s}$ counts the number of $\{\phi_\ell\}_{\ell\ge 2}$ with nonzero coefficient. Moreover, one can also see from Eq.~\eqref{eq:T_def} that if ${\bf d}\bm s_\ell=0$ with $\ell\ge 2$, the function of $\theta_\ell$ in $T_{\bm s,a}$ is in the form of $\sin(\theta_\ell/2)$ but $\cos(\theta_\ell/2)$ if ${\bf d}\bm s_\ell = 1$. To count the total number of $\{\phi_\ell\}_{\ell=1}^L$ with nonzero coefficient in $\Phi_{\bm s,a}$, denoted as $N_{\bm s,a}$, from Eq.~\eqref{eq:Phi_def} the coefficient of $\phi_1$ is $a+\mathbb{P}(a)({\bf d}\bm s_1-1)\bm s_1 =a+ (-1)^a({\bf d}\bm s_1-1)(2{\bf d}\bm s_1-1) =1-a-(-1)^a{\bf d}\bm s_1 \ge 0$, where in the first equation we utilize the definition $\bm s_1=2{\bf d}\bm s_1-1$ and in the second equation we utilize $({\bf d}\bm s_1)^2 = {\bf d}\bm s_1$ and $a+(-1)^a=1-a$. As the coefficient of $\phi_1$ is either zero or one, the total number is just 
\begin{align}
    N_{\bm s,a} = n_{\bm s}+1-a-(-1)^a {\bf d}\bm s_1.
\end{align}
Note that it is also the total number of $\{\sin(\theta_\ell/2)\}_{\ell=1}^L$ in $T_{\bm s,a}$, because $\sin(\theta_1/2)$ exists in $T_{\bm s,a}$ only if $1-(\mathbb{P}(a)\bm s_1+1)/2 = 1$. Using the identities $\bm s_1 = 2{\bf d}\bm s_1-1$ and $1+(-1)^a = 2a$, we have $1-(-1)^a{\bf d}\bm s_1+a$, which is exactly the terms following $n_{\bm s}$ in $N_{\bm s,a}$ we have above.

To help understanding the notations, we provide an example of $L=2$ as following
\begin{align}
    &\ket{\psi(\bm \beta, \bm \theta, \bm \phi)}_{q,m}\nonumber\\
    &= \ket{0}_q\left(e^{i(\phi_1+\phi_2+\beta_1^{\rm R}\beta_2^{\rm I}-\beta_1^{\rm I}\beta_2^{\rm R})}\sin\frac{\theta_1}{2}\sin\frac{\theta_2}{2}\ket{-\beta_1-\beta_2}_m + e^{i(\beta_1^{\rm I}\beta_2^{\rm R}-\beta_1^{\rm R}\beta_2^{\rm I})}\cos\frac{\theta_1}{2}\cos\frac{\theta_2}{2}\ket{+\beta_1-\beta_2}_m\right)\nonumber\\
    &+\ket{1}_q\left(e^{i(\pi-\phi_2+\beta_1^{\rm R}\beta_2^{\rm I}-\beta_1^{\rm I}\beta_2^{\rm R})}\cos\frac{\theta_1}{2}\sin\frac{\theta_2}{2}\ket{+\beta_1+\beta_2}_m + e^{i(\phi_1-\beta_1^{\rm R}\beta_2^{\rm I}+\beta_1^{\rm I}\beta_2^{\rm R})}\sin\frac{\theta_1}{2}\cos\frac{\theta_2}{2}\ket{-\beta_1+\beta_2}_m\right).
    \label{eq:psi_d2}
\end{align}
where $\beta_{1,2}^{\rm R/ \rm I}$ denotes the real or imaginary part of $\beta_{1,2}$ to shorten the formula.
One can check that it agrees with the representation of state in Eq.~\eqref{eq:psi} with weight following the definitions from Eq.~\eqref{eq:psi_w}.

Below we present the detailed proof of the state representation of Eq.~\eqref{eq:psi} (equivalently Eq.~\eqref{eq:state_ensemble} in the main paper) by mathematical induction.

\begin{proof}
First, we start from $L= 2$: in this case it has already been shown in Eq.~\eqref{eq:psi_d2} that Eq.~\eqref{eq:psi} is true. Then we suppose that for an $L$-layer circuit, it is also in the form of Eq.~\eqref{eq:psi}, and for the $L+1$-layer circuit we have
\begin{align}
    &\ket{\psi_{L+1}(\bm \beta, \bm \theta, \bm \phi)} = U_{\rm ECD}(\beta_{L+1})U_{\rm R}(\theta_{L+1},\phi_{L+1})\ket{\psi_L(\bm \beta, \bm \theta, \bm \phi)}\nonumber\\
    &= \begin{pmatrix}
    e^{i\phi_{L+1}}\sin\left(\frac{\theta_{L+1}}{2}\right)D(-\beta_{L+1}) & \cos\left(\frac{\theta_{L+1}}{2}\right)D(-\beta_{L+1})\\
    \cos\left(\frac{\theta_{L+1}}{2}\right)D(\beta_{L+1}) & e^{i(\pi-\phi_{L+1})}\sin\left(\frac{\theta_{L+1}}{2}\right)D(\beta_{L+1})
    \end{pmatrix}
    \begin{pmatrix}
        \sum_{\bm s} w_{\bm s,0}(\bm \theta, \bm \phi)\prod_{\ell=1}^L D(\bm s_\ell \beta_\ell)\ket{0}_m\\
        \sum_{\bm s} w_{\bm s,1}(\bm \theta, \bm \phi) \prod_{\ell=1}^L D(-\bm s_\ell \beta_\ell)\ket{0}_m
    \end{pmatrix}\nonumber\\
    &=
    \begin{pmatrix}
        \sum_{\bm s} e^{i(\Phi_{\bm s,0}+\phi_{L+1})}\sin\left(\frac{\theta_{L+1}}{2}\right)T_{\bm s,0}D(-\beta_{L+1})\prod_{\ell=1}^L D(\bm s_\ell \beta_\ell)\ket{0}_m+\sum_{\bm s} e^{i\Phi_{\bm s,1}}\cos\left(\frac{\theta_{L+1}}{2}\right)T_{\bm s,1}D(-\beta_{L+1})\prod_{\ell=1}^L D(-\bm s_\ell \beta_\ell)\ket{0}_m\\
        \sum_{\bm s} e^{i\Phi_{\bm s,0}}\cos\left(\frac{\theta_{L+1}}{2}\right)T_{\bm s,0}D(\beta_{L+1})\prod_{\ell=1}^L D(\bm s_\ell \beta_\ell)\ket{0}_m + \sum_{\bm s} e^{i(\Phi_{\bm s,1}+\pi-\phi_{L+1})}\sin\left(\frac{\theta_{L+1}}{2}\right)T_{\bm s,1}D(\beta_{L+1})\prod_{\ell=1}^L D(-\bm s_\ell \beta_\ell)\ket{0}_m
    \end{pmatrix}.
    \label{eq:psi_d+1}
\end{align}
Note that the product of displacement operator above is evaluated from $\beta_1$ to $\beta_L$.
To compare with the representation in Eq.~\eqref{eq:psi}, we define $\bm s^\prime$ to be $(\pm \bm s, -1)$, which actually covers all possible cases in the definition of sign vector of length $L+1$. To prove the result, all we need to do is to show that Eq.~\eqref{eq:psi_d+1} agrees with the representation in Eq.~\eqref{eq:psi}. 

We start from the displacement on the qumode. The total displacement on qumode in Eq.~\eqref{eq:psi_d+1} can be directly seen that it satisfies $(-1)^a {\bm s}^\prime \cdot \bm \beta$ with $\bm s=(\pm \bm s,-1)$. There is an phase generated due to the braiding relation of displacement operators $D(\alpha) = e^{\alpha m^\dagger -\alpha^* m} = e^{-|\alpha|^2/2}e^{\alpha m^\dagger}e^{-\alpha^* m}$, and we can directly prove that it is in the form of Eq.~\eqref{eq:chi_def} by
\begin{align}
    & D(\alpha_L)\cdots D(\alpha_1) = e^{-\frac{1}{2}\sum_{\ell=1}^L |\alpha_\ell|^2}e^{\alpha_L m^\dagger}e^{-\alpha_L^* m}\cdots e^{\alpha_1 m^\dagger}e^{-\alpha_1^* m}\nonumber\\
    &=e^{-\frac{1}{2}\sum_{\ell=1}^L |\alpha_\ell|^2} \left(\prod_{\ell=1}^M e^{\alpha_\ell m^\dagger}\right)\left(\prod_{\ell=1}^L e^{-\alpha_\ell^* m}\right)\left(\prod_{\ell=1}^{L-1}e^{-\alpha_\ell(\sum_{\ell^\prime=\ell+1}^L \alpha_{\ell^\prime}^*)}\right)\\
    &= e^{-\frac{1}{2}\sum_{\ell=1}^L |\alpha_\ell|^2} e^{\sum_{\ell=1}^L \alpha_\ell m^\dagger} e^{-\sum_{\ell=1}^L \alpha_\ell^* m} e^{-\sum_{\ell=1}^{L-1}\sum_{\ell^\prime=\ell+1}^L \alpha_\ell \alpha_{\ell^\prime}^*}\\
    &= e^{-\frac{1}{2}\sum_{\ell=1}^L |\alpha_\ell|^2} e^{-\sum_{\ell=1}^{L-1}\sum_{\ell^\prime=\ell+1}^L \alpha_\ell \alpha_{\ell^\prime}^*} e^{\frac{1}{2}|\sum_{\ell=1}^L \alpha_\ell|^2}D\left(\sum_{\ell=1}^L \alpha_\ell\right)\\
    &= e^{i\sum_{\ell=1}^{L-1} \sum_{\ell^\prime=\ell+1}^L (\Re{\alpha_\ell}\Im{\alpha_{\ell^\prime}}-\Im{\alpha_\ell}\Re{\alpha_{\ell^\prime}})} D\left(\sum_{\ell=1}^L \alpha_\ell\right),
\end{align}
where in the second line we perform a reordering to the qumode operators such that all annilation operators follows creation ones, introducing extra phase by the Baker–Campbell–Hausdorff identity $e^{A}e^{B}=e^{B}e^{A}e^{[A,B]}$ (when higher-order commutators are zero). The last line is obtained from expanding every $\alpha_\ell = \Re{\alpha_\ell}+i\Im{\alpha_{\ell}}$. By letting $\alpha_{\ell}=(-1)^a \bm s_\ell \beta_\ell$, we have the formula in Eq.~\eqref{eq:chi_def}.

Next we need to show the weight satisfy the form in Eq.~\eqref{eq:psi_w}. Starting from the first item in the first line of Eq.~\eqref{eq:psi_d+1} which corresponds to $\bm s^\prime = (\bm s,-1)$ and $a=0$, the difference in sign vector is ${\bf d}\bm s^\prime = ({\bf d}\bm s, 0)$, then the phase and amplitude in the weight can be reduced to
\begin{align}
    \Phi_{\bm s,0}+\phi_{L+1} &= \sum_{\ell=1}^L ({\bf d}\bm s_\ell-1)\bm s_\ell \phi_\ell -({\bf d}\bm s_\ell-1)\delta_{\bm s_\ell,1}\pi + \phi_{L+1} = \sum_{\ell=1}^{L+1} ({\bf d}\bm s^\prime_\ell-1)\bm s_\ell^\prime \phi_\ell - ({\bf d}\bm s^\prime_\ell-1)\delta_{\bm s^\prime_\ell,1}\pi = \Phi_{\bm s^\prime, 0},
    \label{eq:Phi_d+1_1}\\
    \sin\left(\frac{\theta_{L+1}}{2}\right)T_{\bm s,0} &= \sin\left(\frac{\theta_{L+1}}{2}\right)\prod_{\ell=2}^L\sin\left(\frac{\theta_\ell+{\bf d}\bm s_\ell \pi}{2}\right)\sin\left(\frac{\theta_1}{2}+\frac{(\bm  s_1+1)\pi}{4}\right) \nonumber\\
   &=\prod_{\ell=2}^{L+1}\sin\left(\frac{\theta_\ell+{\bf d}\bm s^\prime_\ell \pi}{2}\right)\sin\left(\frac{\theta_1}{2}+\frac{(\bm s^\prime_1+1)\pi}{4}\right)= T_{\bm s^\prime, 0}, \label{eq:T_d+1_1}
\end{align}
where in Eqs.~\eqref{eq:Phi_d+1_1} and~\eqref{eq:T_d+1_1} we utilize the fact that ${\bf d}\bm s_{L+1}^\prime = 0$.

The last item in the second line of Eq.~\eqref{eq:psi_d+1} also corresponds to $\bm s^\prime=(\bm s,-1)$ but $a=1$ as the total displacement in qumode is $(-1)^a \bm s\cdot \bm \beta$, therefore the weight becomes 
\begin{align}
    \Phi_{\bm s,1}+\pi-\phi_{L+1} &= n_{\bm s}\pi + \phi_1-\Phi_{\bm s,0}+\pi - \phi_{L+1} = (n_{\bm s}+1)\pi + \phi_1 - \Phi_{\bm s^\prime, 0} = n_{\bm s^\prime}\pi +\phi_1 - \Phi_{\bm s^\prime, 0} = \Phi_{\bm s^\prime, 1} ,\label{eq:Phi_d+1_2}\\
    \sin\left(\frac{\theta_{L+1}}{2}\right)T_{\bm s,1} &= \sin\frac{\theta_{L+1}}{2}\prod_{\ell=2}^L \sin\left(\frac{\theta_\ell+{\bf d}\bm s_\ell \pi}{2}\right)\sin\left(\frac{\theta_1}{2}+\frac{(-\bm  s_1+1)\pi}{4}\right) 
    \nonumber
    \\
    &=\prod_{\ell=2}^{L+1}\sin\left(\frac{\theta_\ell+{\bf d}\bm s^\prime_\ell \pi}{2}\right)\sin\left(\frac{\theta_1}{2}+\frac{(-\bm s^\prime_1+1)\pi}{4}\right)= T_{\bm s^\prime, 1}, \label{eq:T_d+1_2}
\end{align}
where in the second equation of Eq.~\eqref{eq:Phi_d+1_2} we apply result from Eq.~\eqref{eq:Phi_d+1_1}, and the last equation of Eq.~\eqref{eq:Phi_d+1_2} is obtained by directly utilizing the definition of $n_{\bm s^\prime} = L-\sum_{\ell=2}^{L+1}{\bf d}\bm s^\prime_\ell = n_{\bm s}+1$ given ${\bf d}\bm s^\prime = ({\bf d}\bm s, 0)$.

The second item in the first line of Eq.~\eqref{eq:psi_d+1} corresponds to $\bm s^\prime = (-\bm s, -1)$ and $a=0$ where the difference is ${\bf d}\bm s^\prime = (1-{\bf d}\bm s_1, ({\bf d}\bm s)_{2:L}, 1)$. The phase and amplitude in weight are
\begin{align}
     \Phi_{\bm s,1} &= n_{\bm s}\pi + \phi_1 - \sum_{\ell=1}^L ({\bf d}\bm s_\ell-1)(\bm s_\ell \phi_\ell -\delta_{\bm s_\ell,1}\pi) = \left(L-1-\sum_{\ell=2}^L {\bf d}\bm s_\ell\right)\pi+ \phi_1 - \sum_{\ell=1}^L ({\bf d}\bm s_\ell-1)(\bm s_\ell \phi_\ell -\delta_{\bm s_\ell,1}\pi)\nonumber\\
     &= -\sum_{\ell=2}^L ({\bf d}\bm s_\ell-1)\bm s_\ell \phi_\ell + \left[1-({\bf d}\bm s_1-1)\bm s_1\right]\phi_1 - \sum_{\ell=2}^L  ({\bf d}\bm s_\ell-1)(1- \delta_{\bm s_\ell,1})\pi + ({\bf d}\bm s_1-1)\delta_{\bm s_1,1} \pi \nonumber \\
     &= \sum_{\ell=2}^L ({\bf d}\bm s_\ell^\prime-1)\bm s_\ell^\prime \phi_\ell + (1-{\bf d}\bm s^\prime_1\cdot \bm s^\prime_1)\phi_1 - \sum_{\ell=2}^L ({\bf d}\bm s_\ell^\prime-1)\delta_{\bm s^\prime_\ell,1}\pi - {\bf d}\bm s_1^\prime (1-\delta_{\bm s^\prime_\ell,1})\pi \nonumber \\
     &= \sum_{\ell=1}^{L+1}({\bf d}\bm s_\ell^\prime-1)\bm s_\ell^\prime \phi_\ell -\sum_{\ell=1}^{L+1} ({\bf d}\bm s^\prime_\ell-1)\delta_{\bm s^\prime_\ell,1}\pi = \Phi_{\bm s^\prime, 0}, \label{eq:Phi_d+1_3}
     \end{align}
and
     \begin{align}
     \cos\left(\frac{\theta_{L+1}}{2}\right)T_{\bm s,1} &= \cos\left(\frac{\theta_{L+1}}{2}\right)\prod_{\ell=2}^L \sin\left(\frac{\theta_\ell+{\bf d}\bm s_\ell \pi}{2}\right)\sin\left(\frac{\theta_1}{2}+\frac{(1-\bm  s_1)\pi}{4}\right) 
     \nonumber
     \\
     &= \prod_{\ell=2}^{L+1}\sin\left(\frac{\theta_\ell+{\bf d}\bm s^\prime_\ell \pi}{2}\right)\sin\left(\frac{\theta_1}{2}+\frac{(\bm s^\prime_1+1)\pi}{4}\right)= T_{\bm s^\prime, 0}, \label{eq:T_d+1_3}
\end{align}
where in the last line of Eq.~\eqref{eq:Phi_d+1_3}, we apply identities $1-{\bf d}\bm s_1^\prime \cdot \bm s_1^\prime = ({\bf d}\bm s_1^\prime-1)\bm s_1^\prime$ and ${\bf d}\bm s_1^\prime (1-\delta_{\bm s_1^\prime,1}) = ({\bf d}\bm s_1^\prime-1)\delta_{\bm s_1^\prime,1}$. One can easily check those identities by substituting $\bm s_1^\prime = \pm 1$ separately. 

Similarly, the first item in the second line of Eq.~\eqref{eq:psi_d+1} with $\bm s^\prime = (-\bm s, -1)$ but $a=1$ has the weight as
\begin{align}
    \Phi_{\bm s,0} &= n_{\bm s}\pi + \phi_1 - \Phi_{\bm s,1} = n_{\bm s^\prime}\pi +\phi_1 - \Phi_{\bm s^\prime, 0} = \Phi_{\bm s^\prime, 1}\label{eq:Phi_d+1_4}\\
    \cos\left(\frac{\theta_{L+1}}{2}\right)T_{\bm s,0} &= \cos\left(\frac{\theta_{L+1}}{2}\right)\prod_{\ell=2}^L \sin\left(\frac{\theta_\ell+{\bf d}\bm s_\ell \pi}{2}\right)\sin\left(\frac{\theta_1}{2}+\frac{(\bm  s_1+1)\pi}{4}\right) 
    \nonumber
    \\
    &= \prod_{\ell=2}^{L+1}\sin\left(\frac{\theta_\ell+{\bf d}\bm s^\prime_\ell \pi}{2}\right)\sin\left(\frac{\theta_1}{2}+\frac{(1-\bm s^\prime_1)\pi}{4}\right)= T_{\bm s^\prime,1},
\end{align}
where in the second equation of Eq.~\eqref{eq:Phi_d+1_4} we utilize the fact that $n_{\bm s^\prime} = L-\sum_{\ell=2}^{L+1}{\bf d}\bm s_\ell = L-\sum_{\ell=2}^{L}{\bf d}\bm s_\ell-1=n_{\bm s}$ given ${\bf d}\bm s^\prime = (1-{\bf d}\bm s_1, ({\bf d}\bm s)_{2:L}, 1)$.

Therefore, the weight in Eq.~\eqref{eq:psi_w} is proved, and combined with the qumode displacement, we prove that the output state from an $(L+1)$-layer circuit satisfies the form in Eq.~\eqref{eq:psi}.
\end{proof}

We now introduce the energy regularization.
Without losing generality, we consider the displacement on each step following a complex Gaussian distribution, $\beta_\ell \sim \calN_{E/L}^{\rm C}$, or equivalently $\Re{\beta_\ell},\Im{\beta_\ell}\sim \calN_{E/2L}$ with zero mean and variance $E/2L$ such that the ensemble-averaged energy of output state is $E$. To see that, from Eq.~\eqref{eq:psi}, the ensemble-averaged energy of the output state is
\begin{align}
    \mathbb{E}\left[\braket{\psi|m^\dagger m|\psi}\right] &= \sum_{a,a^\prime=0}^1\sum_{\bm s,\bm s^\prime} \mathbb{E}\left[w_{\bm s,a}w^*_{\bm s^\prime,a^\prime}\right]\braket{a^\prime|a}\mathbb{E}\left[e^{i(\chi_{\bm s,a}-\chi_{\bm s^\prime,a^\prime})}\braket{(-1)^{a^\prime}\bm s^\prime \cdot \bm \beta|m^\dagger m|(-1)^a \bm s \cdot \bm \beta}\right]\nonumber\\
    &= \sum_{a=0}^1\sum_{\bm s}\mathbb{E}\left[|w_{\bm s,a}|^2\right]\mathbb{E}\left[|\bm s\cdot \bm \beta|^2\right]\nonumber\\
    &= 2\cdot 2^{L-1}\cdot \frac{1}{2^L} E = E ,\label{eq:total_energy}
\end{align}
where we have applied the identity $\braket{a|a^\prime} = \delta_{a,a^\prime}$, $\mathbb{E}[w^*_{\bm s^\prime,a}w_{\bm s,a}] = \mathbb{E}[|w_{\bm s,a}|^2]\delta_{\bm s,\bm s^\prime} = 1/2^L \delta_{\bm s,\bm s^\prime}$ and
\begin{align}
    \mathbb{E}\left[|\bm s\cdot \bm \beta|^2\right] &=
    \mathbb{E}\left[\Re{\bm s\cdot \bm \beta}^2\right] + \mathbb{E}\left[\Im{\bm s\cdot \bm \beta}^2\right] = 2L\mathbb{E}\left[\Re{\beta_\ell}^2\right] = E,
\end{align}
utilizing the independence of each $\beta_\ell$ in $\bm \beta$ and the symmetry of real and imaginary parts of $\beta_\ell$.

\section{Methods for gradient evaluation}
\label{app:grad_methods}

In this section, we provide some preparation materials for the eventual evaluation of the variance of the gradient in Appendix~\ref{app:variance}.

As stated in the main text, the cost function in general can be written as $\calC = \Tr\left[{O}U\rho_0 U^\dagger\right]$ where $\rho_0$ is the initial state and ${O}$ is the observable. Consider the gradient with respect to the $k$th qubit rotation angle $\theta_k$, then the gradient becomes 
\begin{equation}
    \partial_{\theta_k}\calC = \partial_{\theta_k}\Tr\left[{O}U\rho_0 U^\dagger\right]
    = \frac{-i}{2}\Tr\left[{O} U_{\rm right}[G(\phi_k), U_{\rm left}\rho_0 U_{\rm left}^{\dagger}] U_{\rm right}^\dagger\right]
    = \frac{1}{2}\Big(\braket{O}_{k^{(+1)}} -\braket{O}_{k^{(-1)}} \Big),
\end{equation}
where we denote $G(\phi_k)\equiv \cos\phi_k \sigma^x + \sin\phi_k \sigma^y$, $U_{\rm left} = \prod_{j=1}^{k-1}U_{\rm ECD}(\beta_j)U_{\rm R}(\theta_j,\phi_j)$ as the circuit ahead of $k$th layer and $U_{\rm right}$ as the complement circuit from $k$th to $L$th layer. The last equation is obtained by applying the parameter-shift rule~\cite{mitarai2018quantum},
$
    [G(\phi_k), \rho] = i\left[U_{\rm R}\left(\frac{\pi}{2},\phi_k\right)\rho U_{\rm R}^\dagger\left(\frac{\pi}{2},\phi_k\right)- U_{\rm R}\left(-\frac{\pi}{2}, \phi_k\right) \rho U_{\rm R}^\dagger\left(-\frac{\pi}{2}, \phi_k\right)\right],
$
and $\braket{O}_{k^{(\pm 1)}}$ corresponds to expectation of $O$ with output state from the $L$-layer circuit, where $\theta_k$ is shifted as $\theta_k\rightarrow \theta_k\pm \pi/2$. For convenience, in the following discussion, we denote $w_{\bm s,a,k^{(\mu)}}$ as the weight defined in Eq.~\eqref{eq:psi_w}, where $\theta_k$ is shifted by $\mu \pi/2$ with $\mu=\pm 1$.

It is easy to check that
\begin{equation}
    \mathbb{E}\left[\partial_{\theta_k}\calC\right] = \frac{1}{2}\left(\mathbb{E}\left[\braket{O}_{k^{(+1)}}\right] - \mathbb{E}\left[\braket{O}_{k^{(-1)}}\right]\right) = 0,
\end{equation}
due to the fact that the ensemble average is performed over $\theta_j \in [0,2\pi)$.
From the definition of variance, the variance is reduced to 
\begin{equation}
    {\rm Var}\left[\partial_{\theta_k}\calC\right] = \mathbb{E}\left[(\partial_{\theta_k}\calC)^2\right]
    = \frac{1}{2}\left(\mathbb{E}\left[\braket{O}_{k^{(+1)}}^2\right] -\mathbb{E}\left[\braket{O}_{k^{(+1)}}\braket{O}_{k^{(-1)}}\right]\right),
    \label{eq:variance}
\end{equation}
where again we take $\mathbb{E}[\braket{O}_{k^{(+)}}^2] = \mathbb{E}[\braket{O}_{k^{(-)}}^2]$. 
For the state-preparation task being considered in this paper, operator $O = \ketbra{\phi}{\phi}_q\otimes \ketbra{\psi}{\psi}_m$, the two items in the variance above can be expanded via the output state representation in Eq.~\eqref{eq:psi} as
\begin{align}
     &\mathbb{E}\left[\braket{O}_{k^{(+1)}} \braket{O}_{k^{(\mu)}} \right]\nonumber\\
     &= \sum_{\substack{a,a^\prime,\\b,b^\prime=0}}^1 \braket{\phi|a}\braket{a^\prime|\phi}\braket{\phi|b}\braket{b^\prime|\phi}  \sum_{\substack{\bm s,\bm s^\prime,\\ \bm r,\bm r^\prime}}  \mathbb{E}\left[w_{\bm s,a, k^{(+1)}}w^*_{\bm s^\prime,a^\prime, k^{(+1)}}w_{\bm r,b, k^{(\mu)}}w^*_{\bm r^\prime, b^\prime, k^{(\mu)}}\right]
     \mathbb{E}\left[\braket{\psi|B_{\bm s,a}}\braket{B_{\bm s^\prime,a^\prime}|\psi}\braket{\psi|B_{\bm r,b}}\braket{B_{\bm r^\prime,b^\prime}|\psi}\right].
     \label{eq:Opm_general}
\end{align}

The exact calculation of variance above is hard due to the fact that $w_{\bm s,a,k^{(\mu)}}$ depends on ${\bf d}\bm s$ while $B_{\bm s,a}$ depends on $\bm s$, instead we consider the lower and upper bounds built from the following basic inequalities. For two sets of $N$ real numbers in increasing order $x_1\le x_2 \le \dots \le x_N$ and $y_1\le y_2 \le \dots \le y_N$, there is a well-known rearrangement inequality
\begin{equation}
    \sum_{j=1}^n x_j y_{n+1-j}
    \le
    \sum_{j=1}^n x_{\sigma(j)}y_j 
    \le 
    \sum_{j=1}^n x_j y_j,
    \label{eq:rearrangement_ineq}
\end{equation}
where $\sigma(j) \in \mathcal{S}_n$ is an arbitrary permutation of $n$ elements. In general $x_i$ can be either positive or negative, so we consider a relaxed version of the bounds only under the assumption that $y_j>0$ for all $j$, which decouples the index dependence between $x$ and $y$. The lower and upper bounds are
\begin{subequations}
\begin{align}
    \sum_{j=1}^n x_{\sigma(j)}y_j &\ge \left(\sum_{\substack{j=1\\ x_j\ge 0}}^n x_j\right)y_1 + \left(\sum_{\substack{j=1\\ x_j<0}}^n x_j\right)y_n ,\label{eq:lb}\\
    \sum_{j=1}^n x_{\sigma(j)}y_j &\le \left(\sum_{\substack{j=1\\ x_j\ge 0}}^n x_j\right)y_n + \left(\sum_{\substack{j=1\\ x_j<0}}^n x_j\right)y_1 .\label{eq:ub}
\end{align}    
\end{subequations}

\begin{proof}
We only show the proof for lower bound and the upper bound is a natural extension by swapping the minimum $y_1$ and maximum $y_n$.
The lower bound in Eq.~\eqref{eq:lb} is obvious from the first inequality in Eq.~\eqref{eq:rearrangement_ineq} as
\begin{align}
    \sum_{j=1} x_{\sigma(j)}y_j \ge \sum_{j=1}^n x_j y_{n+1-j} = \sum_{\substack{j=1\\ x_j>0}}^n x_j y_{n+1-j} + \sum_{\substack{j=1\\ x_j<0}}^n x_j y_{n+1-j} \ge \sum_{\substack{j=1\\ x_j>0}}^n x_j y_1 + \sum_{\substack{j=1\\ x_j<0}}^n x_j y_n.
\end{align}
\end{proof}

\subsection{Preliminary}

In this part, we introduce some prerequisite lemmas and propositions which are necessary in the evaluation of gradient variance in Appendix.~\ref{app:variance}. Recall the definition of the difference sign vector ${\bf d}\bm s_\ell = |\bm s_\ell - \bm s_{\ell-1}|/2$ in Eq.~\eqref{eq:sign_diff}, it has the following property.
\begin{lemma}
    The sum of all elements in difference sign vectors is always even, $\mathbb{P}\left(\sum_{\ell=1}^L {\bf d}\bm s_\ell\right) = +1$.
    \label{signdiff_parity}
\end{lemma}

\begin{proof}
As $(-1)^{{\bf d}\bm s_\ell}=\bm s_{\ell}\times \bm s_{\ell-1}$ for $\ell\ge 2$ and $(-1)^{{\bf d}\bm s_1}=\bm s_{1}\times \bm s_{L}$, we have
\begin{align}
\prod_{\ell=1}^L (-1)^{{\bf d}\bm s_\ell}= (-1)^{\sum_{\ell=1}^L {\bf d}\bm s_\ell}=\bm s_{1}\times \bm s_{L}\times \prod_{\ell=2}^L \left(\bm s_{\ell}\times \bm s_{\ell-1}\right)=\prod_{\ell=1}^L \bm s_{\ell}^2=1,
\end{align}
because $\bm s_\ell=\pm 1$.
Therefore, $\sum_{\ell=1}^L {\bf d}\bm s_\ell$ is even and $\mathbb{P}\left(\sum_{\ell=1}^L {\bf d}\bm s_\ell\right) = +1$
\end{proof}

A direct result from Lemma~\ref{signdiff_parity} is
about the number of $\phi_j$s in $\Phi_{\bm s,a}$ satisfies
\begin{corollary}
    The number of $\phi_\ell$s with nonzero coefficient in $\Phi_{\bm s,a}$ is $N_{\bm s,a} = n_{\bm s}+1-a-(-1)^a{\bf d}\bm s_1$, whose parity satisfies 
    $
        \mathbb{P}(N_{\bm s,a}) = \mathbb{P}(a)\mathbb{P}(L).
    $
    \label{phinum_parity}
\end{corollary}

\begin{proof}
As $N_{\bm s,a} = n_{\bm s}+1-a-(-1)^a{\bf d}\bm s_1$, we can see that the parity of $N_{\bm s,a}$ follows
\begin{align}
    \mathbb{P}\left(N_{\bm s,a}\right) = (-1)^{n_{\bm s}}(-1)^{(-1)^a {\bf d}\bm s_1}(-1)^{1-a}=(-1)^{L-1-\sum_{\ell=2}^L{\bf d}\bm s_\ell}(-1)^{-{\bf d}\bm s_1}(-1)^{1-a}
    = (-1)^{L-a}(-1)^{-\sum_{\ell=1}^L {\bf d}\bm s_\ell}
    = \mathbb{P}(a)\mathbb{P}(L),
    \label{eq:phinum_parity_proof}
\end{align}
where in the second equation we rewrite $(-1)^{(-1)^a{\bf d}\bm s_1}=(-1)^{-{\bf d}\bm s_1}$ as the sign of exponent does not change the value, and the last equality is obtained from the Lemma~\ref{signdiff_parity}.
\end{proof}

\begin{proposition}
    For two arbitrary different sign vector $\bm s, \bm r$ uniformly random sampled, the number of elements that ${\bf d}\bm s_\ell = {\bf d}\bm r_\ell$ is $N_T = L- \sum_{\ell=1}^L |{\bf d}\bm s_\ell - {\bf d}\bm r_\ell|$, and the distribution probability is $p(N_T) = \binom{L}{N_T}/(2^{L-1}-1)$ with constraint $0 \le N_T \le L-2$ and $\mathbb{P}(N_T) = \mathbb{P}(L)$.
    \label{Nt_def}
\end{proposition}

\begin{proof}
For arbitrary two sign vectors $\bm s\neq \bm r$, we have $0\le N_T \le L-2$ due to parity constraint in Lemma~\ref{signdiff_parity}. As it is an equal prior of ${\bf d}\bm s_\ell=0,1$, the distribution probability is $\binom{L}{N_T}/(2^{L-1}-1)$. Suppose the number of $\theta_\ell$ in $T_{\bm s,a}$ and $T_{\bm r,a}$ that are both in the form of $\sin(\theta_\ell/2)$ is $n_{\rm sin}$, and the number of $\theta_\ell$ that are both in the form of $\cos(\theta_\ell/2)$ is $N_{\bm s,a}-n_{\rm sin}$. On the other hand, the number of $\theta_\ell$ that are in the form of $\sin(\theta_\ell/2)$ in $T_{\bm s,a}$ but $\cos(\theta_\ell/2)$ in $T_{\bm r,a}$ is $N_T-n_{\rm sin}$, and the number for opposite correspondence is $N_{\bm r,a}-n_{\rm sin}$. The above statement is summarized in Table.~\ref{tab:Nt_proof}.
\begin{table}[]
    \centering
    \begin{tabular}{|c|c|c|}
    \hline
    \diagbox{$T_{\bm s,a}$}{$T_{\bm r,a}$} & $\sin(\theta_\ell/2)$ & $\cos(\theta_\ell/2)$  \\
    \hline
    $\sin(\theta_\ell/2)$  & $n_{\rm sin}$ & $N_{\bm s,a}-n_{\rm sin}$ \\
    \hline
    $\cos(\theta_\ell/2)$ & $N_{\bm r,a}-n_{\rm sin}$ & $N_T - n_{\rm sin}$\\
    \hline
    \end{tabular}
    \caption{The number of $\theta_\ell$ in $T_{\bm s,a}$ and $T_{\rm r,a}$ with corresponding form.}
    \label{tab:Nt_proof}
\end{table}
As the total number $\theta_\ell$ for both $T_{\bm s,a}$ and $T_{\bm r,a}$ is $L$, thus the summation of Table.~\ref{tab:Nt_proof} should equal to $L$, $N_{\bm s,a}+N_{\bm r,a}+N_T-2n_{\rm sin} = L$, and the parity relation is
\begin{equation}
    1=\mathbb{P}(L-N_{\bm s,a}-N_{\bm r,a}-N_T+2n_{\rm sin}) = \mathbb{P}(L)\mathbb{P}(N_{\bm s,a})\mathbb{P}(N_{\bm r,a})\mathbb{P}(N_T) = \mathbb{P}(L)\mathbb{P}(N_T),
\end{equation}
where in the second equality we utilize $\mathbb{P}(N_{\bm s,a}) = \mathbb{P}(N_{\bm r,a})$ from Corollary~\ref{phinum_parity}.
\end{proof}

\begin{proposition}
    For two arbitrary different sign vector $\bm s, \bm r$, the number of elements that $\bm s_\ell = \bm r_\ell$ is $\ell = L- \sum_{\ell=1}^L |\bm s_\ell - \bm r_\ell|/2$ with distribution $p(\ell) = \binom{L-1}{\ell-1}/(2^{L-1}-1)$ under the constraint $1\le \ell \le L-1$.
    \label{Nb_def}
\end{proposition}

\subsection{Ensemble average of four-fold product of weights}

In this part, we evaluate the ensemble average of the four-fold weight which occurs in Eq.~\eqref{eq:Opm_general}, and discuss its properties.
The four-fold weight product in general is
$w_{\bm s,a, k^{(+1)}}w^*_{\bm s^\prime,a^\prime,  k^{(+1)}}w_{\bm r,b, k^{(\mu)}}w^*_{\bm r^\prime, b^\prime, k^{(\mu)}}$ with $\mu = \pm 1$, and the ensemble average over all $\theta$ and $\phi$ is
\begin{align}
    &\mathbb{E}_{\bm \theta, \bm \phi}\left[w_{\bm s,a, k^{(+1)}}(\bm \theta, \bm \phi)w^*_{\bm s^\prime,a^\prime, k^{(+1)}}(\bm \theta, \bm \phi)w_{\bm r,b, k^{(\mu)}}(\bm \theta, \bm \phi)w^*_{\bm r^\prime, b^\prime, k^{(\mu)}}(\bm \theta, \bm \phi)\right]\nonumber\\
    &= \mathbb{E}_{\bm \phi}\left[e^{i\left(\Phi_{\bm s,a}(\bm \phi)-\Phi_{\bm s^\prime, a^\prime}(\bm \phi)+\Phi_{\bm r,b}(\bm \phi)-\Phi_{\bm r^\prime,b^\prime}(\bm \phi)\right)}\right]\mathbb{E}_{\bm \theta}\left[T_{\bm s,a, k^{(+1)}}(\bm \theta)T_{\bm s^\prime, a^\prime, k^{(+1)}}(\bm \theta)T_{\bm r,b, k^{(\mu)}}(\bm \theta)T_{\bm r^\prime,b^\prime, k^{(\mu)}}(\bm \theta)\right].
    \label{eq:w_average}
\end{align}
The average over $\phi$'s is simply zero if there is any $\phi_\ell$ left in the phase $\Phi_{\bm s,a}-\Phi_{\bm s^\prime,a^\prime}+\Phi_{\bm r,b}-\Phi_{\bm r^\prime,b^\prime}$, otherwise it can be $\pm 1$ depending on whether an even number of $\pi$ presents in the phase.
The average with respect to $\bm \theta$ is
\begin{proposition}
    The ensemble average over $\theta_\ell$s in the four-fold weight product is
    \begin{align}
        &\mathbb{E}_{\bm \theta}\left[T_{\bm s,a, k^{(+1)}}(\bm \theta)T_{\bm s^\prime, a^\prime, k^{(+1)}}(\bm \theta)T_{\bm r,b, k^{(\mu)}}(\bm \theta)T_{\bm r^\prime,b^\prime, k^{(\mu)}}(\bm \theta)\right]\nonumber\\
        &= \frac{1}{8^L}\prod_{\substack{\ell=2\\\ell\neq k}}^L \left(2\delta_{{\bf d}\bm s_\ell,{\bf d}\bm s^\prime_\ell}\delta_{{\bf d}\bm r_\ell,{\bf d}\bm r^\prime_\ell}+\cos\left[\frac{\pi}{2}({\bf d} \bm s_\ell+{\bf d} \bm s^\prime_\ell-{\bf d} \bm r_\ell-{\bf d} \bm r^\prime_\ell)\right]\right)
        \left(2\delta_{{\bf d}\bm s_k,{\bf d}\bm s^\prime_k}\delta_{{\bf d}\bm r_k,{\bf d}\bm r^\prime_k}+\mu \cos\left[\frac{\pi}{2}({\bf d} \bm s_k+{\bf d} \bm s^\prime_k-{\bf d} \bm r_k-{\bf d} \bm r^\prime_k)\right]\right)^{1-\delta_{k,1}}\nonumber\\
        & \quad 
        \times \left(2\delta_{\mathbb{P}(a)\bm s_1,\mathbb{P}(a^\prime)\bm s^\prime_1}\delta_{\mathbb{P}(b)\bm r_1,\mathbb{P}(b^\prime)\bm r^\prime_1}+(1+(\mu-1)\delta_{k,1}) \cos\left[\frac{\pi}{4}(\mathbb{P}(a)\bm s_1+\mathbb{P}(a^\prime)\bm s^\prime_1-\mathbb{P}(b)\bm r_1-\mathbb{P}(b^\prime)\bm r^\prime_1)\right]\right).
        \label{eq:T_average}
    \end{align}
\end{proposition}

\begin{proof}
As the $\theta_\ell$'s are independent each other, it is allowed to handle the average over each $\theta_\ell$ independently, and according to Eq.~\eqref{eq:T_def}, we only need to calculate the average over $\theta_1$ and $\theta_\ell$ with $\ell\ge 2$ separately. The ensemble average over $\theta_\ell$ with $\ell\ge 2$ in $\mathbb{E}\left[T_{\bm s,a,k^{(+1)}}T_{\bm s^\prime,a^\prime,k^{(+1)}}T_{\bm r,b,k^{(\mu)}}T_{\bm r^\prime,b^\prime,k^{(\mu)}}\right]$ is
\begin{align}
    &\mathbb{E}_{\bm \theta}\left[\sin\left(\frac{\theta_\ell+\pi{\bf d}\bm s_\ell}{2}+\frac{\pi}{4}\right)\sin\left(\frac{\theta_\ell+\pi{\bf d}\bm s^\prime_\ell}{2}+\frac{\pi}{4}\right)\sin\left(\frac{\theta_\ell+\pi{\bf d}\bm r_\ell}{2}+\frac{\mu\pi}{4}\right)\sin\left(\frac{\theta_\ell+\pi{\bf d}\bm r^\prime_\ell}{2}+\frac{\mu\pi}{4}\right)\right]\nonumber\\
    &= \frac{1}{8}\left(\cos\left[\frac{\pi}{2}({\bf d} \bm s_\ell-{\bf d} \bm s^\prime_\ell+{\bf d} \bm r_\ell-{\bf d} \bm r^\prime_\ell)\right]+\cos\left[\frac{\pi}{2}({\bf d} \bm s_\ell-{\bf d} \bm s^\prime_\ell-{\bf d} \bm r_\ell+{\bf d} \bm r^\prime_\ell)\right]+\mu\cos\left[\frac{\pi}{2}({\bf d} \bm s_\ell+{\bf d} \bm s^\prime_\ell-{\bf d} \bm r_\ell-{\bf d} \bm r^\prime_\ell)\right]\right)\nonumber\\
    &= \frac{1}{8}\left(2\cos\left[\frac{\pi({\bf d}\bm s_\ell-{\bf d}\bm s^\prime_\ell)}{2}\right]\cos\left[\frac{\pi({\bf d}\bm r_\ell-{\bf d}\bm r^\prime_\ell)}{2}\right]+\mu \cos\left[\frac{\pi}{2}({\bf d} \bm s_\ell+{\bf d} \bm s^\prime_\ell-{\bf d} \bm r_\ell-{\bf d} \bm r^\prime_\ell)\right]\right)\nonumber\\
    &= \frac{1}{8}\left(2\delta_{{\bf d}\bm s_\ell,{\bf d}\bm s^\prime_\ell}\delta_{{\bf d}\bm r_\ell,{\bf d}\bm r^\prime_\ell}+\mu \cos\left[\frac{\pi}{2}({\bf d} \bm s_\ell+{\bf d} \bm s^\prime_\ell-{\bf d} \bm r_\ell-{\bf d} \bm r^\prime_\ell)\right]\right),
\end{align}
where in the last line we apply the identity that $\cos(\pi(x-y)/2) = \delta_{x,y}$ for $x,y\in \{0,1\}$.
For $\theta_1$, the average is
\begin{align}
    &\mathbb{E}_{\bm \theta}\left[\sin\left(\frac{\theta_1}{2}+\frac{(\mathbb{P}(a)\bm s_1+2)\pi}{4}\right)\sin\left(\frac{\theta_1}{2}+\frac{(\mathbb{P}(a^\prime)\bm s^\prime_1+2)\pi}{4}\right)\sin\left(\frac{\theta_1}{2}+\frac{(\mathbb{P}(b)\bm r_1+\mu+1)\pi}{4}\right)\sin\left(\frac{\theta_1}{2}+\frac{(\mathbb{P}(b^\prime)\bm r^\prime_1+\mu+1)\pi}{4}\right)\right]\nonumber\\
    &= \frac{1}{8}\left(\cos\left[\frac{\pi}{4}\left( \mathbb{P}(a)\bm s_1-\mathbb{P}(a^\prime)\bm s^\prime_1+\mathbb{P}(b)\bm r_1-\mathbb{P}(b^\prime)\bm r^\prime_1\right)\right]+\cos\left[\frac{\pi}{4}\left(\mathbb{P}(a)\bm s_1-\mathbb{P}(a^\prime)\bm s^\prime_1-\mathbb{P}(b)\bm r_1+\mathbb{P}(b^\prime)\bm r^\prime_1\right)\right]\right.\nonumber\\
    &\left.\quad \quad +\mu\cos\left[\frac{\pi}{4}\left(\mathbb{P}(a)\bm s_1+\mathbb{P}(a^\prime)\bm s^\prime_1-\mathbb{P}(b)\bm r_1-\mathbb{P}(b^\prime)\bm r^\prime_1\right)\right]\right)\nonumber\\
    &= \frac{1}{8}\left(2\cos\left[\frac{\pi(\mathbb{P}(a)\bm s_1-\mathbb{P}(a^\prime)\bm s^\prime_1)}{4}\right]\cos\left[\frac{\pi(\mathbb{P}(b)\bm r_1-\mathbb{P}(b^\prime)\bm r^\prime_1)}{4}\right]+\mu \cos\left[\frac{\pi}{4}(\mathbb{P}(a)\bm s_1+\mathbb{P}(a^\prime)\bm s^\prime_1-\mathbb{P}(b)\bm r_1-\mathbb{P}(b^\prime)\bm r^\prime_1)\right]\right)\nonumber\\
    &= \frac{1}{8}\left(2\delta_{\mathbb{P}(a)\bm s_1,\mathbb{P}(a^\prime)\bm s^\prime_1}\delta_{\mathbb{P}(b)\bm r_1,\mathbb{P}(b^\prime)\bm r^\prime_1}+\mu \cos\left[\frac{\pi}{4}(\mathbb{P}(a)\bm s_1+\mathbb{P}(a^\prime)\bm s^\prime_1-\mathbb{P}(b)\bm r_1-\mathbb{P}(b^\prime)\bm r^\prime_1)\right]\right),
\end{align}
where in last line we apply $\cos(\pi(x-y)/4) = \delta_{x,y}$ for $x,y=\pm 1$.
Note that the average over all $\theta_\ell$'s depends on the choice of $k$, more specifically, whether $k> 1$ or not. One can write out the ensemble average for those two separately, and figure out that the ensemble average over all $\theta_\ell$'s can be unified as
\begin{align}
    &\mathbb{E}_{\bm \theta}\left[T_{\bm s,a, k^{(+1)}}(\bm \theta)T_{\bm s^\prime, a^\prime, k^{(+1)}}(\bm \theta)T_{\bm r,b, k^{(\mu)}}(\bm \theta)T_{\bm r^\prime,b^\prime, k^{(\mu)}}(\bm \theta)\right]\nonumber\\
    &= \frac{1}{8^L}\prod_{\substack{\ell=2\\\ell\neq k}}^L \left(2\delta_{{\bf d}\bm s_\ell,{\bf d}\bm s^\prime_\ell}\delta_{{\bf d}\bm r_\ell,{\bf d}\bm r^\prime_\ell}+\cos\left[\frac{\pi}{2}({\bf d} \bm s_\ell+{\bf d} \bm s^\prime_\ell-{\bf d} \bm r_\ell-{\bf d} \bm r^\prime_\ell)\right]\right)
    \left(2\delta_{{\bf d}\bm s_k,{\bf d}\bm s^\prime_k}\delta_{{\bf d}\bm r_k,{\bf d}\bm r^\prime_k}+\mu \cos\left[\frac{\pi}{2}({\bf d} \bm s_k+{\bf d} \bm s^\prime_k-{\bf d} \bm r_k-{\bf d} \bm r^\prime_k)\right]\right)^{1-\delta_{k,1}}\nonumber\\
    &\quad \quad \times \left(2\delta_{\mathbb{P}(a)\bm s_1,\mathbb{P}(a^\prime)\bm s^\prime_1}\delta_{\mathbb{P}(b)\bm r_1,\mathbb{P}(b^\prime)\bm r^\prime_1}+(1+(\mu-1)\delta_{k,1}) \cos\left[\frac{\pi}{4}(\mathbb{P}(a)\bm s_1+\mathbb{P}(a^\prime)\bm s^\prime_1-\mathbb{P}(b)\bm r_1-\mathbb{P}(b^\prime)\bm r^\prime_1)\right]\right).
\end{align}
\end{proof}

A direct corrollary drawn from Eq.~\eqref{eq:T_average} is the following
\begin{corollary}
    If three of the sign vectors $\bm s, \bm s^\prime, \bm r, \bm r^\prime$ are equal, then the ensemble average in Eq.~\eqref{eq:w_average} is zero.
    \label{weight_average_zero_sign}
\end{corollary}

\begin{proof} If three of the sign vectors are the same, then one can check that $2\delta_{{\bf d}\bm s_j,{\bf d}\bm s^\prime_j}\delta_{{\bf d}\bm r_j,{\bf d}\bm r^\prime_j} = 0$ and $\cos\left[\frac{\pi}{2}({\bf d} \bm s_j+{\bf d} \bm s^\prime_j-{\bf d} \bm r_j-{\bf d} \bm r^\prime_j)\right] = \cos(\pm \pi/2) = 0$, which makes the average in Eq.~\eqref{eq:T_average} be zero.
\end{proof}

\section{Variance of gradient in state preparation of single-mode CV state}
\label{app:variance}

In this section, we provide the detailed proof for the bounds of variance of gradient (Ineqs.~\eqref{eq:LB_main} and~\eqref{eq:UB_main} with $M=1$) in the preparation of a single-mode CV state $\ket{\psi}$ with target energy $\braket{m^\dagger m} = E_t$, while the target state of the ancilla qubit is simply chosen as $\ket{0}_q$. 

With the qubit target state $\ket{0}_q$, the expansion of items in variance of gradient (see Eq.~\eqref{eq:Opm_general}) is reduced to 
\begin{align}
     \mathbb{E}\left[\braket{O}_{k^{(+1)}} \braket{O}_{k^{(\mu)}} \right] &= \sum_{\bm s,\bm s^\prime,\bm r,\bm r^\prime} \mathbb{E}_{\bm \theta, \bm \phi}\left[w_{\bm s,k^{(+1)}}w^*_{\bm s^\prime, k^{(+1)}}w_{\bm r,k^{(\mu)}}w^*_{\bm r^\prime, k^{(\mu)}}\right] \mathbb{E}_{\bm \beta}\left[\braket{\psi|B_{\bm s}}\braket{B_{\bm s^\prime}|\psi}\braket{\psi|B_{\bm r}}\braket{B_{\bm r^\prime}|\psi}\right],
     \label{eq:Opm}
\end{align}
where we omit the subscript related to qubit state $a,a^\prime,b,b^\prime$ in the expression for simplicity since $a=a^\prime=b=b^\prime=0$. The summation over $\bm s,\bm s^\prime, \bm r, \bm r^\prime$ can be nonzero in the following four cases: (i) $\bm s=\bm s^\prime=\bm r=\bm r^\prime$; (ii) $\bm s-\bm s^\prime=\bm r-\bm r^\prime=\bm 0$ but $\bm s\neq \bm r$; (iii)  $\bm s-\bm r^\prime=\bm r-\bm s^\prime=\bm 0$ but $\bm s\neq \bm r$; (iv) $\bm s,\bm s^\prime,\bm r,\bm r^\prime {\rm unequal}$. However $\bm s-\bm r=\bm s^\prime-\bm r^\prime=\bm 0$ with $\bm s\neq \bm s^\prime$ does {\it not} contribute as the average over $\bm \phi$ becomes $\mathbb{E}[e^{2i(\Phi_{\bm s,a}-\Phi_{\bm s^\prime,a})}] = 0$ with $\bm s\neq \bm s^\prime$.

From the above analysis, the variance of the gradient  shown in Eq.~\eqref{eq:variance} becomes
\begin{align}
    &{\rm Var}\left[\partial_{\theta_k}\calC\right]=\frac{1}{2}\left(\mathbb{E}\left[\braket{O}_{k^{(+1)}}^2\right] -\mathbb{E}\left[\braket{O}_{k^{(+1)}} \braket{O}_{k^{(-1)}}\right]\right)\nonumber\\
    &= \frac{1}{2}\left(\sum_{\bm s} \Delta_\mu\left\{\mathbb{E}_{\bm \theta,\bm \phi}\left[|w_{\bm s,k^{(+1)}}|^2|w_{\bm s, k^{(\mu)}}|^2\right]\right\} \mathbb{E}_{\bm \beta}\left[|\braket{\psi|B_{\bm s}}|^4\right]\right.\nonumber\\
    & \quad \quad
    + \sum_{\bm s\neq \bm r}\Delta_\mu\left\{\mathbb{E}_{\bm \theta,\bm \phi}\left[|w_{\bm s,k^{(+1)}}|^2 |w_{\bm r, k^{(\mu)}}|^2\right] + \mathbb{E}_{\bm \theta,\bm \phi}\left[ w_{\bm s, k^{(+1)}}w^*_{\bm r,k^{(+1)}}w_{\bm r,k^{(\mu)}}w^*_{\bm s, k^{(\mu)}}\right]\right\}
    \mathbb{E}_{\bm \beta}\left[|\braket{\psi|B_{\bm s}}|^2|\braket{\psi|B_{\bm r}}|^2\right]\nonumber\\
    & \quad \quad
    \left.+\sum_{\substack{\bm s,\bm s^\prime,\bm r, \bm r^\prime\\ \text{unequal}}}\Delta_\mu\left\{\mathbb{E}_{\bm \theta, \bm \phi}\left[w_{\bm s,k^{(+1)}}w^*_{\bm s^\prime, k^{(+1)}}w_{\bm r,k^{(\mu)}}w^*_{\bm r^\prime, k^{(\mu)}}\right]\right\}\mathbb{E}_{\bm \beta}\left[\braket{\psi|B_{\bm s}}\braket{B_{\bm s^\prime}|\psi}\braket{\psi|B_{\bm r}}\braket{B_{\bm r^\prime}|\psi}\right]\right)\\
    &\equiv \frac{1}{2}\left(S_1+S_2+S_3\right).
    \label{eq:variance_single_mode}
\end{align}
The notation $\Delta_\mu \{X\} \equiv X|_{\mu = 1} - X|_{\mu = -1}$ represents the difference of quantity $X$ with $\mu = 1$ and $\mu = -1$, and the difference is only evaluated on the average over $\theta_\ell$'s. We introduce $\{S_1,S_2,S_3\}$ in the last line to denote the three summations in the large parenthesis above for convenience. In the following, we will evaluate them term by term. 

For $S_1$, the ensemble average of weight from Eq.~\eqref{eq:T_average} is
\begin{align}
    &\mathbb{E}\left[|w_{\bm s, k^{(+1)}}|^2|w_{\bm s, k^{(\mu)}}|^2\right] = \mathbb{E}\left[T_{\bm s,k^{(+1)}}^2T_{\bm s, k^{(\mu)}}^2\right]
    = \frac{1}{8^L}3^{L-2+\delta_{k,1}}(2+\mu)^{1-\delta_{k,1}}\left(3+(\mu-1)\delta_{k,1}\right) = \frac{3^{L-1}}{8^L}(2+\mu).
\end{align}
The difference with $\mu = \pm 1$ is
\begin{equation}
    \Delta_\mu \left\{\mathbb{E}\left[|w_{\bm s, k^{(+1)}}|^2|w_{\bm s, k^{(\mu)}}|^2\right]\right\} = \frac{2\cdot 3^{L-1}}{8^L}.
    \label{eq:S1_wdiff}
\end{equation}
For the displacement part, recall that $\ket{B_{\bm s}}=e^{i\chi_{\bm s}}\ket{\bm s\cdot \bm \beta}$ where $\chi_{\bm s}$ is a pure phase and $\bm s\cdot \bm \beta$ is a complex Gaussian variable with zero mean and variance $E$ as $(\bm s\cdot \bm \beta) \sim \calN^{\rm C}_E$, and the average over $\bm \beta$ becomes
\begin{align}
    \mathbb{E}\left[|\braket{\psi|B_{\bm s}}|^4\right] = \mathbb{E}\left[|\braket{\psi|\bm s \cdot \bm \beta}|^4\right] = \mathbb{E}_{\alpha\sim \calN_E^{\rm C}}\left[|\braket{\psi|\alpha}|^4\right] \equiv C_1,
    \label{eq:C1_def}
\end{align}
where in the second equation $\ket{\alpha}$ is a coherent state with displacement $\alpha\sim \calN_{E}^{\rm C}$, and we denote the average to be correlator $C_1$.
As Eqs~\eqref{eq:S1_wdiff} and~\eqref{eq:C1_def} are both independent of $\bm s$, the summation $S_1$ is simply
\begin{align}
    S_1 &= \sum_{\bm s} \Delta_\mu \left\{\mathbb{E}\left[|w_{\bm s, k^{(+1)}}|^2|w_{\bm s, k^{(\mu)}}|^2\right]\right\}  \mathbb{E}\left[|\braket{\psi|B_{\bm s}}|^4\right]= 2^{L-1} \frac{2\cdot 3^{L-1}}{8^L} C_1 = \frac{3^{L-1}}{4^L} C_1.
    \label{eq:S1}
\end{align}

The average over qubit rotation angles $\bm \phi, \bm \theta$ in $S_2$ utilizing Eq.~\ref{eq:T_average} is
\begin{align}
    &\mathbb{E}\left[|w_{\bm s, k^{(+1)}}|^2 |w_{\bm r,k^{(\mu)}}|^2\right] + \mathbb{E}\left[ w_{\bm s, k^{(+1)}}w^*_{\bm r,k^{(+1)}}w_{\bm r,k^{(\mu)}}w^*_{\bm s, k^{(\mu)}}\right]= \mathbb{E}\left[T_{\bm s,k^{(+1)}}^2T_{\bm r,k^{(\mu)}}^2\right] + \mathbb{E}\left[T_{\bm s,k^{(+1)}}T_{\bm r,k^{(+1)}}T_{\bm r,k^{(\mu)}}T_{\bm s,k^{(\mu)}}\right]\nonumber\\
    &= \frac{1}{8^L}\prod_{\substack{\ell=2\\ \ell\neq k}}^L\left(1+2\delta_{{\bf d}\bm s_\ell,{\bf d}\bm r_\ell}\right)\left(2-\mu+2\mu\delta_{{\bf d}\bm s_k,{\bf d}\bm r_k}\right)^{1-\delta_{k,1}}\left[2+(1+(\mu-1)\delta_{k,1})\left(2\delta_{{\bf d}\bm s_1,{\bf d}\bm r_1}-1\right)\right]\nonumber\\
    & \quad 
    + \frac{1}{8^L}\prod_{\substack{\ell=2\\\ell\neq k}}^L\left(1+2\delta_{{\bf d}\bm s_\ell,{\bf d}\bm r_\ell}\right)\left(2\delta_{{\bf d}\bm s_k,{\bf d}\bm r_k}+\mu\right)^{1-\delta_{k,1}}\left[2\delta_{{\bf d}\bm s_1,{\bf d}\bm r_1}+(1+(\mu-1)\delta_{k,1})\right]\nonumber\\
    &= \frac{1}{8^L}\prod_{\substack{\ell=1\\ \ell\neq k}}^L \left(1+2\delta_{{\bf d}\bm s_\ell,{\bf d}\bm r_\ell}\right)\left(2-\mu+2\mu\delta_{{\bf d}\bm s_k,{\bf d}\bm r_k}\right) + \frac{1}{8^L}\prod_{\substack{\ell=1\\ \ell\neq k}}^L \left(1+2\delta_{{\bf d}\bm s_\ell,{\bf d}\bm r_\ell}\right)\left(2\delta_{{\bf d}\bm s_k,{\bf d}\bm r_k}+\mu\right)\nonumber\\
    &= \frac{2}{8^L}\prod_{\substack{\ell=1\\ \ell\neq k}}^L\left(1+2\delta_{{\bf d}\bm s_\ell,{\bf d}\bm r_\ell}\right)\left(1+(1+\mu)\delta_{{\bf d}\bm s_k,{\bf d}\bm r_k}\right).
\end{align}
Thus, the difference with respect to $\mu$ is 
\begin{equation}
    \Delta_\mu\left\{\mathbb{E}\left[|w_{\bm s,k^{(+1)}}|^2 |w_{\bm r,k^{(\mu)}}|^2\right] + \mathbb{E}\left[ w_{\bm s, k^{(+1)}}w^*_{\bm r, k^{(+1)}}w_{\bm r,k^{(\mu)}}w^*_{\bm s, k^{(\mu)}}\right]\right\} = \frac{4}{8^L}\prod_{\ell=1,\ell\neq k}^L(1+2\delta_{{\bf d}\bm s_\ell,{\bf d}\bm r_\ell})\delta_{{\bf d}\bm s_k,{\bf d}\bm r_k}.
    \label{eq:S2_wdiff}
\end{equation}
Note that the last delta function above is nonzero only if ${\bf d}\bm s_k = {\bf d}\bm r_k$, and as $N_T$ is the total number of elements that satisfy ${\bf d}\bm s_\ell = {\bf d}\bm r_\ell$ from Proposition~\ref{Nt_def}, there are $N_T-1$ elements that satisfy ${\bf d}\bm s_\ell = {\bf d}\bm r_\ell$ with $\ell\neq k$, making Eq.~\eqref{eq:S2_wdiff} equal to $4\cdot 3^{N_T-1}/8^L$ with probability $\binom{L-1}{N_T-1}/(2^{L-1}-1)$ from Proposition~\ref{Nt_def}.
The summation over all $\bm s\neq \bm r$ of the difference is 
\begin{align}
    &\sum_{\bm s \neq \bm r} \Delta_{\mu}\left\{\mathbb{E}\left[|w_{\bm s, k^{(+1)}}|^2 |w_{\bm r,k^{(\mu)}}|^2\right] + \mathbb{E}\left[ w_{\bm s,k^{(+1)}}w^*_{\bm r,k^{(+1)}}w_{\bm r, k^{(\mu)}}w^*_{\bm s,k^{(\mu)}}\right]\right\}\nonumber\\
    &= 2^{L-1}\left(2^{L-1}-1\right)\sum_{\substack{N_T = 0\\ \mathbb{P}(N_T)=\mathbb{P}(L)}}^{L-2} \frac{4\cdot 3^{N_T-1}}{8^L} \frac{\binom{L-1}{N_T-1}}{2^{L-1}-1}= \frac{1}{4}-\frac{2\cdot 3^{L-1}-2^{L-1}}{4^L}.
    \label{eq:S2_wsum}
\end{align}
For the average over displacement $\bm \beta$ in $S_2$, note that for any two $\bm s\cdot \bm \beta, \bm r \cdot \bm \beta$ with $\bm s\neq \bm r$, we can always write them as $\bm s \cdot \bm \beta = \alpha_z + \alpha_{1-z}$ and $\bm r \cdot \bm \beta = \alpha_z - \alpha_{1-z}$ where $\alpha_z, \alpha_{1-z}$ are complex Gaussian variables obeying distributions $\calN_{z E}^{\rm C},\calN_{(1-z) E}^{\rm C}$. Note that $z=\ell/L$ with $\ell$ being an integer in the range of $[1,L-1]$. The displacement average can therefore be simplified to
\begin{align}
    &\mathbb{E}_{\bm \beta}\left[|\braket{\psi|B_{\bm s}}|^2|\braket{\psi|B_{\bm r}}|^2\right] = \mathbb{E}_{\bm \beta}\left[|\braket{\psi|\bm s\cdot \bm \beta}|^2|\braket{\psi|\bm r \cdot \bm \beta}|^2\right] = \mathbb{E}_{\alpha_y\sim \calN_{yE}^{\rm C}}\left[\prod_{h=0}^1|\braket{\psi|\alpha_z+(-1)^h\alpha_{1-z}}|^2 \right]\equiv C_2(z),
    \label{eq:C2_def}
\end{align}
where in the second equation we take an average over independent variables $\alpha_z$ and $\alpha_{1-z}$ with $\alpha_z \sim \calN_{zE}^{\rm C}$ and $\alpha_{1-z}\sim \calN_{(1-z)E}^{\rm C}$, denoted as $\alpha_y\sim \calN_{yE}^{\rm C}$ for simplicity. With Eqs.~\eqref{eq:S2_wsum} and~\eqref{eq:C2_def}, we can have bounds for $S_2$ as
\begin{align}
    S_2 &\ge \sum_{\bm s\neq \bm r}\Delta_\mu\left\{\mathbb{E}\left[|w_{\bm s, k^{(+1)}}|^2 |w_{\bm r, k^{(\mu)}}|^2\right] + \mathbb{E}\left[ w_{\bm s, k^{(+1)}}w^*_{\bm r, k^{(+1)}}w_{\bm r, k^{(\mu)}}w^*_{\bm s, k^{(\mu)}}\right]\right\} \min\mathbb{E}\left[|\braket{\psi|B_{\bm s}}|^2|\braket{\psi|B_{\bm r}}|^2\right]\nonumber\\
    &=\left(\frac{1}{4}-\frac{2\cdot 3^{L-1}-2^{L-1}}{4^L}\right)\min_\ell C_2\left(\frac{\ell}{L}\right) \label{eq:S2_lb},\\
     S_2 &\le \sum_{\bm s\neq \bm r}\Delta_\mu\left\{\mathbb{E}\left[|w_{\bm s,k^{(+1)}}|^2 |w_{\bm r,k^{(\mu)}}|^2\right] + \mathbb{E}\left[ w_{\bm s,k^{(+1)}}w^*_{\bm r,k^{(+1)}}w_{\bm r,k^{(\mu)}}w^*_{\bm s,k^{(\mu)}}\right]\right\} \max \mathbb{E}\left[|\braket{\psi|B_{\bm s}}|^2|\braket{\psi|B_{\bm r}}|^2\right]\nonumber\\
    &=\left(\frac{1}{4}-\frac{2\cdot 3^{L-1}-2^{L-1}}{4^L}\right)\max_\ell C_2\left(\frac{\ell}{L}\right), \label{eq:S2_ub}
\end{align}
where the minimization and maximization are taken over all integers $\ell \in [1, L-1]$.

The summation $S_3$ involves four different sign vectors, although the total number of summation is about $16^{L-1}$, a large amount of them can be excluded by the averaging over phase. To have $\mathbb{E}\left[w_{\bm s,k^{(+1)}}w^*_{\bm s^\prime, k^{(+1)}}w_{\bm r,k^{(\mu)}}w^*_{\bm r^\prime, k^{(\mu)}}\right]$ nonzero, $\Phi_{\bm s}-\Phi_{\bm s^\prime}+\Phi_{\bm r}-\Phi_{\bm r^\prime}$ can only be a constant that is independent of any $\phi_\ell$. According to Eq.~\eqref{eq:Phi_def}, it requires the coefficient for each $\phi_\ell$ to be zero,
\begin{equation}
    \left({\bf d}\bm s_\ell-1\right)\bm s_\ell-\left({\bf d}\bm s_\ell^\prime-1\right)\bm s_\ell^\prime +\left({\bf d}\bm r_\ell-1\right)\bm r_\ell-\left({\bf d}\bm r^\prime_\ell-1\right)\bm r_\ell^\prime = 0.
\end{equation}
Note that ${\bf d}\bm s_\ell = |\bm s_\ell - \bm s_{\ell-1}|/2=(1-\bm s_\ell \bm s_{\ell-1})/2$, then the above constraint can reduce to
\begin{align}
    &\bm s_\ell(1-\bm s_\ell \bm s_{\ell-1}) - \bm s_\ell^\prime(1-\bm s^\prime_\ell \bm s^\prime_{\ell-1}) + \bm r_\ell(1-\bm r_\ell \bm r_{\ell-1}) - \bm r^\prime_\ell(1-\bm r^\prime_\ell \bm r^\prime_{\ell-1}) = 2(\bm s_\ell -\bm s^\prime_\ell + \bm r_\ell - \bm r_\ell^\prime) \nonumber\\
    &\Rightarrow \bm s_\ell - \bm s_\ell^\prime + \bm r_\ell - \bm r_\ell^\prime = -( \bm s_{\ell-1} - \bm s_{\ell-1}^\prime + \bm r_{\ell-1} - \bm r_{\ell-1}^\prime),
\end{align}
where we have used $\bm s_\ell^2=\bm s_\ell^{\prime 2}=\bm r_\ell^2=\bm r_\ell^{\prime 2}=1$ to get the last line.
As $\bm s_L - \bm s_L^\prime + \bm r_L - \bm r_L^\prime = 0$, the constraint above becomes
\begin{align}
    \bm s_\ell + \bm r_\ell - \bm s^\prime_\ell - \bm r^\prime_\ell = 0,\quad \forall \ell \in [1, L]\cap \mathbb{N}.
    \label{eq:S3_constraint}
\end{align}
In Table.~\ref{tab:S3_table}, we list all the combination of $\bm s_\ell, \bm s^\prime_\ell, \bm r_\ell, \bm r^\prime_\ell$ with $1\le \ell \le L-1$ up to a global reverse of signs and test if the constraint Eq.~\eqref{eq:S3_constraint} is satisfied. A global reverse of all signs of $\bm s_\ell, \bm s^\prime_\ell, \bm r_\ell, \bm r^\prime_\ell$ for instance only $\bm s_\ell = -1$ and only $\bm s_\ell = +1$ leads to same satisfiability result.
\begin{table}[t]
    \centering
    \begin{tabular}{|c|c|c|c|c|}
        \hline
        $\bm s_\ell$ & $\bm s_\ell^\prime$ & $\bm r_\ell$ & $\bm r_\ell^\prime$ & Is Eq.~\eqref{eq:S3_constraint} satisfied\\
        \hline
        $-1$ & $-1$ & $-1$ & $-1$ & Yes\\
        \hline
        $-1$ & $-1$ & $-1$ & $+1$ & No\\
        \hline
        $-1$ & $-1$ & $+1$ & $-1$ & No\\
        \hline
        $-1$ & $+1$ & $-1$ & $-1$ & No\\
        \hline
        $+1$ & $-1$ & $-1$ & $-1$ & No\\
        \hline
        $-1$ & $-1$ & $+1$ & $+1$ & Yes\\
        \hline
        $-1$ & $+1$ & $-1$ & $+1$ & No\\
        \hline
        $+1$ & $-1$ & $-1$ & $+1$ & Yes\\
        \hline
    \end{tabular}
    \caption{A satisfiability test of Eq.~\eqref{eq:S3_constraint} for all possible combination of $\bm s_\ell, \bm s_\ell^\prime, \bm r_\ell, \bm r_\ell^\prime$ with $1\le \ell \le L-1$ up to a global reverse of signs.}
    \label{tab:S3_table}
\end{table}
Summarized from Table.~\ref{tab:S3_table}, there are only three allowed combination of $\bm s_\ell, \bm s^\prime_\ell, \bm r_\ell, \bm r^\prime_\ell$, therefore the partition of $\bm s\cdot \bm \beta$, $\bm s^\prime\cdot\bm \beta$, $\bm r\cdot\bm \beta$, $\bm r^\prime\cdot\bm \beta$ is
\begin{align}
    \bm s\cdot \bm \beta &= \alpha_{z} + \alpha_{\tilde{z}} + \alpha_{1-z-\tilde{z}}\nonumber\\
    \bm s^\prime \cdot \bm \beta &= \alpha_{z} + \alpha_{\tilde{z}} - \alpha_{1-z-\tilde{z}}\nonumber\\
    \bm r\cdot \bm \beta &= \alpha_{z} - \alpha_{\tilde{z}} - \alpha_{1-z-\tilde{z}}\nonumber\\
    \bm r^\prime\cdot \bm \beta &= \alpha_{z} - \alpha_{\tilde{z}} + \alpha_{1-z-\tilde{z}},
    \label{eq:S3_Bstructure}
\end{align}
where $\alpha_{z}\sim \calN_{z E}^{\rm C}$ is Gaussian variable, and similar for $\alpha_{\tilde{z}}$ and $\alpha_{1-z-\tilde{z}}$. As $\bm s,\bm r,\bm s^\prime,\bm r^\prime$ are different, the $z,\tilde{z}$ are limited to $z = \ell_1/L, \tilde{z} = \ell_2/L$ with $\ell_1,\ell_2\in [1,L-2]\cap \mathbb{N}$ and $\ell_1+\ell_2\le L-1$.

The average over displacements in $S_3$ is
\begin{align}
    &\mathbb{E}\left[\braket{\psi|B_{\bm s}}\braket{B_{\bm s^\prime}|\psi}\braket{\psi|B_{\bm r}}\braket{B_{\bm r^\prime}|\psi}\right]\nonumber\\
    &= \mathbb{E}_{\alpha_y\sim \calN_{yE}^{\rm C}}\left[e^{i(\chi_{\bm s}-\chi_{\bm s^\prime}+\chi_{\bm r}-\chi_{\bm r^\prime})}\prod_{a=0}^1 \braket{\psi|\alpha_{z}+(-1)^a \alpha_{\tilde{z}}+(-1)^a \alpha_{1-z-\tilde{z}}}\braket{\alpha_{z}+(-1)^a\alpha_{\tilde{z}}-(-1)^a\alpha_{1-z-\tilde{z}}|\psi}\right]\\
    &\le \mathbb{E}_{\alpha_y\sim \calN_{yE}^{\rm C}}\left[\prod_{a=0}^1 |\braket{\psi|\alpha_{z}+(-1)^a \alpha_{\tilde{z}}+(-1)^a \alpha_{1-z-\tilde{z}}}||\braket{\alpha_{z}+(-1)^a\alpha_{\tilde{z}}-(-1)^a\alpha_{1-z-\tilde{z}}|\psi}|\right]\equiv C_3(z,\tilde{z}),
    \label{eq:C3_def}
\end{align}
where we upper bound it by $\mathbb{E}[x]\le \mathbb{E}[|x|]$ in the last line.
In each class of states under consideration, we will show that $C_3(z, \tilde{z})$ leads to higher order terms compared to Eq.~\eqref{eq:C1_def} and~\eqref{eq:C2_def} and thus can be neglected in the asymptotic region of $E$ in later discussion.

To conclude, we have the lower and upper bounds for variance of gradient in Eq.~\eqref{eq:variance_single_mode} as
\begin{align}
    {\rm Var}\left[\partial_{\theta_k}\calC(
    \bm x)\right] &= \frac{1}{2}\left[\frac{3^{L-1}}{4^L} C_1 + \left(\frac{1}{4}-\frac{2\cdot 3^{L-1}-2^{L-1}}{4^L}\right)\min_{\ell}C_2\left(\frac{\ell}{L}\right)\right] + \mathcal{O}\left(C_2\right)\nonumber\\
    &\ge \frac{1}{2}\left[\frac{3^{L-1}}{4^L} C_1 + \left(\frac{1}{4}-\frac{3^{L}}{4^L}\right)\min_{\ell}C_2\left(\frac{\ell}{L}\right)\right] + \mathcal{O}\left(C_2\right),
    \label{eq:var_lb}
\end{align}
and
\begin{align}
     {\rm Var}\left[\partial_{\theta_k}\calC(
    \bm x)\right] &= \frac{1}{2}\left[\frac{3^{L-1}}{4^L} C_1 + \left(\frac{1}{4}-\frac{2\cdot 3^{L-1}-2^{L-1}}{4^L}\right)\max_{\ell}C_2\left(\frac{\ell}{L}\right)\right] + \mathcal{O}\left(C_2\right)\nonumber\\
    &\le \frac{1}{2}\left[\frac{3^{L-1}}{4^L} C_1 + \left(\frac{1}{4}+\frac{2^{L-1}}{4^L}\right)\max_{\ell}C_2\left(\frac{\ell}{L}\right)\right] + \mathcal{O}\left(C_2\right),
    \label{eq:var_ub}
\end{align}
where the minimization and maximization are over integers $1\le \ell \le L-1$.
The exact expression of correlators $C_1, C_2$ in the above bounds depend on the specific target state, and in the following, we evaluate those correlators with different target state $\ket{\psi}_m$ to provide an insight of their asymptotic behavior, and thus the behavior of the gradient. In subsection~\ref{app:var_gaussian}, we consider single-mode Gaussian states; In subsection~\ref{app:var_fock}, we consider Fock number states.

\subsection{Single-mode Gaussian states}
\label{app:var_gaussian}

Suppose the target qumode state is an arbitrary Gaussian (pure) state, as the correlators $C_1$ and $C_2$ only depend on the fidelity between target state $\ket{\psi}_m$ and a coherent states, analytical evaluation is possible thanks to Refs.~\cite{marian2012uhlmann,spedalieri2012limit,banchi2015quantum}. For a brief introduction to Gaussian states, please refer to Appendix~\ref{app:Gaussian_intro}.

As one can see from Eqs.~\eqref{eq:C1_def} and~\eqref{eq:C2_def}, the correlators we need to evaluate only depends on the fidelity between target 
state $\ket{\psi}_m$ and coherent states, we first evaluate the fidelity between Gaussian state in Eq.~\eqref{eq:one_mode_gaussian} and coherent state $\ket{\alpha}$ from Eq.~\eqref{eq:Fidelity_gaussian}.
\begin{lemma}
    The fidelity between an arbitrary one-mode Gaussian state $\ket{\psi}_m = D(\gamma)R(\tau)S(\zeta)\ket{0}_m$ and a coherent state $\ket{\alpha}_m = D(\alpha)\ket{0}_m$ is
    \begin{equation}
        F(\psi, \ket{\alpha}) = \sech(\zeta)e^{-(1+\kappa_1)(\Re{\gamma}-\Re{\alpha})^2-(1-\kappa_1)(\Im{\gamma}-\Im{\alpha})^2+2\kappa_2(\Re{\gamma}-\Re{\alpha})(\Im{\gamma}-\Im{\alpha})},
        \label{eq:one_mode_gaussian_fidelity}
    \end{equation}
    where $\kappa_1\equiv \cos(2\tau)\tanh(\zeta)$ and $\kappa_2\equiv \sin(2\tau)\tanh(\zeta)$.
    \label{lemma:one_mode_gaussian_fidelity}
\end{lemma}


The correlator $C_1$ is simply the square of fidelity as
\begin{align}
    C_1^{\rm Gauss} &\equiv \mathbb{E}_{\alpha\sim \calN_E^{\rm C}}\left[|\braket{\psi|\alpha}|^4\right] = \mathbb{E}_{\alpha\sim \calN_E^{\rm C}}\left[F(\psi,\ket{\alpha})^2\right]= \frac{\sech^2(\zeta)e^{-R(E)/G_1(E)}}{\sqrt{G_1(E)}},
    \label{eq:C1_gaussian}
\end{align}
where the last line is obtained from the average over real and imaginary parts of $\alpha$ separately. Here we define
\begin{align}
    G_1(x) &= 1+4x+4\sech^2(\zeta)x^2 \label{eq:G1_def}\\
    R(x) &= 2|\gamma|^2+4\sech^2(\zeta)|\gamma|^2x + 2\tanh(\zeta)|\gamma|^2\cos(2(\varphi+\tau)),
\end{align}
where $\varphi = \arctan(\Im{\gamma}/\Re{\gamma})$ is the angle of complex number $\gamma$. 
In the asymptotic region of $E$, one can see that $C_1^{\rm Gauss} \sim 1/2E$. The above $C_1^{\rm Gauss}$ for coherent state with $\zeta=0,\tau=0$ and single-mode squeezed vacuum (SMSV) state with $\gamma=0, \tau=0$ is reduced to
\begin{align}
    C_1^{\rm Coh} &= \frac{e^{-2|\gamma|^2/(1+2E)}}{1+2E}, \label{eq:C1_coh}\\
    C_1^{\rm SMSV} &= \frac{\sech^2(\zeta)}{\sqrt{1+4E+4\sech^2(\zeta) E^2}}. \label{eq:C1_smsv}
\end{align}

The correlator $C_2$ is the product of fidelity between $\psi$ and coherent states $\ket{\alpha_z\pm \alpha_{1-z}}$ as
\begin{align}
     C_2^{\rm Gauss}(z) &\equiv \mathbb{E}_{\alpha_y\sim \calN_{yE}^{\rm C}}\left[\prod_{h=0}^1|\braket{\psi|\alpha_z+(-1)^h\alpha_{1-z}}|^2 \right] 
     \nonumber
     \\
     &= \mathbb{E}_{\alpha_y\sim \calN_{yE}^{\rm C}}\left[\prod_{h=0}^1F(\psi,\ket{\alpha_z+(-1)^h\alpha_{1-z}}) \right]=\frac{\sech^2(\zeta)e^{-R(zE)/G_1(zE)}}{\sqrt{G_1(E-zE)G_1(zE)}}. \label{eq:C2_gaussian}
\end{align}
In the asymptotic region of $E$, we also see that $C_2^{\rm Gauss}(z)\sim 1/4z(1-z)E^2$. For coherent and SMSV states, we also have
\begin{align}
    C_2^{\rm Coh}(z) &= \frac{-2|\gamma|^2/(1+2zE)}{[1+2(1-z)E](1+2zE)}, \label{eq:C2_coh}\\
    C_2^{\rm SMSV}(z) &= \frac{\sech^2(\zeta)}{\sqrt{G_1(E-zE)G_1(zE)}}. \label{eq:C2_smsv}
\end{align}

For $C_3$, we have
\begin{align}
    C_3^{\rm Gauss}(z, \tilde{z}) &= \mathbb{E}_{\alpha_y\sim \calN_{yE}^{\rm C}}\left[\prod_{h=0}^1 |\braket{\psi|\alpha_{z}+(-1)^h \alpha_{\tilde{z}}+(-1)^h \alpha_{1-z-\tilde{z}}}||\braket{\alpha_{z}+(-1)^h\alpha_{\tilde{z}}-(-1)^h\alpha_{1-z-\tilde{z}}|\psi}|\right]\nonumber\\
    &= \mathbb{E}_{\alpha_y\sim \calN_{yE}^{\rm C}}\left[\prod_{h=0}^1 \sqrt{F(\psi,\ket{\alpha_{z}+(-1)^h \alpha_{\tilde{z}}+(-1)^h \alpha_{1-z-\tilde{z}}})}\sqrt{F(\psi,\ket{\alpha_{z}+(-1)^h\alpha_{\tilde{z}}-(-1)^h\alpha_{1-z-\tilde{z}}})}\right]\nonumber\\
    &= \frac{\sech^2(\zeta)e^{-R(zE)/G_1(z E)}}{\sqrt{G_1(z E)G_1(\tilde{z} E)G_1[(1-z-\tilde{z})E]}},
    \label{eq:C3_gaussian}
\end{align}
which clearly approaches the scaling of $1/E^3$ in the asymptotic region of $E$, and thus can be omitted. We expect that the scaling of $C_3$ can also be generalized to other non-Gaussian states as well, though it may not be easy to solve.

The bounds for variance of gradient in preparation of an arbitrary Gaussian state are
\begin{align}
    {\rm Var}\left[\partial_{\theta_k}\calC(
    \bm x)\right]
    &\ge \frac{1}{2}\left[\frac{3^{L-1}}{4^L}  C_1^{\rm Gauss} + \left(\frac{1}{4}-\frac{3^{L}}{4^L}\right)\min_{\ell} C_2^{\rm Gauss}\left(\frac{\ell}{L}\right)\right] + \mathcal{O}\left(\frac{1}{E^3}\right) \label{eq:var_lb_gaussian},\\
    {\rm Var}\left[\partial_{\theta_k}\calC(
    \bm x)\right]
    &\le \frac{1}{2}\left[\frac{3^{L-1}}{4^L}  C_1^{\rm Gauss} + \left(\frac{1}{4}+\frac{2^{L-1}}{4^L}\right)\max_{\ell} C_2^{\rm Gauss}\left(\frac{\ell}{L}\right)\right] + \mathcal{O}\left(\frac{1}{E^3}\right) ,\label{eq:var_ub_gaussian}
\end{align}
where the minimization and maximization is over all integers $1\le \ell \le L-1$. In the asymptotic region of $E$, as $C_1^{\rm Gauss}$ and $C_2^{\rm Gauss}$ shows different scaling, we can find the variance of the gradient is dominated by $1/E$ when
\begin{align}
    &\frac{1/4E^2}{\left(3/4\right)^L/6E} \in \mathcal{O}(1)\Rightarrow E \in \Omega(1)\frac{3}{2}\left(\frac{4}{3}\right)^L\in \Omega(\exp L),\label{eq:Ec}
\end{align}
or equivalently, 
\begin{align}
    L\in \frac{1}{\log(4/3)}\log\left(\mathcal{O}(1)\frac{2E}{3}\right)\in \mathcal{O}(\log E).
\end{align}
Depending on the energy, we can classify the CV VQCs in asymptotic $E$ region as shallow and deep circuits.

When the circuit is as shallow as $L \in \mathcal{O}\left(\log E\right)$, the bounds for variance of gradient is dominated by the first $\sim1/E$ term from correlator $C_1^{\rm Gauss}$, which are identical and thus describe the variance of the gradient as
\begin{align}
    {\rm Var}\left[\partial_{\theta_k}\calC\right] = \frac{1}{6}\left(\frac{3}{4}\right)^L C_1^{\rm Gauss}  + \mathcal{O}\left(\frac{1}{E^2}\right).
    \label{eq:var_shallow_gaussian}
\end{align}
In the preparation of a nonzero mean Gaussian state, i.e. coherent state, the leading order above brings us a peak of variance at about $E\sim |\gamma|^2=E_t$ which is the target state energy, and the variance shows the scaling of $1/E$ when $E\gtrsim E_t$. While for zero-mean Gaussian state i.e. SMSV state, Eq.~\eqref{eq:C1_smsv} monotonically decreases with $E$, and it can be estimated that when $E\ge \cosh(\zeta)=\sqrt{1+E_t}$, the variance of the gradient approaches the scaling of $1/E$.

On the other hand, when the circuit is as deep as $L \in \Omega(\log E)$, then the bounds of variance is denominated by the second $\sim1/E^2$ term from correlator $C_2^{\rm Gauss}$, and bounded as
\begin{align}
    \frac{1}{8}\min_{\ell}C_2^{\rm Gauss}\left(\frac{\ell}{L}\right) \le {\rm Var} \left[\partial_{\theta_k}\calC\right] \le \frac{1}{8}\max_{\ell} C_2^{\rm Gauss}\left(\frac{\ell}{L}\right).
     \label{eq:var_deep_gaussian}
\end{align}
In the asymptotic region of $E$, both sides follow the scaling of $1/E^2$, and so as the variance of the gradient itself. For a nonzero mean Gaussian state like coherent state, there is also a peak of variance by solving the extremals of Eq.~\eqref{eq:C2_coh}, which stays in the range of $[E_t/2, E_t]$. Beyond it, the variance begins to approach $1/E^2$. However, for zero-mean Gaussian state like SMSV state, $C_2^{\rm SMSV}(z)$ in Eq.~\eqref{eq:C2_smsv} simply decreases with $E$, and through a comparison of terms involving $E^{3/2}$ and $E^2$ in the denominator,  the scaling of $1/E^2$ is estimated to start from $E\sim L\cosh{\zeta}=L\sqrt{1+E_t}$.

\subsection{Fock number states}
\label{app:var_fock}

For non-Gaussian states, the evaluation of fidelity is in general difficult. In this part, we consider the preparation of a Fock number state with a closed form of fidelity to provide a physical insight on the scaling of variance in preparation of non-Gaussian states. Fock number states form a complete orthonormal basis for the Hilbert space. In this section, we will denote a number state as $\ket{E_t}^{\rm F}$. 
We begin with lemmas about the distribution of the norm and complex argument angle of a Gaussian complex variable $\alpha \sim \calN_{\sigma^2}^{\rm C}$. As they are widely known, we simply state the results.
\begin{lemma}
    Given a complex Gaussian distributed random variable $\alpha\sim \calN_{\sigma^2}^{\rm C}$, the square of its norm follows the Gamma distribution $|\alpha|^2 \sim {\rm Gamma}(1,\sigma^2)$, with probability density function $p(|\alpha|^2) = e^{-|\alpha|^2/\sigma^2}/\sigma^2$. The argument $\arg\{\alpha\}\equiv \tan^{-1}\left(\Im{\alpha}/\Re{\alpha}\right)$ is uniform in $[-\pi,\pi]$.
\end{lemma}
One can further find that the difference of the arguments of two complex Gaussian variables $\alpha_1$, $\alpha_2$ from the same ensemble satisfy the triangular distribution as the following.
\begin{corollary}
    For complex Gaussian variables $\alpha_i\in \calN_{\sigma_i^2}^{\rm C}$ with $i=1,2$, the difference of their argument $\delta \equiv \arg\{\alpha_1\}-\arg\{\alpha_2\}$ satisfies the triangular distribution $\delta \sim {\rm Tri}(-2\pi,2\pi,0)$ with distribution $p(\delta)=(2\pi-|\delta|)/4\pi^2$.
\end{corollary}

With the lemmas in hand, we can obtain the correlator $C_1$ as
\begin{align}
    C_1^{\rm Fock} &\equiv \mathbb{E}_{\alpha\sim \calN_E^{\rm C}}\left[|{}^{\rm F}\langle E_t\ket{\alpha}|^4\right] = \mathbb{E}_{\alpha\sim \calN_E^{\rm C}} \left[\frac{e^{-2|\alpha|^2}}{(E_t!)^2}|\alpha|^{4E_t}\right]= \mathbb{E}_{|\alpha|^2 \sim {\rm Gamma}(1,E)}\left[\frac{e^{-2|\alpha|^2}}{(E_t!)^2}|\alpha|^{4E_t}\right]= \frac{(2E_t)!}{(2^{E_t}E_t!)^2}\frac{\left(1+1/2E\right)^{-2E_t}}{1+2E}.
    \label{eq:C1_fock}
\end{align}
Similarly, the correlator $C_2$ becomes
\begin{align}
    &C_2^{\rm Fock}(x) \equiv \mathbb{E}_{\alpha_y\sim \calN_{yE}^{\rm C}}\left[\prod_{h=0}^1|{}^{\rm F}\langle E_t\ket{\alpha_z+(-1)^h\alpha_{1-z}}|^2 \right] = \mathbb{E}_{\alpha_y\sim \calN_{yE}^{\rm C}}\left[\frac{e^{-|\alpha_z+\alpha_{1-z}|^2}}{E_t!}|\alpha_z+\alpha_{1-z}|^{2E_t} \frac{e^{-|\alpha_z-\alpha_{1-z}|^2}}{E_t!}|\alpha_z-\alpha_{1-z}|^{2E_t}\right]\nonumber\\
    &= \frac{1}{(E_t!)^2}\mathbb{E}_{|\alpha_y|^2\sim {\rm Gamma}(1,yE), \delta\sim {\rm Tri}(-2\pi,2\pi,0)} \left[e^{-2|\alpha_z|^2-2|\alpha_{1-z}|^2}\left(|\alpha_z|^4+|\alpha_{1-z}|^4+2|\alpha_z|^2|\alpha_{1-z}|^2 - 4|\alpha_z|^2|\alpha_{1-z}|^2\cos^2\delta\right)^{E_t} \right] \nonumber\\
    &= \frac{1}{(E_t!)^2}\mathbb{E}_{|\alpha_y|^2\sim {\rm Gamma}(1,yE)}\left[e^{-2(|\alpha_z|^2 + |\alpha_{1-z}|^2)}\mathbb{E}_{\delta\sim {\rm Tri}(-2\pi,2\pi,0)}\left[\left(|\alpha_z|^4+|\alpha_{1-z}|^4-2|\alpha_z|^2|\alpha_{1-z}|^2\cos(2\delta)\right)^{E_t}\right]\right]\nonumber\\
    &= \frac{1}{(E_t!)^2} \mathbb{E}_{|\alpha_y|^2\sim {\rm Gamma}(1,yE)}\left[e^{-2(|\alpha_z|^2 + |\alpha_{1-z}|^2)} \left(|\alpha_z|^2 + |\alpha_{1-z}|^2\right)^{2E_t} {}_2F_1\left(\frac{1}{2},-E_t,1,\frac{4|\alpha_z|^2 |\alpha_{1-z}|^2}{(|\alpha_z|^2 + |\alpha_{1-z}|^2)^2}\right)\right],
    \label{eq:C2_fock_reduce}
\end{align}
where ${}_2F_1(a,b,c,z)$ is the hypergeometric function. The integral over $|\alpha_z|^2, |\alpha_{1-z}|^2$ above is hard to evaluate, but noticing that $0\le 4|\alpha_z|^2 |\alpha_{1-z}|^2/(|\alpha_z|^2 + |\alpha_{1-z}|^2)^2 \le 1$,
\begin{align}
     {}_2F_1\left(\frac{1}{2},-E_t,1,1\right) \le {}_2F_1\left(\frac{1}{2},-E_t,1,\frac{4|\alpha_z|^2 |\alpha_{1-z}|^2}{(|\alpha_z|^2+|\alpha_{1-z}|^2)^2}\right) \le 1,
\end{align}
where the L.H.S. is a constant depending on $E_t$ only. The ensemble average in Eq.~\eqref{eq:C2_fock_reduce} without hypergeometric function is
\begin{align}
    &\frac{1}{(E_t!)^2} \mathbb{E}_{|\alpha_y|^2\sim {\rm Gamma}(1,yE)}\left[e^{-2(|\alpha_z|^2 + |\alpha_{1-z}|^2)} \left(|\alpha_z|^2 + |\alpha_{1-z}|^2\right)^{2E_t}\right]\nonumber\\
    &= \frac{(2E_t)!}{(2^{E_t} E_t!)^2}\frac{\left[ (1-z)(1+2zE)\left(1+\frac{1-2z}{z+2(1-z)zE}\right)^{2E_t}-2(1-z)zE-z\right]}{1-2z}\frac{\left(1+\frac{1}{2zE}\right)^{-2E_t}}{\left[1+2(1-z)E\right](1+2zE)}.
\end{align}
Therefore, we have the correlator $C_2^{\rm Fock}$ as
\begin{align}
    C_2^{\rm Fock}(z) = \eta \frac{(2E_t)!}{(2^{E_t} E_t!)^2}\frac{\left[ (1-z)(1+2zE)\left(1+\frac{1-2z}{z+2(1-z)zE}\right)^{2E_t}-2(1-z)zE-z\right]}{1-2z}\frac{\left(1+\frac{1}{2zE}\right)^{-2E_t}}{\left[1+2(1-z)E\right](1+2zE)} .\label{eq:C2_fock_full}
\end{align}
where $\eta$ equals ${}_2F_1(1/2,-E_t,1,1)$ in lower bound and $1$ in upper bound.
In the asymptotic region of $E$, the long fraction in the middle can be reduced to $1+2E_t$, and thus the correlator becomes
\begin{align}
    C_2^{\rm Fock}(x) = \eta \frac{(1+2E_t)(2E_t)!}{(2^{E_t} E_t!)^2}\frac{\left(1+1/2zE\right)^{-2E_t}}{\left[1+2(1-z)E\right](1+2zE)}.
    \label{eq:C2_fock_simplify}
\end{align}

The correlator $C_3$ for Fock state is
\begin{align}
    C_3^{\rm Fock}(z, \tilde{z}) &= \mathbb{E}_{\alpha_y\sim \calN_{yE}^{\rm C}}\left[\prod_{h=0}^1 |{}^{\rm F}\langle E_t\ket{\alpha_{z}+(-1)^h \alpha_{\tilde{z}}+(-1)^h \alpha_{1-z-\tilde{z}}}||\bra{\alpha_{z}+(-1)^h\alpha_{\tilde{z}}-(-1)^h \alpha_{1-z-\tilde{z}}}E_t\rangle^{\rm F}|\right]\nonumber\\
    &= \mathbb{E}_{\alpha_y \sim \calN_{yE}^{\rm C}}\left[\prod_{h=0}^1 \left(\frac{e^{-|\alpha_z+(-1)^h\alpha_{\tilde{z}}+(-1)^h\alpha_{1-z-\tilde{z}}|^2/2}}{\sqrt{E_t!}}|\alpha_{z}+(-1)^h \alpha_{\tilde{z}}+(-1)^h \alpha_{1-z-\tilde{z}}|^{E_t}\right.\right.\nonumber\\
    &\quad \quad \quad \quad \quad \quad \quad \left.\left. \times\frac{e^{-|\alpha_{z}+(-1)^h\alpha_{\tilde{z}}-(-1)^h\alpha_{1-z-\tilde{z}}|^2/2}}{\sqrt{E_t!}}|\alpha_{z}+(-1)^h\alpha_{\tilde{z}}-(-1)^h\alpha_{1-z-\tilde{z}}|^2\right)\right]
    \label{eq:C3_fock}
\end{align}
It is hard to solve it anlytically, instead, we perform monte-carlo calculation to show its asymptotic scaling in Fig.~\ref{fig:C3_fock}. Here we choose $z=\tilde{z}=1/3$. It clearly shows that in the asymptotic region of $E$, the scaling of $C_3^{\rm Fock}(z,\tilde{z})$ as $\sim 1/E^3$ (dashed-dot line), which is a higher order term compared to $C_1^{\rm Fock}$ and $C_2^{\rm Fock}(z)$, and can be omitted as well.

\begin{figure}
    \centering
    \includegraphics[width=0.4\textwidth]{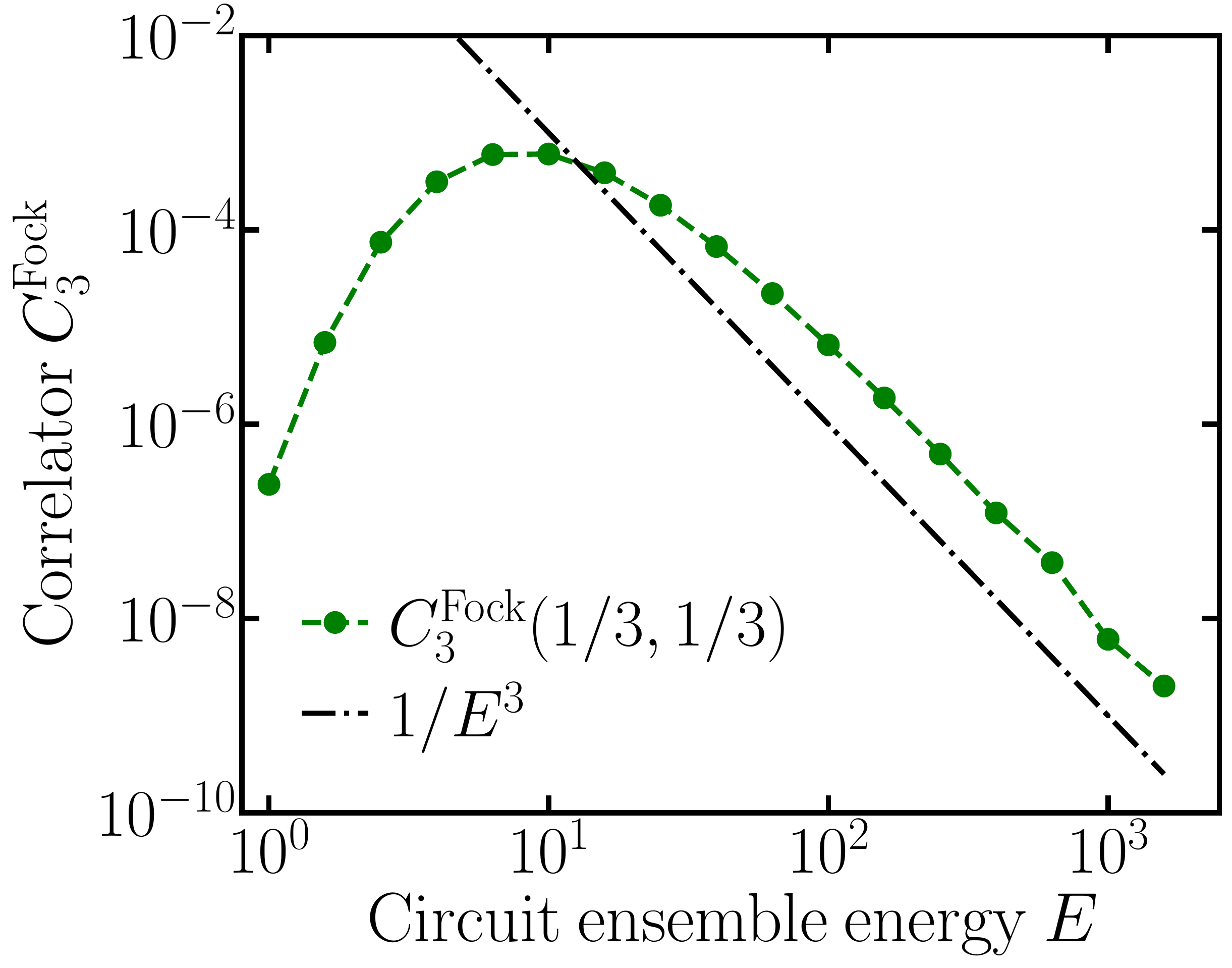}
    \caption{Numerical result of correlator $C_3^{\rm Fock}(z,\tilde{z})$ for Fock state in Eq.~\eqref{eq:C3_fock}. Here we choose $z=\tilde{z}=1/3$. The dashed-dot line is $1/E^3$ for reference.}
    \label{fig:C3_fock}
\end{figure}

Combining Eqs.~\eqref{eq:C1_fock} and~\eqref{eq:C2_fock_simplify}, we have the bounds for variance of gradient in preparation of a Fock state as
\begin{align}
    {\rm Var}\left[\partial_{\theta_k}\calC(
    \bm x)\right]
    &\ge \frac{1}{2}\left[\frac{3^{L-1}}{4^L}  C_1^{\rm Fock} + \left(\frac{1}{4}-\frac{3^{L}}{4^L}\right)\min_{\ell} C_2^{\rm Fock}\left(\frac{\ell}{L}\right)\right] + \mathcal{O}\left(\frac{1}{E^3}\right) \label{eq:var_lb_fock},\\
    {\rm Var}\left[\partial_{\theta_k}\calC(
    \bm x)\right]
    &\le \frac{1}{2}\left[\frac{3^{L-1}}{4^L}  C_1^{\rm Fock} + \left(\frac{1}{4}+\frac{2^{L-1}}{4^L}\right)C_2^{\rm Fock}\left(1-\frac{1}{L}\right)\right] + \mathcal{O}\left(\frac{1}{E^3}\right). \label{eq:var_ub_fock}
\end{align}

Similar to the discussion about Gaussian state preparation, we can also identify the critical $E$ for scaling transition from $1/E^2$ to $1/E$ with fixed $L$ at
\begin{align}
    \frac{\eta(1+2E_t)/4E^2}{\left(3/4\right)^L/6E} \in \mathcal{O}(1) \Rightarrow 
    E\in \Omega(1)\frac{3\eta(1+2E_t)}{2}\left(\frac{4}{3}\right)^L \in \Omega\left(\exp L\right).
    \label{eq:E_critical}
\end{align}
Or equivalently, we have 
\begin{align}
    L\in \frac{1}{\log(4/3)}\log\left(\mathcal{O}(1)\frac{2E}{3\eta(1+2E_t)}\right)\in\mathcal{O}(\log E).
\end{align}

When the circuit depth is as shallow as $L \in \mathcal{O}\left(\log E\right)$, the bounds for variance of gradient is dominated by the first $\sim 1/E$ term from correlator $C_1^{\rm Fock}$, which are identical and thus describe the variance of the gradient as
\begin{align}
    {\rm Var}\left[\partial_{\theta_k}\calC\right] = \frac{1}{6}\left(\frac{3}{4}\right)^L C_1^{\rm Fock}  + \mathcal{O}\left(\frac{1}{E^2}\right).
    \label{eq:var_shallow_coh}
\end{align}
The peak of variance can also be found at $E\sim E_t$. On the other hand, when the circuit depth is as deep as $L \in \Omega(\log E)$, then the bounds of variance is denominated by the second $\sim 1/E^2$ term from correlator $C_2^{\rm Fock}$ as
\begin{align}
    \frac{1}{8}\min_{\ell}C_2^{\rm Fock}\left(\frac{\ell}{L}\right) \le {\rm Var} \left[\partial_{\theta_k}\calC\right] \le \frac{1}{8}C_2^{\rm Fock}\left(1-\frac{1}{L}\right),
     \label{eq:var_deep_coh}
\end{align}
where the coefficient $\eta$ in $C_2^{\rm Fock}$ on L.H.S. and R.H.S of inequality is chosen to be ${}_2F_1(1/2,-E_t, 1,1)$ and $1$ separately. In asymptotic region, the variance also follows the scaling $1/E^2$ and the peak is in the region $[E_t/2, E_t]$.

\section{Variance of gradient in preparation of multi-mode qumode CV states}
\label{app:multi-mode}

In this section, we show that the results in the single-mode case in Appendix~\ref{app:variance} generalize to the variance of the gradient in preparation of an arbitrary multi-mode CV state $\ket{\psi}_{\bm m}$. Lemma~\ref{lemma_universality} (in Appendix~\ref{app:universal_control}) states that one can achieve universal control on multiple modes and one qubit by applying the set of ECD gates and single qubit rotations. Therefore, we consider a ladder setup of gates in circuits, as shown in Fig.~\ref{fig:circuit_modes_scheme}. In the following, we use superscript in $A^{(j)}$ to denote the operator $A$ that applies to all qumode trivially other than $j$th mode.

\begin{figure}
    \centering
    \includegraphics[width=0.9\textwidth]{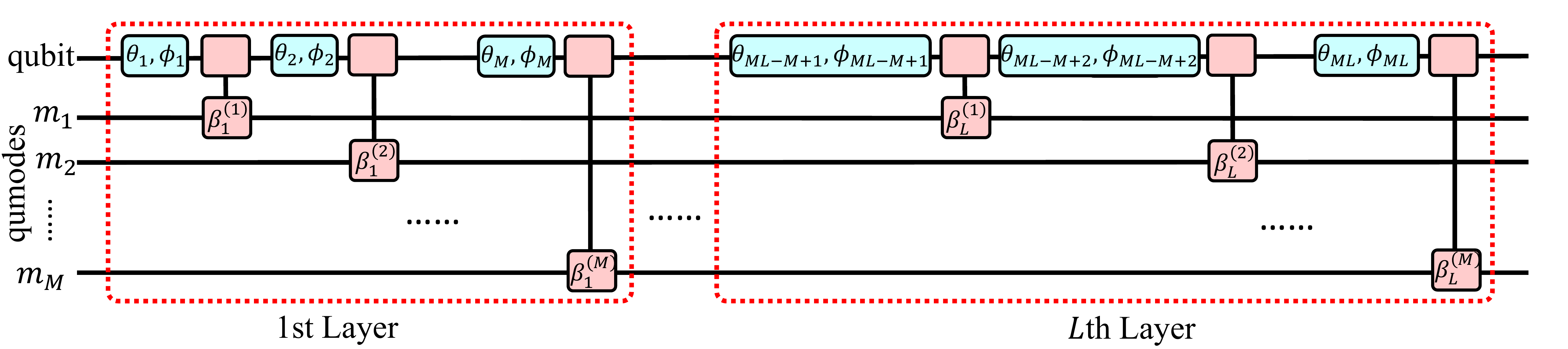}
    \caption{Scheme of $M$-mode $L$-layer CV VQC. Cyan boxes with $\theta_\ell,\phi_\ell$ ranging from $1\le \ell \le ML$ represents the qubit rotation $U_{\rm R}(\theta_\ell, \phi_\ell)$; Pink boxes with $\beta_\ell^{(j)}$ denotes the ECD gate $U_{\rm ECD}^{(j)}(\beta_\ell^{(j)})$ applying on the qubit and $j$th mode.}
    \label{fig:circuit_modes_scheme}
\end{figure}
We begin the analyses by generalizing the state representation in Appendix~\ref{app:state_repre} to the multi-mode case. To simplify the notation, we define ${\bm \beta}^{(j)} = (\beta_1^{(j)}, \dots, \beta_L^{(j)})^T$, ${\bm \theta}= (\theta_1, \dots, \theta_{ML})^T$, ${\bm \phi}= (\phi_1, \dots, \phi_{ML})^T$ and the overall parameters $\bm x=(\{{\bm \beta}^{(j)}\}_{j=1}^M,{\bm \theta}, {\bm \phi})$. We also denote $\bm m=(m_1,\cdots, m_M)$ as all modes.
For an $M$-mode system, each of the $L$ layers involves $M$ single qubit rotations and ECD gates applied between the control qubit and the modes $m_j$ from $j=1$ to $M$ (see the set of gates surrounded by the red dashed box in Fig.~\ref{fig:circuit_modes_scheme}). The corresponding variational parameters in an $L$-depth circuit are $\{\theta_\ell\}_{\ell=1}^{ML},\{\phi_\ell\}_{\ell=1}^{ML}$ and $\bigcup_{j=1}^M \{\beta_\ell^{(j)}\}_{\ell=1}^L$ with superscript $(j)$ denoting the $j$th mode as explained above. 
The unitary of the $M$-mode $L$-depth circuit in Fig.~\ref{fig:circuit_modes_scheme} is
\begin{align}
    U &= \prod_{\ell=1}^L\prod_{j=1}^M U_{\rm ECD}^{(j)}\left(\beta_\ell^{(j)}\right)U_{\rm R}(\theta_{M(\ell-1)+j},\phi_{M(\ell-1)+j})\nonumber\\
    &= \prod_{\ell=1}^L \prod_{j=1}^M \begin{pmatrix}
    e^{i\phi_{M(\ell-1)+j}}\sin\frac{\theta_{M(\ell-1)+j}}{2}D^{(j)}(-\beta_\ell^{(j)}) & \cos\frac{\theta_{M(\ell-1)+j}}{2}D^{(j)}(-\beta_\ell^{(j)})\\
    \cos\frac{\theta_{M(\ell-1)+j}}{2}D^{(j)}(\beta_\ell^{(j)}) & e^{i(\pi-\phi_{M(\ell-1)+j})}\sin\frac{\theta_{M(\ell-1)+j}}{2}D^{(j)}(\beta_\ell^{(j)})
    \end{pmatrix},
\end{align}
where the block matrix representation above is adopted from  Eq.~\eqref{eq:block_matrix}.
Note that the displacement operator $D^{(j)}(\cdot)$ acts on all $M$ modes, where the $j$th mode has a displacement while the rest are trivial identity.
The output state of unitary $U_L^{(M)}$ on initial state $\ket{0}_q\otimes_{j=1}^M \ket{0}_{m_j}$ is
\begin{align}
    &\ket{\psi(\bm \theta, \bm \phi, \{\bm \beta^{(j)}\}_{j=1}^M)}_{q,\bm m} \equiv U \left(\ket{0}_q\otimes_{j=1}^M \ket{0}_{m_j}\right)\nonumber\\
    &= \sum_{a=0}^1 \sum_{\bm s} w_{\bm s,a}(\bm \theta, \bm \phi) \ket{a}_q e^{i\sum_{j=1}^M\chi_{\bm s^{(j)}}(\bm \beta^{(j)})}\bigotimes_{j=1}^M \ket{(-1)^a \bm s^{(j)}\cdot \bm \beta^{(j)}}_{m_j}
    =\sum_{a=0}^1 \sum_{\bm s} w_{\bm s,a}(\bm \theta, \bm \phi) \ket{a}_q \ket{\bm B_{\bm s,a}}_{\bm m},
    \label{eq:psi_multi}
\end{align}
where $\bm s$ is a length-$ML$ sign vector as $\bm s = (\bm s_{1:ML-1},-1)$ with $\bm s_{1:ML-1}\in\{-1,1\}^{ML-1}$, defined in the same way as $\bm s$ in Eq.~\eqref{eq:psi}. The weight $w_{\bm s,a}$ is defined in terms of $\bm s$ and $a$ in the same way as in Eq.~\eqref{eq:psi} via replacing $L\to ML$. 
Here $\bm s^{(j)}$ is the corresponding sign vector for $j$th mode, which is easily generated by collecting all $(M(\ell-1)+j)$th with $\ell\in[1, L-1]$ elements of $\bm s$ in order as
\begin{align}
    \bm s^{(j)} = \left(\bm s_{j}, \bm s_{M+j},\dots, \bm s_{(L-1)M+j}\right),
    \label{eq:s_mode}
\end{align}
where $\bm s_j$ denotes the $j$th element of the whole length-$ML$ sign vector $\bm s$. Note that the sign vectors for all modes $\{\bm s^{(j)}\}_{j=1}^M$ together form a partition of $\bm s$. Inversely, another explicit way to generate all $\{\bm s^{(j)}\}_{j=1}^M$ is
\begin{align}
    \bm s^{(j)} &\in \{-1,1\}^{L}, \quad \quad \text{if $1\le j\le M-1$},\label{eq:s_mode_1}\\
    \bm s^{(M)} &= (\bm s^{(M)}_{1:L-1},-1),  \quad \text{where $\bm s^{(M)}_{1:L-1} \in \{-1,1\}^{L-1}$.}\label{eq:s_mode_2}
\end{align}
and join them together in the inverse way of partition to construct the whole sign vector $\bm s$.
The displacement $B_{\bm s^{(j)},a}$ for each mode is defined the same as in Eq.~\eqref{eq:psi}, and the state on all qumodes in Eq.~\eqref{eq:psi_multi} is defined as
\begin{align}
    \ket{\bm B_{\bm s,a}}_{\bm m} \equiv e^{i\sum_{j=1}^M\chi_{\bm s^{(j)}}}\otimes_{j=1}^M \ket{(-1)^a \bm s^{(j)}\cdot \bm \beta}_{m_j}
\end{align}
for convenience.  If we define $v_{\bm s,a}(\bm \theta,\bm \phi, \bm \beta)\equiv e^{i\sum_{j=1}^M\chi_{\bm s^{(j)}}(\bm \beta^{(j)})} w_{\bm s,a}(\bm \theta, \bm \phi)$, we have
\begin{align}
    \ket{\psi(\bm x)}_{q,\bm m} = \sum_{a=0}^1 \sum_{\bm s}v_{\bm s,a}(\bm \theta, \bm \phi, \bm \beta)\ket{a}_q\bigotimes_{j=1}^M \ket{(-1)^a \bm s^{(j)}\cdot \bm \beta^{(j)}}_{m_j},
    \label{eq:state_ensemble_modes}
\end{align}
which generalizes Eq.~\eqref{eq:state_ensemble}.

To conclude, the correspondance between Eq.~\eqref{eq:psi} and~\eqref{eq:psi_multi} indicates a map from the $M$-mode state generated by $U_{L}$ to a single mode state generated by $U_{ML}$
\begin{align}
    \mathbb{M}: \psi_{L,M}(\bm \theta, \bm \phi, \{\bm \beta^{(j)}\}_{j=1}^M) \rightarrow \psi_{ML,1}(\bm \theta, \bm \phi, \bm \beta).
    \label{eq:map_modes}
\end{align}
The proof is easy to see from an explicit example of $L=1$ and $M=2$ and then generalize by mathematical induction, which is same as in Appendix~\ref{app:state_repre}. The output state of $M=2$ modes and $L=1$ circuit as
\small
\begin{align}
    \ket{\psi}
    &= \begin{pmatrix}
    e^{i\phi_2}\sin\frac{\theta_2}{2}D^{(2)}(-\beta^{(2)}) & \cos\frac{\theta_2}{2}D^{(2)}(-\beta^{(2)})\\
    \cos\frac{\theta_2}{2}D^{(2)}(\beta^{(2)}) & e^{i(\pi-\phi_2)}\sin\frac{\theta_2}{2}D^{(2)}(\beta^{(2)})
    \end{pmatrix}\begin{pmatrix}
        e^{i\phi_1}\sin\frac{\theta_1}{2}D^{(1)}(-\beta^{(1)}) & \cos\frac{\theta_1}{1}D^{(1)}(-\beta^{(2)})\\
    \cos\frac{\theta_1}{2} D^{(1)}(\beta^{(1)}) & e^{i(\pi-\phi_1)}\sin\frac{\theta_1}{2} D^{(1)}(\beta^{(1)})
    \end{pmatrix}
    \begin{pmatrix}
        \ket{0}_{m_1}\ket{0}_{m_2}\\
        0
    \end{pmatrix}\\
    &= \ket{0}_q\otimes\left(
        e^{i(\phi_1+\phi_2)}\sin\frac{\theta_1}{2}\sin\frac{\theta_2}{2}\ket{-\beta^{(1)}}_{m_1}\ket{-\beta^{(2)}}_{m_2}+\cos\frac{\theta_1}{2}\cos\frac{\theta_2}{2}\ket{\beta^{(1)}}_{m_1}\ket{-\beta^{(2)}}_{m_2}\right)\nonumber\\
    &\quad + \ket{1}_q\otimes\left(e^{i\phi_1}\sin\frac{\theta_1}{2}\cos\frac{\theta_2}{2}\ket{-\beta^{(1)}}_{m_1}\ket{\beta^{(2)}}_{m_2} + e^{i(\pi-\phi_2)}\cos\frac{\theta_1}{2}\sin\frac{\theta_2}{2}\ket{\beta^{(1)}}_{m_1}\ket{\beta^{(2)}}_{m_2}\right),
\end{align}
\normalsize
which indicates a clear mapping to the state $\ket{\psi_{L=2,M=1}(\bm \theta, \bm \phi, \bm \beta)}$ shown in Eq.~\eqref{eq:psi_d2}.

For energy regularization, we still have the displacement in every ECD gate Gaussian distributed, $\beta_\ell^{(j)}\sim \calN_{E/L}^{\rm C}$, and thus the ensemble-averaged energy per mode is also $\mathbb{E}\braket{m_j^\dagger m_j}=E$.

We still consider the gradient with respect to the $k$th qubit rotation angle $\theta_k$. For a general $M$-mode operator, it is easy to check that the ensemble average of gradient is still zero, and the variance can be written in the same form as in Eq.~\eqref{eq:variance}. Aligned with the study in Appendix~\ref{app:variance}, the target state of control qubit is also $\ket{0}_q$ and Eq.~\eqref{eq:Opm} becomes
\begin{align}
    \mathbb{E}\left[\braket{O}_{k^{(+1)}} \braket{O}_{k^{(\mu)}} \right] &= \sum_{\bm s,\bm s^\prime,\bm r,\bm r^\prime} \mathbb{E}\left[w_{\bm s,k^{(+1)}}w^*_{\bm s^\prime, k^{(+1)}}w_{\bm r,k^{(\mu)}}w^*_{\bm r^\prime, k^{(\mu)}}\right] \mathbb{E}\left[\braket{\psi|\bm B_{\bm s}}\braket{\bm B_{\bm s^\prime}|\psi}\braket{\psi|\bm B_{\bm r}}\braket{\bm B_{\bm r^\prime}|\psi}\right].
\end{align}
Via the mapping from $\psi_{E,L,M}$ to $\psi_{E,ML,1}$, the variance in the multi-mode scenario has the same form with Eq.~\eqref{eq:variance_single_mode} when one replaces the single-mode correlators with the multi-mode correlators. We discuss them in the following.
For $C_1$, each $\bm s^{(j)}\cdot \bm \beta^{(j)}\sim \calN^{\rm C}_{E}$ is in Gaussian distribution, and we have
\begin{align}
     C_1 = \mathbb{E}\left[|\braket{\psi|\bm B_{\bm s}}|^4\right] = \mathbb{E}_{\alpha\sim \calN_E^{\rm C}}\left[|\bra{\psi}\left(\otimes_{j=1}^M \ket{\alpha_j}\right)|^4\right].
     \label{eq:C1_modes}
\end{align}
Note that the ensemble average is performed over every $\{\alpha_j\}_{j=1}^M$ independently. 

For correlator $C_2$, we can still have for each $\bm s^{(j)}\cdot \bm \beta=\alpha_{z_j}+\alpha_{1-z_j}$ and $\bm r^{(j)} \cdot \bm \beta =\alpha_{z_j}-\alpha_{1-z_j}$, thus $C_2$ can be written as
\begin{align}
     C_2(\bm z) = \mathbb{E}\left[|\braket{\psi|\bm B_{\bm s}}|^2|\braket{\psi|\bm B_{\bm r}}|^2\right] = \mathbb{E}_{\alpha_y\sim \calN_{yE}^{\rm C}}\left[\prod_{a=0}^1 |\bra{\psi}\left(\otimes_{j=1}^M\ket{\alpha_{z_j}+(-1)^a\alpha_{1-z_j}}\right)|^2\right],
     \label{eq:C2_modes}
\end{align}
where we define $\bm z = (z_1, \dots, z_M)$.
However, unlike the one-mode case, in general it is possible that some of the elements in $\bm z$ is zero as long as at least one element of $\bm 1-\bm z$ is nonzero ($\bm 1=(1,\dots,1)$ is a length-$M$ vector), such that $\bm s\neq \bm r$.  Suppose the number of elements in $\bm z$ within $(0,1)$ is $N_{\bm z}$, then the probability of $N_{\bm z} = M$ is
\begin{align}
    p\left(N_{\bm z} = M\right) = \frac{(2^{ML-1})(2^L-2)^{M-1}(2^{L-1}-1)}{(2^{ML-1})(2^{ML-1}-1)} = \frac{(2^L-2)^M}{2^{ML}-2}.
    \label{eq:pr_multi}
\end{align}
The probability is exponentially approaching unity as $L$ increases, at a fixed value of $M$. We will discuss the consequence of the exponential scaling after we show the correlator's scaling of some typical states. 

The last correlator $C_3$ is
\begin{align}
    &\mathbb{E}\left[|\braket{\psi|\bm B_{\bm s}}||\braket{\bm B_{\bm s^\prime}|\psi}|\braket{\psi|\bm B_{\bm r}}|\braket{\bm B_{\bm r^\prime}|\psi}|\right]\nonumber\\
    &= \mathbb{E}\left[\left\lvert\bra{\psi}\left(\otimes_{j=1}^M\ket{\bm s^{(j)}\cdot \bm \beta^{(j)}}\right)\right\rvert \left\lvert\left(\otimes_{j=1}^M \bra{\bm s^{\prime (j)}\cdot \bm \beta^{(j)}}\right) \ket{\psi}\right\rvert \left\lvert\left(\otimes_{j=1}^M\ket{\bm r^{(j)}\cdot \bm \beta^{(j)}}\right)\right\rvert 
    \left\lvert\left(\otimes_{j=1}^M \bra{\bm r^{\prime (j)}\cdot \bm \beta^{(j)}}\right)\ket{\psi}\right\rvert\right]\nonumber\\
    &= \mathbb{E}_{\alpha_y\sim \calN_{yE}^{\rm C}}\left[\prod_{a=0}^1\left\lvert\bra{\psi}\left(\otimes_{j=1}^M\ket{\alpha_{z_j}+(-1)^a\alpha_{\tilde{z}_j}+(-1)^a\alpha_{1-z_j-\tilde{z}_j}}\right)\right\rvert \left\lvert\left(\otimes_{j=1}^M \bra{\alpha_{z_j}+(-1)^a\alpha_{\tilde{z}_j}-(-1)^a\alpha_{1-z_j-\tilde{z}_j}}\right) \ket{\psi}\right\rvert \right]
    \nonumber
    \\
    &\equiv C_3(\bm z, \tilde{\bm z}),\label{eq:C3_modes} 
\end{align}
where we utilize the same method as in Eq.~\eqref{eq:C3_def}. Similar to the discussion for $C_2$ above, it is also possible that some elements of $\bm z, \tilde{\bm z}$ are zeros, as long as there are at least one nonzero element in $\tilde{\bm z},\bm 1-\bm z-\tilde{\bm z}$ so that $\bm s,\bm s^\prime,\bm r,\bm r^\prime$ different from each other. Its scaling is also left to later discussion.

We then have the bound for variance of gradient as
\begin{align}
    {\rm Var}\left[\partial_{\theta_k}\calC(
    \bm x)\right] &= \frac{1}{2}\left[\frac{3^{ML-1}}{4^{ML}} C_1 + \left(\frac{1}{4}-\frac{2\cdot 3^{ML-1}-2^{ML-1}}{4^{ML}}\right)\min_{\bm z} C_2(\bm z)\right] + \mathcal{O}\left(C_2\right)\nonumber\\
    &\ge \frac{1}{2}\left[\frac{3^{ML-1}}{4^{ML}} C_1 + \left(\frac{1}{4}-\frac{3^{ML}}{4^{ML}}\right)\min_{\bm z} C_2(\bm z)\right] + \mathcal{O}\left(C_2\right),
    \label{eq:var_lb_modes}
\end{align}
and
\begin{align}
     {\rm Var}\left[\partial_{\theta_k}\calC(
    \bm x)\right] &= \frac{1}{2}\left[\frac{3^{ML-1}}{4^{ML}} C_1 + \left(\frac{1}{4}-\frac{2\cdot 3^{ML-1}-2^{ML-1}}{4^{ML}}\right)\max_{\{x^{(j)}\}} C_2(\bm z) \right] + \mathcal{O}\left(C_2\right)\nonumber\\
    &\le \frac{1}{2}\left[\frac{3^{ML-1}}{4^{ML}} C_1 + \left(\frac{1}{4}+\frac{2^{ML-1}}{4^{ML}}\right)\max_{\bm z} C_2(\bm z)\right] + \mathcal{O}\left(C_2\right),
    \label{eq:var_ub_modes}
\end{align}
where we have omitted the contribution of $C_3$ in the asymptotic region of $E\gg1$. Note that the coefficient ahead of each correlator is exactly the same as in Eqs.~\eqref{eq:var_lb} and~\eqref{eq:var_ub} by replacing $L\rightarrow ML$ suggested by the map in Eq.~\eqref{eq:map_modes}. We consider the asymptotic region where the circuit ensemble energy per mode is larger than the target state energy at any mode, $E\ge \max_j \braket{\psi|\left(m_j^\dagger m_j\right)|\psi}$.
In general, the above correlators are hard to evaluate, we obtain insights to their properties in two examples.

\subsection{Product states}
\label{app:var_prod}
A simple example to begin with is the product state with no correlation between any modes, $\ket{\psi}_{\bm m} = \otimes_{j=1}^M \ket{\psi_j}_{m_j}$, where $\ket{\psi_j}_{m_j}$ is the state of $j$th mode. We show that it is directly related to the one-mode correlators.

In this case, $C_1$ reduces to a product form,
\begin{align}
    C_1^{\rm Prod} &= \mathbb{E}_{\alpha_j\sim \calN_E^{\rm C}}\left[|\left(\otimes_{j=1}^M{}_{m_j}\langle\psi_j|\right)\left(\otimes_{j=1}^M \ket{\alpha_j}\right)|^4\right] 
    = 
    \prod_{j=1}^M \left(\mathbb{E}_{\alpha_j\sim \calN_E^{\rm C}}\left[|{}_{m_j}\langle\psi_j\ket{\alpha_j}|^4\right]\right).
    \label{eq:C1_prod_state}
\end{align}
Note that the ensemble average of the term inside parentheses is the correlator $C_1$ for one mode CV state in Eq.~\eqref{eq:C1_def}, which has been shown to have the scaling of $1/E$ in the asymptotic region. Therefore, we have $C_1\sim 1/E^M$ for an $M$-mode product state.

$C_2(\bm z)$ reduces to
\begin{align}
    C_2^{\rm Prod}(\bm z) &= \mathbb{E}_{\alpha_y\sim \calN_{yE}^{\rm C}}\left[\prod_{a=0}^1 |\left(\otimes_{j=1}^M\bra{\psi_j}\right)\left(\otimes_{j=1}^M\ket{\alpha_{z_j}+(-1)^a\alpha_{1-z_j}}\right)|^2\right]\nonumber\\
    &= \prod_{j=1}^M \left(\mathbb{E}_{\alpha_y\sim \calN_{yE}^{\rm C}}\left[\prod_{a=0}^1 |\braket{\psi_j|\alpha_{z_j}+(-1)^a\alpha_{1-z_j}}|^2\right]\right).
    \label{eq:C2_prod}
\end{align}
Note that for a specific $z_j$, if $z_j\in (0,1)$, the corresponding term inside the parenthesis of Eq.~\eqref{eq:C2_prod} is single-mode correlator $C_2$ whereas if $z_j=0,1$ the correspnonding one becomes $C_1$. Suppose the number of elements in $\bm z$ within the range of $(0,1)$ is $N_{\bm z}$, then the scaling of $C_2^{\rm Prod}$ is $\sim 1/E^{M+N_{\bm z}}$. According to Eq.~\eqref{eq:pr_multi}, the probability that $N_{\bm z} = M$ is $p=(2^L-2)^M/(2^{ML}-2)\sim 1-1/2^L$.

Similarly, $C_3$ becomes
\begin{align}
    &C_3^{\rm Prod}(\bm z,\tilde{\bm z}) 
    \nonumber
    \\
    &=\mathbb{E}_{\alpha_y\sim \calN_{yE}^{\rm C}}\left[\prod_{a=0}^1\left\lvert\bra{\psi}\left(\otimes_{j=1}^M\ket{\alpha_{z_j}+(-1)^a\alpha_{\tilde{z}_j}+(-1)^a\alpha_{1-z_j-\tilde{z}_j}}\right)\right\rvert \left\lvert\left(\otimes_{j=1}^M \bra{\alpha_{z_j}+(-1)^a\alpha_{\tilde{z}_j}-(-1)^a\alpha_{1-z_j-\tilde{z}_j}}\right) \ket{\psi}\right\rvert \right]\nonumber\\
    &= \prod_{j=1}^M \left(\mathbb{E}_{\alpha_y\sim \calN_{yE}^{\rm C}}\left[\prod_{a=0}^1\left\lvert\braket{\psi_j|\alpha_{z_j}+(-1)^a\alpha_{\tilde{z}_j}+(-1)^a\alpha_{1-z_j-\tilde{z}_j}}\right\rvert \left\lvert\braket{\alpha_{z_j}+(-1)^a\alpha_{\tilde{z}_j}-(-1)^a\alpha_{1-z_j-\tilde{z}_j}|\psi_j}\right\rvert \right]\right)\label{eq:C3_prod},
\end{align}
For convenience, we denote the term inside the parenthesis of Eq.~\eqref{eq:C3_prod} to be $C^{(j)}$. For a specific combination of $z_j,\tilde{z}_j,1-z_j-\tilde{z}_j$, if all of them are within $(0,1)$, $C^{(j)}$ is the single-mode correlator $C_3\sim 1/E^3$ in Eq.~\eqref{eq:C3_def}, on the other hand, if only two of them are in $(0,1)$ while the other is zero, $C^{(j)}$ becomes $C_2\sim 1/E^2$, furthermore if only one of the three is nonzero $C^{(j)}$ is just $C_1\sim 1/E$. Therefore, given different $\bm z, \tilde{\bm z}$, Eq.~\eqref{eq:C3_prod} can have the scaling of $1/E^\nu$ with integer $\nu\in [M,3M]$. As shown in Table.~\ref{tab:S3_table}, there are only three allowed assignments for each element of $\bm s, \bm s^\prime, \bm r, \bm r^\prime$, thus for $j$th mode, the probability of $C^{(j)}\sim 1/E^{\nu_j}$ where $\nu_j\in \{1,2,3\}$ is
\begin{align}
    p(C^{(j)}\sim 1/E^{\nu_j}) = \begin{cases}
    1/{3^{L-1}}, &\textit{if $\nu_j = 1$}\\
    2^L/{3^{L-1}}, &\textit{if $\nu_j = 2$}\\
    1-(1+2^L)/3^{L-1}, &\textit{otherwise}.
    \end{cases}
\end{align}
Suppose the number of $\nu_j$ in $\{\nu_j\}_{j=1}^M$ being $1, 2$ is $n_1, n_2$, then the probability of $\sum_{j=1}^M \nu_j\le \nu_c$ is
\begin{align}
    p\left(\sum_{j=1}^M \nu_j \le \nu_c\right) = \sum_{\substack{n_1,n_2\ge 0,\\ n_1+n_2\le M,\\ n_1+2n_2+3(M-n_1-n_2)\le\nu_c}}^M \frac{M!}{n_1!n_2!(M-n_1-n_2)!} \left(\frac{1}{3^{L-1}}\right)^{n_1}\left(\frac{2^L}{3^{L-1}}\right)^{n_2} \left(\frac{3^{L-1}-2^L-1}{3^{L-1}}\right)^{M-n_1-n_2}
    \label{eq:C3_nonhigher_pr}
\end{align}
Setting $\nu_c = 2M$ allows us to determine the probability that $C_3^{\rm Prod}$ is {\it not} of a higher order than $C_2^{\rm Prod}$. Analytical calculation is hard, and we show a numerical calculation result in Fig.~\ref{fig:C3_modes_pr}. The exponential suppression of probability for non-higher order indicates that $C_3^{\rm Prod}$ can be thought as higher order terms compared to $C_1^{\rm Prod}$ and $C_2^{\rm Prod}$, as we did in single mode case.

To the end of product state section, we discuss the criterion for shallow and deep circuits. Recall the probability that $C_2^{\rm Prod}\sim 1/E^{2M}$ is $p\sim 1-1/2^L$. For shallow circuits, the leading order is $1/E^M$ under the constraint $\frac{1/4E^{2M}}{(3/4)^{ML}/3E^M} \in \mathcal{O}(1)$ and $\frac{(1-p)/4E^{M+1}}{(3/4)^{ML}/3E^M}\in \mathcal{O}(1)$, resulting in the condition $L \in \mathcal{O}(\log E)$.
On the other hand, for deep circuits the variance of the gradient is in $1/E^{2M}$ under the condition $\frac{1/4E^{2M}}{(3/4)^{ML}/3E^M} \in \Omega(1)$ and $\frac{(1-p)/4E^{M+1}}{1/4E^{2M}}\in \mathcal{O}(1)$ leading to $L \in \Omega(\log E)$.
\begin{figure}
    \centering
    \includegraphics[width=0.4\textwidth]{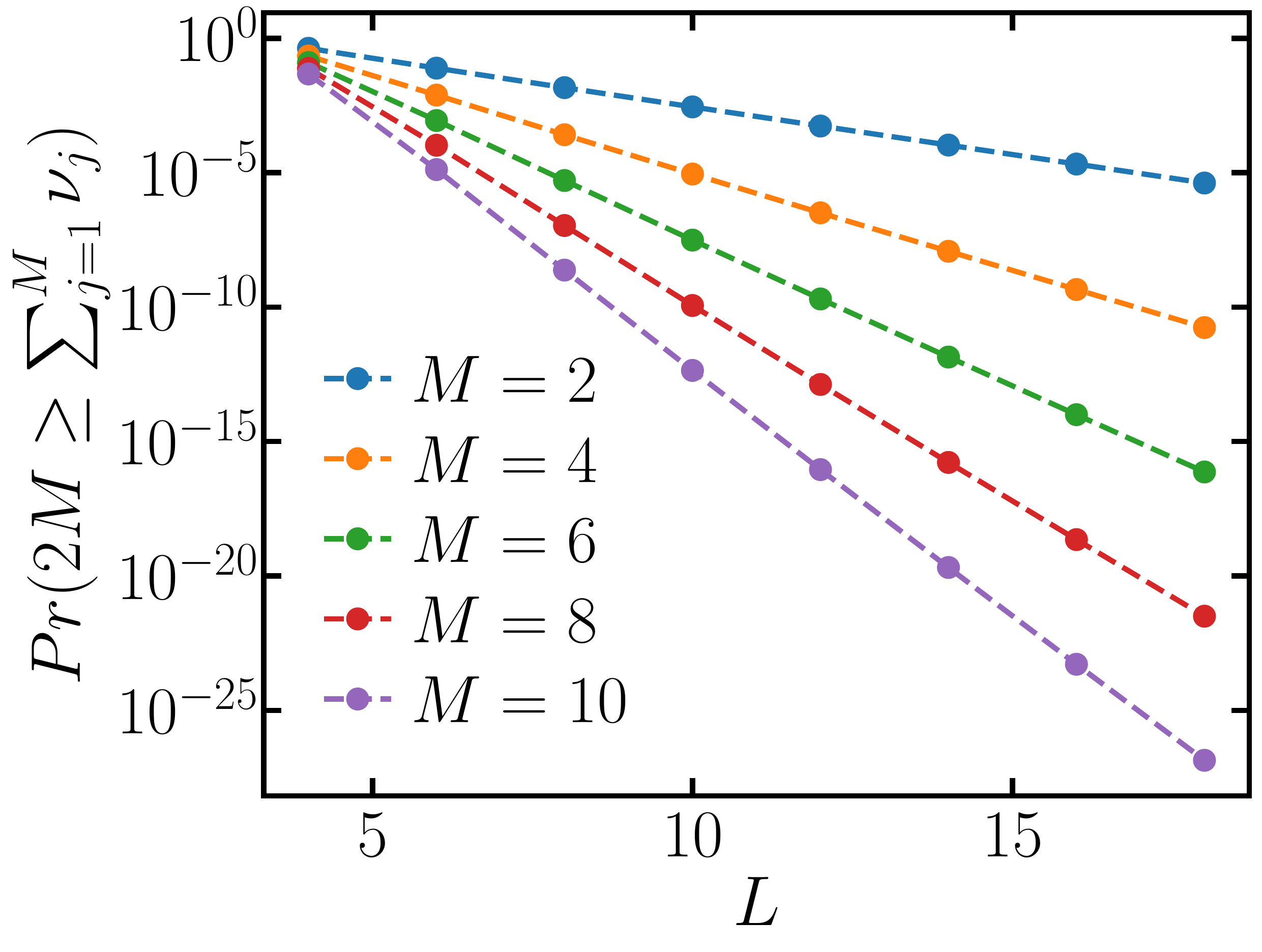}
    \caption{The probability of $\sum_{j=1}^M \nu_j\le 2M$ in Eq.~\eqref{eq:C3_nonhigher_pr} versus different circuit depth $L$ and modes number $M$.}
    \label{fig:C3_modes_pr}
\end{figure}


\subsection{Multi-mode Gaussian states}
\label{app:var_Gaussian_modes}
In this part, we present the results for the preparation of an arbitrary multi-mode Gaussian state, either entangled or not. We consider the target state to be an $M$-mode Gaussian state described by mean quadrature $\overline{\mathcal{X}}_{\bm m}$ and CM $\bV_{\bm m}$ for simplicity. As all correlators are functions of fidelity between target state and a product of coherent states, we begin with the phase space representation of a product of coherent state $\otimes_{j=1}^M \ket{\alpha_j}_{m_j}$, with quadrature mean and covariance matrix given as
\begin{align}
    \overline{\mathcal{X}} &= 2(\Re{\alpha_1},\Im{\alpha_1},\dots,\Re{\alpha_M},\Im{\alpha_M})^T \equiv 2\left(\xi_1, \xi_2, \dots, \xi_{2M-1},\xi_{2M}\right)^T = 2{\bm \xi},\\
    \bV_{\rm coh} &= \bI,
\end{align}
where we denote real and imaginary parts of all $\alpha_j$ in a unified vector $\bm \xi$ with each element in the distribution $\calN_{E/2}$. The CM $\bV_{\rm coh}$ is a $2M\times 2M$ identity matrix. Similarly, we can define $\overline{\mathcal{X}}_{\bm m}=2\bm \xi_{\bm m}$. Via applying the general fidelity formula in Eq.~\eqref{eq:Fidelity_gaussian}, we can find  that $C_1$ becomes
\begin{align}
    &C_1^{\rm Gauss} = \mathbb{E}_{\alpha\sim \calN_E^{\rm C}}\left[|\bra{\psi}\left(\otimes_{j=1}^M \ket{\alpha_j}\right)|^4\right] = \mathbb{E}_{\alpha\sim \calN_E^{\rm C}}\left[F(\ket{\psi}_{\bm m}, \otimes_{j=1}^M\ket{\alpha_j})^2\right]\nonumber\\
    &= \int d{\bm \xi} \frac{4^{M}}{\det(\bV_{\bm m}+\bI)}e^{-4({\bm \xi}-{\bm \xi_{\bm m}})^T(\bV_{\bm m}+\bI)^{-1}({\bm \xi}-{\bm \xi_{\bm m}})}\frac{1}{(\pi E)^M}e^{-\bm \xi^T \bm \xi/E}\nonumber\\
    &= \frac{4^{M}}{(\pi E)^M\det(\bV_{\bm m}+\bI)}\int d{\bm \xi}\exp\left[-{\bm \xi}^T\left(4(\bV_{\bm m}+\bI)^{-1}+\frac{\bI}{E}\right){\bm \xi}+2{\bm \xi_{\bm m}}^T4(\bV_{\bm m}+\bI)^{-1}{\bm \xi}-4{\bm \xi_{\bm m}}^T(\bV_{\bm m}+\bI)^{-1}{\bm \xi_{\bm m}}\right]\nonumber\\
    &= \frac{4^{M}\det(\bK)e^{-4{\bm \xi_{\bm m}}^T \bK{\bm \xi_{\bm m}}}}{(\pi E)^M}\int d{\bm \xi}\exp\left[-{\bm \xi}^T(4\bK+\bI/E){\bm \xi}+2{\bm \xi}_{\bm m}^{\prime T}{\bm \xi}\right]\\
    &= \frac{4^{M}e^{-4{\bm \xi_{\bm m}}^T \bK{\bm \xi_{\bm m}}}}{(\pi E)^M \det(\bK)^{-1}} \int d{\bm \xi}\exp\left[-({\bm \xi}-(4\bK+\bI/E)^{-1}{\bm \xi}_{\bm m}^\prime)^T (4\bK+\bI/E)({\bm \xi}-(4\bK+\bI/E)^{-1}{\bm \xi}_{\bm m}^\prime) + {\bm \xi}_{\bm m}^{\prime T} (4\bK+\bI/E)^{-1} {\bm \xi}_{\bm m}^\prime\right],\label{eq:C1_gauss_modes_line1}
\end{align}
where in the second to last line we denote $\bK=(\bV_{\bm m}+\bI)^{-1}$ and ${\bm \xi}_{\bm m}^\prime = 4\bK{\bm \xi}_{\bm m}$ for convenience. In the last line, we write the exponent to complete the square.
As $4\bK+\bI/E$ is a symmetric real matrix, we can diagonalize it via an orthogonal matrix $Q$ as $4\bK+\bI/E=Q^T \Lambda Q$, with $\Lambda = {\rm diag}(\lambda_1, \dots, \lambda_{2M})$ to be a diagonal matrix. We can do a variable transformation $\tilde{\bm \xi} = Q({\bm \xi}-\bK^{-1}{\bm \xi}_{\bm m}^\prime)$, then the integrand is reduced to
\begin{align}
\int d\tilde{\bm \xi} \left\lvert\frac{\partial{\bm \xi}}{\partial \tilde{\bm \xi}}\right\rvert\exp\left[-\tilde{\bm \xi}^T\Lambda \tilde{\bm \xi}\right] = \int d\tilde{\bm \xi} \exp\left[-\sum_{i=1}^{2M} \lambda_i \tilde{\xi}_i^2\right] = \frac{\pi^{M}}{\sqrt{\det(4\bK+\bI/E)}},
\end{align}
and thus the correlator $C_1^{\rm Gauss}$ becomes
\begin{align}
    C_1^{\rm Gauss}
    = \frac{4^{M}\det(\bK)\exp\left\{-4{\bm \xi}_{\bm m}^T \left[\bK-4\bK(4\bK+\bI/E)^{-1}\bK\right]{\bm \xi}_{\bm m}\right\}}{\sqrt{\det(4\bK+\bI/E)}E^M}.
    \label{eq:C1_gaussian_modes}
\end{align}

In the evaluation of $C_2$, there are both $\otimes_{j=1}^M \ket{\alpha_{z_j}\pm \alpha_{1-z_j}}$, which can also be characterized by $\overline{\mathcal{X}}_\pm = 2(\bm \xi_{\bm z}\pm {\bm \xi}_{\bm 1-\bm z})$ and $V_{\pm} = \bI$.
The distribution of $\bm \xi_{\bm z}$ is
$p(\bm \xi_{\bm z}) = \exp[-\bm \xi_{\bm z}^T \bS_{\bm z}\bm \xi_{\bm z}]\sqrt{\det \bS_{\bm z}}/\pi^M$ with $\bS_{\bm z} = \oplus_{j=1}^M \bI_2/(z_jE)$, where $\bI_2$ is a $2\times 2$ identity matrix.
The correlator $C_2$ becomes
\begin{align}
    &C_2^{\rm Gauss}(\bm z) = \mathbb{E}_{\alpha_y\sim \calN_{yE}^{\rm C}}\left[\prod_{h=0}^1 |\bra{\psi}\left(\otimes_{j=1}^M\ket{\alpha_{z_j}+(-1)^h\alpha_{1-z_j}}\right)|^2\right] = \mathbb{E}_{\alpha_y\sim \calN_{yE}^{\rm C}}\left[\prod_{h=0}^1 F(\ket{\psi}_{\bm m}\left(\otimes_{j=1}^M\ket{\alpha_{z_j}+(-1)^h\alpha_{1-z_j}}\right)\right]\nonumber\\
    &= \int d{\bm \xi_{\bm z}}d{\bm \xi_{\bm 1-\bm z}} \left(\frac{2^{M}}{\sqrt{\det(\bV_{\bm m}+\bI)}}e^{-2({\bm \xi_{\bm z}}+{\bm \xi_{\bm 1-\bm z}}-{\bm \xi}_{\bm m})^T(\bV_{\bm m}+\bI)^{-1}({\bm \xi_{\bm z}}+{\bm \xi_{\bm 1-\bm z}}-{\bm \xi}_{\bm m})}
    \right.
    \nonumber
    \\
    &\quad \quad \quad \quad \quad \quad \quad\times \frac{2^{M}}{\sqrt{\det(\bV_{\bm m}+\bI)}}e^{-2({\bm \xi_{\bm z}}-{\bm \xi_{\bm 1-\bm z}}-{\bm \xi}_{\bm m})^T(\bV_{\bm m}+\bI)^{-1}({\bm \xi_{\bm z}}-{\bm \xi_{\bm 1-\bm z}}-{\bm \xi}_{\bm m})}\nonumber\\
    &\quad \quad \quad \quad \quad \quad \quad \left.\times\frac{\sqrt{\det \bS_{\bm z}}}{\pi^M }e^{-{\bm \xi_{\bm z}}^T\bS_{\bm z}\bm \xi_{\bm z}}\frac{\sqrt{\det \bS_{\bm 1-\bm z}}}{\pi^M }e^{-{\bm \xi_{\bm 1-\bm z}}^T\bS_{\bm 1-\bm z}\bm \xi_{\bm 1-\bm z}}\right)\label{eq:C2_gaussian_modes_line0}\\
    &= \frac{4^M\det(\bK)\sqrt{\det \bS_{\bm z}\det \bS_{\bm 1-\bm z}}e^{-4{\bm \xi}_{\bm m}^T \bK{\bm \xi}_{\bm m}}}{\pi^{2M}}\int d{\bm \xi_{\bm z}}d{\bm \xi_{\bm 1-\bm z}}e^{-{\bm \xi}_{\bm z}^T(4\bK+\bS_{\bm z}){\bm \xi}_{\bm z}+2{\bm \xi}_{\bm m}^{\prime T}{\bm \xi}_{\bm z}-{\bm \xi}_{\bm 1-\bm z}^T(4\bK+\bS_{\bm 1-\bm z}){\bm \xi}_{\bm 1-\bm z}} \label{eq:C2_gaussian_modes_line1}\\
    &= \frac{4^M\det(\bK)\sqrt{\det \bS_{\bm z}\det \bS_{\bm 1-\bm z}}e^{-4{\bm \xi}_{\bm m}^T \bK{\bm \xi}_{\bm m}+{\bm \xi}_{\bm m}^{\prime T}(4\bK+\bS_{\bm z})^{-1}{\bm \xi}_{\bm m}^{\prime T}}}{\pi^{2M}\sqrt{\det \bS_{\bm z}\det \bS_{\bm 1-\bm z}}}\int d\tilde{\bm \xi_{\bm z}}d \tilde{\bm \xi}_{\bm 1-\bm z}\exp\left[-\tilde{\bm \xi}_{\bm z}^T\Lambda_{\bm z}\bm \tilde{\bm \xi}_{\bm z}\right] \exp\left[-\tilde{\bm \xi}_{\bm 1-\bm z}^T\Lambda_{\bm 1-\bm z}\bm \tilde{\bm \xi}_{\bm 1-\bm z}\right] \label{eq:C2_gaussian_modes_line2}\\
    &= \frac{4^M\det(\bK)\sqrt{\det \bS_{\bm z}\det \bS_{\bm 1-\bm z}}e^{-4{\bm \xi}_{\bm m}^T \bK{\bm \xi}_{\bm m}+{\bm \xi}_{\bm m}^{\prime T}(4\bK+\bS_{\bm z})^{-1}{\bm \xi}_{\bm m}^{\prime T}}}{\pi^{2M}\sqrt{\det \bS_{\bm z}\det \bS_{\bm 1-\bm z}}}\frac{\pi^M}{\sqrt{\det(4\bK+\bS_{\bm z})}}\frac{\pi^M}{\sqrt{\det(4\bK+\bS_{\bm 1-\bm z})}}\nonumber \\
    &= \frac{4^M \det(\bK)\exp\left\{-4{\bm \xi}_{\bm m}^T\left[\bK-4\bK(4\bK+\bS_{\bm z})^{-1}\bK\right]{\bm \xi}_{\bm m}\right\}}{\sqrt{\det(4\bK+\bS_{\bm z})\det (4\bK+\bS_{\bm 1-\bm z})}}\frac{1}{\left[\prod_{j=1}^M z_j(1-z_j)\right]E^{2M}}
    \label{eq:C2_gaussian_modes}
\end{align}
where in Eq.~\eqref{eq:C2_gaussian_modes_line1} we denote $\bK$ and $\bm \xi_{\bm m}^\prime$ in the same way as Eq.~\eqref{eq:C1_gauss_modes_line1}. In Eq.~\eqref{eq:C2_gaussian_modes_line2}, we apply the same diagonalization method we use in the derivation of Eq.~\eqref{eq:C1_gaussian_modes}, where $\Lambda_{\bm z} = Q_{\bm z}(4\bK+\bS_{\bm z})Q_{\bm z}^T$ and so as $\Lambda_{\bm 1-\bm z}$, 
$\tilde{\bm \xi}_{\bm z} = Q_{\bm z}(\bm \xi_{\bm z} - (4\bK+\bS_{\bm z})^{-1}\bm \xi_{\bm m}^\prime)$ and $\tilde{\bm \xi}_{\bm 1-\bm z} = Q_{\bm 1-\bm z}\bm \xi_{\bm 1-\bm z}$. Similar to the discussion of $C_2^{\rm Prod}$, if there are $M-N_{\bm z}$ elements in $\bm z$ that are either zero or one, then $C_2^{\rm Gauss}(\bm z)\sim 1/E^{M+N_{\bm z}}$ due to the fact that a Gaussian distribution with a zero variance in Eq.~\eqref{eq:C2_gaussian_modes_line0} becomes a Dirac-delta function.

The last correlator left is $C_3^{\rm Gauss}(\bm z_1, \bm z_2)$ 
\begin{align}
    &C_3^{\rm Gauss}(\bm z,\tilde{\bm z}) = \mathbb{E}_{\alpha_y\sim \calN_{yE}^{\rm C}}\left[\prod_{h=0}^1\left\lvert\bra{\psi}\left(\otimes_{j=1}^M\ket{\alpha_{z_j}+(-1)^h\alpha_{\tilde{z}_j}+(-1)^h\alpha_{1-z_j-\tilde{z}_j}}\right)\right\rvert \left\lvert\left(\otimes_{j=1}^M \bra{\alpha_{z_j}+(-1)^h\alpha_{\tilde{z}_j}-(-1)^h\alpha_{1-z_j-\tilde{z}_j}}\right) \ket{\psi}\right\rvert \right]\nonumber\\
    &= 4^M\det(\bK) \int d\bm \xi_{\bm z}d\bm \xi_{\tilde{\bm z}}d\bm \xi_{\bm 1-\bm z-\tilde{\bm z}} \left( e^{-(\bm \xi_{\bm z}+\bm \xi_{\tilde{\bm z}}+\bm \xi_{\bm 1-\bm z-\tilde{\bm z}}-\bm \xi_{\bm m})^T \bK(\bm \xi_{\bm z}+\bm \xi_{\tilde{\bm z}}+\bm \xi_{\bm 1-\bm z-\tilde{\bm z}}-\bm \xi_{\bm m})} e^{-(\bm \xi_{\bm z}+\bm \xi_{\tilde{\bm z}}-\bm \xi_{\bm 1-\bm z-\tilde{\bm z}}-\bm \xi_{\bm m})^T \bK(\bm \xi_{\bm z}+\bm \xi_{\tilde{\bm z}}-\bm \xi_{\bm 1-\bm z-\tilde{\bm z}}-\bm \xi_{\bm m})}\right.\nonumber\\
    &\quad \quad \quad \quad \quad \quad \quad  \quad \quad \quad \quad  \quad \times e^{-(\bm \xi_{\bm z}-\bm \xi_{\tilde{\bm z}}-\bm \xi_{\bm 1-\bm z-\tilde{\bm z}}-\bm \xi_{\bm m})^T \bK(\bm \xi_{\bm z}-\bm \xi_{\tilde{\bm z}}-\bm \xi_{\bm 1-\bm z-\tilde{\bm z}})} e^{-(\bm \xi_{\bm z}-\bm \xi_{\tilde{\bm z}}+\bm \xi_{\bm 1-\bm z-\tilde{\bm z}}-\bm \xi_{\bm m})^T \bK(\bm \xi_{\bm z}-\bm \xi_{\tilde{\bm z}}+\bm \xi_{\bm 1-\bm z-\tilde{\bm z}}-\bm \xi_{\bm m})}\nonumber\\
    &\quad \quad \quad \quad \quad \quad \quad \quad \quad  \quad \quad  \quad \left.\times\frac{\sqrt{\det \bS_{\bm z}}}{\pi^M}e^{-{\bm \xi_{\bm z}}^T\bS_{\bm z}\bm \xi_{\bm z}} \frac{\sqrt{\det \bS_{\tilde{\bm z}}}}{\pi^M}e^{-{\bm \xi_{\tilde{\bm z}}}^T\bS_{\tilde{\bm z}}\bm \xi_{\tilde{\bm z}}} \frac{\sqrt{\det \bS_{\bm 1-\bm z-\tilde{\bm z}}}}{\pi^M}e^{-{\bm \xi_{\bm 1-\bm z-\tilde{\bm z}}}^T\bS_{\bm 1-\bm z-\tilde{\bm z}}\bm \xi_{\bm 1-\bm z-\tilde{\bm z}}}\right)\nonumber\\
    &= \frac{4^M\sqrt{\det \bS_{\bm z}\det \bS_{\tilde{\bm z}}\det \bS_{\bm 1-\bm z-\tilde{\bm z}}}e^{-4\bm \xi_{\bm m}^T \bK \bm \xi_{\bm m}}}{\pi^{3M}\det(\bK)^{-1}}
    \int d\bm \xi_{\bm z}d\bm \xi_{\tilde{\bm z}}d\bm \xi_{\bm 1-\bm z-\tilde{\bm z}} e^{-\bm \xi_{\bm z}^T (4\bK+\bS_{\bm z})\bm \xi_{\bm z}+2{\bm \xi}_{\bm m}^{\prime T}\bm \xi_{\bm z}-\bm \xi_{\tilde{\bm z}}^T (4\bK+\bS_{\tilde{\bm z}})\bm \xi_{\tilde{\bm z}}-\bm \xi_{\bm 1-\bm z-\tilde{\bm z}}^T (4\bK+\bS_{\bm 1-\bm z-\tilde{\bm z}})\bm \xi_{\bm 1-\bm z-\tilde{\bm z}}}\nonumber\\
    &= \frac{4^M
   \sqrt{\det \bS_{\bm z}\det \bS_{\tilde{\bm z}}\det \bS_{\bm 1-\bm z-\tilde{\bm z}}}e^{-4\bm \xi_{\bm m}^T \bK \bm \xi_{\bm m}+{\bm \xi}_{\bm m}^{\prime T}(4\bK+\bS_{\bm z})^{-1}{\bm \xi}_{\bm m}^{\prime T}} }{\pi^{3M} \det(\bK)^{-1}}
    \int d \tilde{\bm \xi}_{\bm z} d \tilde{\bm \xi}_{\tilde{\bm z}}d \tilde{\bm \xi}_{\bm 1-\bm z-\tilde{\bm z}} e^{-\tilde{\bm \xi}_{\bm z}\Lambda_{\bm z}\tilde{\bm \xi}_{\bm z}}e^{-\tilde{\bm \xi}_{\tilde{\bm z}}\Lambda_{\tilde{\bm z}}\tilde{\bm \xi}_{\tilde{\bm z}}}e^{-\tilde{\bm \xi}_{\bm 1-\bm z-\tilde{\bm z}}\Lambda_{\bm 1-\bm z-\tilde{\bm z}}\tilde{\bm \xi}_{\bm 1-\bm z-\tilde{\bm z}}}\label{eq:C3_gaussian_modes_line1}\\
    &= \frac{4^M\sqrt{\det \bS_{\bm z}\det \bS_{\tilde{\bm z}}\det \bS_{\bm 1-\bm z-\tilde{\bm z}}}e^{-4\bm \xi_{\bm m}^T \bK \bm \xi_{\bm m}+{\bm \xi}_{\bm m}^{\prime T}(4\bK+\bS_{\bm z})^{-1}{\bm \xi}_{\bm m}^{\prime T}} }{\pi^{3M}\det(\bK)^{-1}} \frac{\pi^M}{\sqrt{\det(4\bK+\bS_{\bm z})}} \frac{\pi^M}{\sqrt{\det(4\bK+\bS_{\tilde{\bm z}})}} \frac{\pi^M}{\sqrt{\det(4\bK+\bS_{\bm 1-\bm z-\tilde{\bm z}})}}\nonumber\\
    &= \frac{4^M \det(\bK)\exp\left\{-4{\bm \xi}_{\bm m}^T\left[\bK-4\bK(4\bK+\bS_{\bm z})^{-1}\bK\right]{\bm \xi}_{\bm m}\right\} }{\sqrt{\det(4\bK+\bS_{\bm z})}\sqrt{\det(4\bK+\bS_{\tilde{\bm z}})}\sqrt{\det(4\bK+\bS_{\bm 1-\bm z-\tilde{\bm z}})}}\frac{1}{\left[\prod_{j=1}^M z_j\tilde{z}_j(1-z_j-\tilde{z}_j)\right]E^{3M}},
\end{align}
where in Eq.~\eqref{eq:C3_gaussian_modes_line1} we do the same diagonalization method as in Eq.~\eqref{eq:C2_gaussian_modes_line1}, with $\Lambda_{\tilde{\bm z}} = Q_{\tilde{\bm z}}(4\bK+\bS_{\tilde{\bm z}})Q_{\tilde{\bm z}}^T$ and $\tilde{\bm \xi}_{\tilde{\bm z}} = Q_{\tilde{\bm z}}\bm \xi_{\tilde{\bm \xi}}$, and so as $\Lambda_{\bm 1-\bm z-\tilde{\bm z}},\tilde{\bm \xi}_{\bm 1-\bm z-\tilde{\bm z}}$. Following the same analysis from Eq.~\eqref{eq:C3_nonhigher_pr}, the bulk contribution of $C_3^{\rm Gauss}(\bm z,\tilde{\bm z})$ behaves as $1/E^\nu$ with $\nu > 2M$, and thus can be omitted as they are higher order in $E$ when $E$ is large.

In the asymptotic limit of $E$, we can omit the contribution of $\bI/E$ in $C_1^{\rm Gauss}$ and $\bS_{\bm z},\bS_{\bm 1-\bm z}$ in $C_2^{\rm Gauss}$ compared to $4\bK$, and thus the determinants in those correlators reduce to constants independent of $E$, which directly leads to the scaling of $1/E^{M}$ and $1/E^{2M}$ separately. Therefore, the critical depth between shallow and deep circuits in preparation of a general $M$-mode Gaussian state is also determined by the logarithm of circuit ensemble energy. 

In the following, we present an explicit example of two-mode squeezed vacuum (TMSV) state. The CM has been introduced in Appendix~\ref{app:Gaussian_intro}. Through the calculation of Eqs.~\eqref{eq:C1_gaussian_modes} and~\eqref{eq:C2_gaussian_modes}, one can directly find the $C_1$ and $C_2$ for TMSV state as
\begin{align}
    C_1^{\rm TMSV} &= \frac{\sech^4(\zeta)}{G_1(E)},\label{eq:C1_tmsv}\\
    C_2^{\rm TMSV}(\bm z) &= \frac{\sech^4(\zeta)}{G_2(\bm zE)G_2[(\bm 1-\bm z)E]},\label{eq:C2_tmsv}
\end{align}
where we have defined
\begin{align}
    G_2(\bm z) = 1+2\|\bm z\|_1 + 4\sech^2(\zeta)\prod_j z_j. \label{eq:G2_def} 
\end{align}
Here $\bm x$ is a vector and $\bm 1$ above is an identity vector with same length as $\bm x$.  In the asymptotic region of large $E$, we also have $C_2^{\rm TMSV}(\bm z) \sim 1/64z_1z_2(1-z_1)(1-z_2)E^4$. Note that both correlators monotonically decrease with $E$.


To summarize, for shallow circuits $L\in \mathcal{O}(\log E)$, the variance of the gradient is ${\rm Var}\sim 1/E^M$; while for deep ones $L\in \Omega(\log E)$, it becomes $\sim 1/E^{2M}$ in the asymptotic limit of $E$.

\section{Results on coherent and single-mode squeezed vacuum states}
\label{app:one_mode_gaussian}
In this section, we provide more results on the state preparation of a single-mode Gaussian state. In Fig.~\ref{fig:var_gaussian_app}(a),(b), we show the variance of the gradient ${\rm Var}\left[\partial_{\theta_k}\calC\right]$ versus ensemble energy $E$ in the preparation of a coherent $\ket{\gamma}$ and an SMSV state $\ket{\zeta}_{\rm SMSV}$ with $E_t = |\gamma|^2=\sinh^2(r)=8$, which are presented in the same way as Fig.~\ref{fig:var_states}. For shallow circuits, the numerical results (orange dots) agree with Theorem~\ref{res:shallow}, where $C_1$ is chosen to be Eqs.~\eqref{eq:C1_coh},~\eqref{eq:C1_smsv} separately. For deep circuits, the bounds are evaluated with Eq.~\eqref{eq:var_deep_gaussian}, where $C_2$ is shown in Eqs.~\eqref{eq:C2_coh},~\eqref{eq:C2_smsv}. We see similar behavior to the main text examples and same conclusions also hold here.

In Fig.~\ref{fig:training_SMSV}, we show another training example in preparation of an SMSV state with different values of the initial ensemble energy. With the increase of ensemble energy, the training performance becomes worse, as we already see in Eq.~\eqref{eq:var_deep_gaussian}: $C_2^{\rm SMSV}(z)$ in Eq.~\eqref{eq:C2_smsv} monotonically decreases, a larger ensemble energy will prevent the efficient training of circuits.

\begin{figure}
    \centering
    \includegraphics[width=1\textwidth]{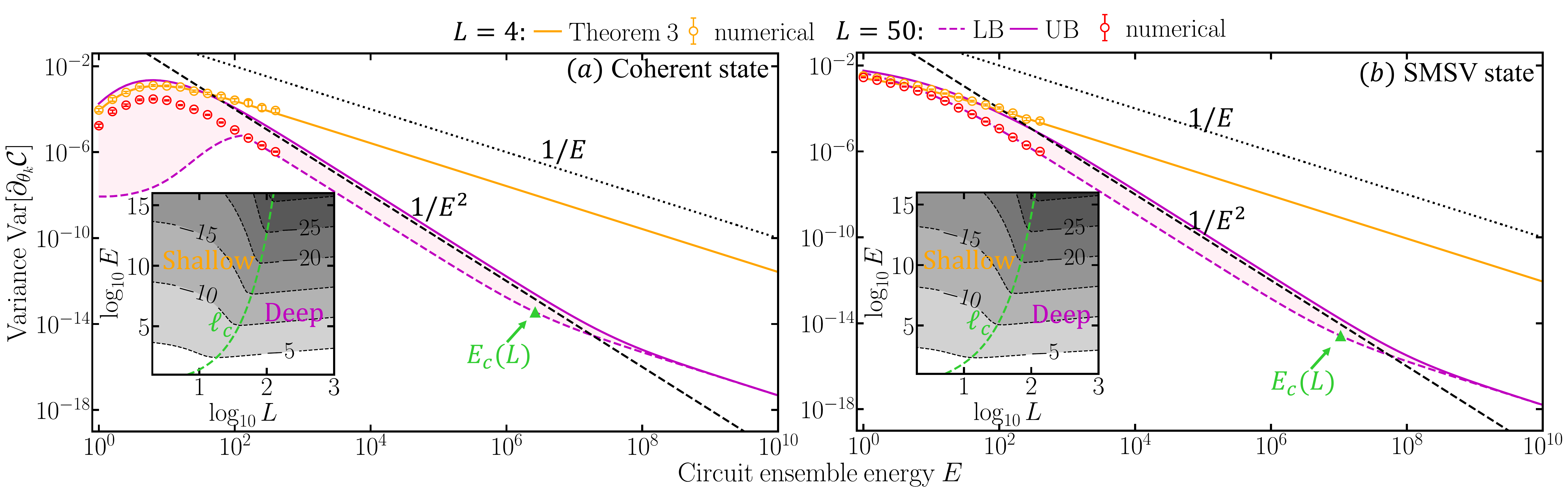}
    \caption{Variance of gradient ${\rm Var}[\partial_{\theta_{k}}\calC]$ at $k=L/2$ in preparation of (a) coherent states $\ket{\gamma}$ and SMSV state $\ket{\zeta}_{\rm SMSV}$ with $E_t=8$. All legends share the same meaning as Fig.~\ref{fig:var_states}.
    \label{fig:var_gaussian_app}}
    \centering
    \includegraphics[width=0.45\textwidth]{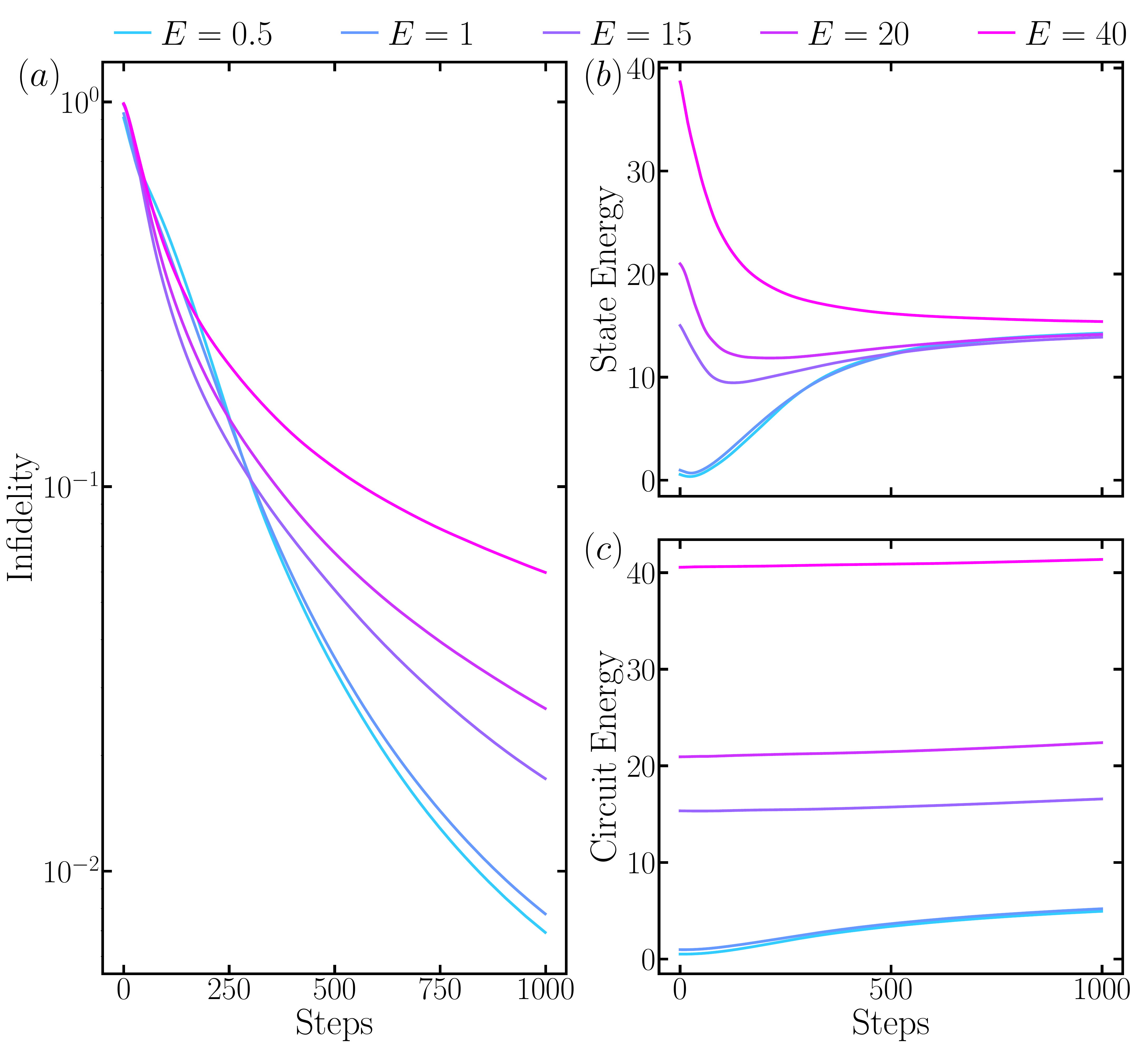}
    \caption{Training for the SMSV state $\ket{\psi} = \ket{\zeta}_{\rm SMSV}$ with $E_t=\sinh^2\zeta=15$, utilizing a $L=50$ CV VQC initialized with different ensemble energy $E$. We show (a) average infidelity of the output state with the target state, (b) average output state energy and (c) average circuit energy $\sum_{j=1}^L |\beta_j|^2$ versus training steps.
    \label{fig:training_SMSV}}
\end{figure}

\end{widetext}

\end{document}